%% file: main.tex
\documentclass{article}
\usepackage[utf8]{inputenc} %
\usepackage[T1]{fontenc}    %
\usepackage{graphicx} %
\usepackage{tikz}
\usepackage{array}
\usepackage{mdframed}
\usepackage[frozencache,cachedir=minted-cache]{minted}

\input{preamble.tex}

\newif\ificlrfinal
\iclrfinaltrue

\hypersetup{linktocpage = true, linkcolor = purple}

\newcommand{\Comment}[1]{ \color{teal} #1 \color{black} }

\title{Online Tokenization with Optimal Delay in Linear Time}

\author{
Mina Dalirrooyfard \\
Morgan Stanley\\
\texttt{minad@mit.edu}
\and
Christopher Ye\thanks{Supported by NSF HDR TRIPODS Phase II grant 2217058 (EnCORE Institute). This work was completed as part of an internship at Morgan Stanley Machine Learning Research.}  \\
University of California San Diego  \\
\texttt{czye@ucsd.edu}
}

\begin{document}

\maketitle

\thispagestyle{empty}

\begin{abstract} 
    \input{abstract}
\end{abstract}

\thispagestyle{empty}

\newpage

\clearpage
\pagenumbering{arabic}

\input{intro}

\input{prelims}

\input{online_tokenization}

\input{improved_delay}

\input{experiments}

\section*{Statement on AI Usage}

For our theoretical results, generative AI was particularly useful in designing the reduction behind \Cref{thm:np-delay-lb}.
For our empirical results, essentially all of our code was implemented by AI.
The code was reviewed and manually verified.
For our DFA algorithm, the main empirical contribution of this work, the code was written by AI but we manually prove the correctness of the key optimization (DFA simplification, see \Cref{sec:experimental-setup}). 
We take responsibility for the final content of this work, including text, claims or artifacts produced with the aid of generative AI.

\printbibliography

\appendix
\input{appendix_omitted_proofs}
\end{document}

%% file: preamble.tex
\usepackage{array}
\usepackage{tabularx}
\usepackage{graphicx}
\usepackage{amsmath,amssymb,amsthm,mathtools}
\usepackage{bm}
\usepackage{bbm}
\usepackage{xspace}
\usepackage{url}
\usepackage{fullpage, prettyref}
\usepackage{boxedminipage}
\usepackage{wrapfig}
\usepackage{ifthen}
\usepackage{color}
\usepackage{xcolor}
\usepackage{framed}
\usepackage[colorlinks=true,urlcolor=blue,linkcolor=blue,citecolor=blue,pdfstartview=FitH]{hyperref}
\usepackage{subcaption}
\usepackage[nameinlink]{cleveref}
\usepackage{esvect}
\usepackage{mdframed}
\usepackage{ifthen}
\usepackage{tikz}
\usetikzlibrary{automata,positioning,arrows.meta}
\newcommand{\pL}{\backslash\mathrm p\{L\}}
\newcommand{\pN}{\backslash\mathrm p\{N\}}
\newcommand{\leadc}{[\hat{}\,\backslash\mathrm r\backslash\mathrm n\pL\pN]}
\newcommand{\punct}{[\hat{}\,\backslash\mathrm s\pL\pN]}
\newcommand{\Hclass}{[\hat{}\,\backslash\mathrm S\backslash\mathrm r\backslash\mathrm n]}
\newcommand{\Rclass}{[\backslash\mathrm r\backslash\mathrm n]}

\usepackage[
  backend=biber,
  style=alphabetic,
  natbib=true,
  backref=true
]{biblatex}
\usepackage{enumitem}
\usepackage{thmtools}
\usepackage{thm-restate}

\newtheorem{theorem}{Theorem}[section]

\newtheorem{lemma}[theorem]{Lemma}
\newtheorem{claim}[theorem]{Claim}

\newtheorem{definition}[theorem]{Definition}
\newtheorem{proposition}[theorem]{Proposition}

\newtheorem{remark}[theorem]{Remark}

\newcommand{\ignore}[1]{}

\newcommand{\poly}{\mathrm{poly}}

\newcommand{\calE}{{\cal E}}
\newcommand{\calL}{{\cal L}}
\newcommand{\calF}{{\cal F}}

\newcommand{\calP}{{\cal P}}
\newcommand{\calQ}{{\cal Q}}

\newcommand{\calS}{{\cal S}}
\newcommand{\calT}{{\cal T}}

\newcommand{\calC}{{\cal C}}

\newcommand{\calA}{{\cal A}}

\newcommand{\calR}{{\cal R}}

\newcommand{\set}[1]{\left\{ #1 \right\}}
\newcommand{\ceil}[1]{\lceil#1\rceil}

\newcommand{\given}{\textrm{ s.t. }}

\newcommand{\otherwise}{\textrm{ o/w }}
\renewcommand{\leq}{\leqslant}
\renewcommand{\geq}{\geqslant}

\newcommand{\p}{\boldsymbol{p}}

\newcommand{\Sec}[1]{\hyperref[sec:#1]{\Cref*{sec:#1}}} %
\newcommand{\Eqn}[1]{\hyperref[eq:#1]{(\ref*{eq:#1})}} %
\newcommand{\Fig}[1]{\hyperref[fig:#1]{Fig.\,\ref*{fig:#1}}} %
\newcommand{\Tab}[1]{\hyperref[tab:#1]{Tab.\,\ref*{tab:#1}}} %
\newcommand{\Thm}[1]{\hyperref[thm:#1]{Theorem\,\ref*{thm:#1}}} %
\newcommand{\Fact}[1]{\hyperref[fact:#1]{Fact\,\ref*{fact:#1}}} %
\newcommand{\Lem}[1]{\hyperref[lem:#1]{Lemma\,\ref*{lem:#1}}} %
\newcommand{\Prop}[1]{\hyperref[prop:#1]{Prop.~\ref*{prop:#1}}} %
\newcommand{\Cor}[1]{\hyperref[cor:#1]{Corollary~\ref*{cor:#1}}} %
\newcommand{\Conj}[1]{\hyperref[conj:#1]{Conjecture~\ref*{conj:#1}}} %
\newcommand{\Def}[1]{\hyperref[def:#1]{Definition~\ref*{def:#1}}} %
\newcommand{\Alg}[1]{\hyperref[alg:#1]{Alg.~\ref*{alg:#1}}} %
\newcommand{\Obs}[1]{\hyperref[obs:#1]{Obs.~\ref*{obs:#1}}} %
\newcommand{\Ex}[1]{\hyperref[ex:#1]{Ex.~\ref*{ex:#1}}} %
\newcommand{\Clm}[1]{\hyperref[clm:#1]{Claim~\ref*{clm:#1}}} %
\newcommand{\Step}[1]{\hyperref[step:#1]{Step~\ref*{step:#1}}} %

\newcommand{\tO}[1]{\tilde{O}(#1)}

\crefname{algocf}{alg.}{algs.}
\Crefname{algocf}{Alg.}{Algs.}

\crefname{step}{Step}{Steps}
\Crefname{step}{Step}{Steps}
\newcommand{\stepnames}{%
  \crefalias{enumi}{step}\crefalias{enumii}{step}\crefalias{enumiii}{step}}

\renewcommand{\P}{\mathbf{P}}
\newcommand{\NP}{\mathbf{NP}}
\newcommand{\SAT}{\mathrm{SAT}}
\newcommand{\kSAT}[1]{#1\text{-}\SAT}

\newcommand{\Tok}{\operatorname{Tok}}
\newcommand{\Seg}{\operatorname{Seg}}
\newcommand{\Split}{\operatorname{Split}}
\newcommand{\Merge}{\operatorname{Merge}}

\newcommand{\Segment}{\python{Segment}}
\newcommand{\Pretokenize}{\python{Pretokenize}}
\newcommand{\BPEMerge}{\python{BPEMerge}}
\newcommand{\FastMatch}{\python{FastMatch}}
\newcommand{\applyFixup}{\python{ApplyFixup}}
\newcommand{\FastPretokenization}{\python{FastPretokenization}}

\newcommand{\last}{\mathrm{last}}
\newcommand{\state}{\mathbf{state}}

\newcommand{\buffer}{\mathrm{buffer}}
\newcommand{\start}{\mathrm{start}}
\newcommand{\accept}{\calP}
\newcommand{\dead}{\mathrm{dead}}
\newcommand{\fail}{\mathrm{fail}}
\newcommand{\out}{\mathrm{out}}
\newcommand{\base}{\mathrm{base}}
\newcommand{\trim}{\mathrm{trim}}
\newcommand{\cut}{\mathrm{cut}}
\newcommand{\brk}{\;\|\;}

\newcommand{\push}{{\bf push}}
\newcommand{\pop}{{\bf pop}}

\newcommand{\query}{{\bf query}}
\newcommand{\cqueue}{{\rm cQueue}}
\newcommand{\bpequeue}{{\rm bpeQ}}
\newcommand{\incbpe}{{\rm incBPE}}

\newcommand{\cutQ}{{\rm cutQ}}

\newcommand{\addgets}{\overset{+}{\gets}}

\newcommand{\raw}{{\rm raw}}
\newcommand{\pre}{{\rm pre}}
\newcommand{\suc}{{\rm suc}}

\newcommand{\lookback}{\mathfrak{b}}
\newcommand{\longest}{\mathfrak{t}}
\newcommand{\merges}{\mathfrak{m}}

\newcommand{\rust}[1]{%
  \relax
  \ifmmode
    \mathord{\text{\mintinline{rust}{#1}}}
  \else
    \mintinline{rust}{#1}%
  \fi
}

\newcommand{\python}[1]{%
  \relax
  \ifmmode
    \mathord{\text{\mintinline{python}{#1}}}
  \else
    \mintinline{python}{#1}%
  \fi
}

\newcommand{\rfk}{{\rm R50K}}
\newcommand{\pfk}{{\rm P50K}}
\newcommand{\clok}{{\rm CL100K}}
\newcommand{\otk}{{\rm O200K}}
\newcommand{\refT}{\mathrm{ref}}

%% file: abstract.tex
Tokenization is an essential preprocessing step in almost all large language models (LLMs). One of the most popular frameworks is BPE tokenization, which consists of special token segmentation, pretokenization, and byte-pair encoding (BPE). However, existing algorithms are sequential, with each stage only processing complete output from the previous stage. In order to reduce latency and memory footprints, recent works (Jiang and Gong, ICML 2026) (Li, Yang and Mamouras, ASPLOS 2026) (Mamouras, Li and Yang, PLDI 2026) have initiated the study of online algorithms for subroutines such as pretokenization and BPE with a goal of streaming input to output as the input arrives online. However, the search for efficient online tokenization remains open.

In this work, we formalize and resolve online tokenization. We introduce the notion of delay for online tokenization which characterizes the extent to which the algorithm output lags behind the (possibly computationally intractable) optimum. We give a near-linear time algorithm with optimal delay (up to constant factors). We complement this algorithm with a conditional lower bound: unless $\mathbf{P} = \mathbf{NP}$, there is no polynomial time algorithm that beats the delay of our algorithm by more than a constant factor.

We implement a practical variant for tokenizers from OpenAI's tiktoken and HuggingFace's tokenizers. We obtain a $7.9\times$-$34.8\times$ latency advantage (measured as time to first token) over tiktoken (as well as the recent gigatoken tokenizer) and our implementation maintains approximately constant memory usage independent of input size. Our implementation also achieves a $2.4\times$-$7.0\times$ improvement in throughput over tiktoken and $15.8\times$-$38.1\times$ over tokenizers.

%% file: intro.tex
\input{motivation}

\subsection{Implementation and Experiments}

\input{empirical}

\section{Technical Overview}

\input{iclr/iclr_overview}

\paragraph{Organization.}
In \Cref{sec:prelims}, we define the necessary preliminaries.
In \Cref{sec:online-tokenization}, we establish the hardness of online tokenization with zero delay, and obtain an efficient algorithm via composition.
In \Cref{sec:improved-delay}, we give our efficient optimal delay algorithm.
In \Cref{sec:experimental-setup}, we present our experimental setup and in \Cref{sec:experiment-results} our results.
We provide omitted proofs in \Cref{app:omitted}.

%% file: motivation.tex
\section{Introduction}

In many natural language processing (NLP) applications, including powerful and popular large language models (LLMs), raw text undergoes tokenization before being provided as model input.
Tokenization defines a vocabulary containing as many common words and sequences as possible, which help compress raw text into tokens, balancing expressivity and computational efficiency.

Byte-Pair Encoding (BPE) is a popular method for tokenization proposed by \citet{gage1994new} that has been widely adopted in LLMs \citep{sennrich2016neural}.
Indeed, BPE Tokenization is both the subject of intense recent theoretical study \citep{berglund2023formalizing, jiang2025incrementalbpe, mamouras2026efficient} and the basis for popular practical tokenizers, including OpenAI's \python{tiktoken} and HuggingFace's \python{tokenizers} \citep{tiktoken, huggingface}.

Concretely, tokenization consists of a three step process: Special Token, Pretokenization, and BPE.
As observed in \citet{jiang2025incrementalbpe}, existing implementations of tokenization are inherently offline and sequential: we take in the full input, apply segmentation to the full input, then apply pretokenization on each segment, and then apply BPE to each pretoken.
While completed segments and pretokens can be passed on to the next stage, one should in principle be able to apply pretokenization to a segment as it forms and similarly apply BPE to a pretoken as it forms. 
Following this observation, \citet{jiang2025incrementalbpe} raise the question of online tokenization: can we process each stage as the input arrives in an online setting?

While offline tokenization is now well understood both theoretically and empirically, the question of efficient online tokenization remains wide open.
Online algorithms are desirable for many reasons.
First, we can drastically reduce the the memory footprint of our algorithm by avoiding holding the entire input or output in memory, a shortcoming of all existing end-to-end tokenization algorithms.
Second, streaming can reduce latency of tokenization, drastically reducing the time to first token.
Indeed, streaming and online algorithms are a fundamental and well-studied paradigm both in theory and in practice, with efficient algorithms being developed in many applications including counting \citep{morris1978counting}, selection and sorting \citep{munro1980selection}, heavy hitters \citep{misra1982finding},  frequency moments \citep{alon1996space}, and cardinality estimation \citep{flajolet1985probabilistic, durand2003loglog, flajolet2007hyperloglog}.

Motivated by latency and memory constraints, online algorithms have been developed for individual steps of tokenization: \citet{li2026static} studies the problem of online pretokenization and \citet{berglund2023formalizing, jiang2025incrementalbpe, mamouras2026efficient} study online BPE algorithms.
However, none of these works provide an online algorithm for the full tokenization pipeline.
Indeed, \citet{jiang2025incrementalbpe} raise the design of end-to-end online tokenization as a crucial direction for future research.
Our work is thus guided by the following question:

\begin{center}
    {\it Is there an efficient algorithm for online tokenization?}
\end{center}

\subsection{Formalizing Online Tokenization}

We begin with formalizing tokenization which consists of the following three stage process:

\begin{enumerate}
    \item Special Token ($\Seg$): Given a list of special tokens $S^*$ and a string $x$, split the string by occurrences of special tokens (see \Cref{def:special-tokens}) into segments.
    For example, if $S^* = \set{\rust{"<|endoftext|>"}}$ and the input string is $x_1 = \rust{"file1<|endoftext|>file2"}$, we have
    \begin{equation*}
        \Seg(x_1; S^*) = \rust{"file1"} \brk \rust{"<|endoftext|>"} \brk \rust{"file2"} \text{.}
    \end{equation*}
    We use $\brk$ to denote segment boundaries.

    \item Pretokenization ($\Split$): Given a regular expression (regex) $R$ and a string $x$, split each segment of $x$ into pretokens by repeatedly finding the longest match of $R$ in $x$ (see \Cref{def:pretokenization}).
    For example, if $R = a^* \mid b^*$ and $x_2 = \rust{"aaabbba"} \brk \rust{"aaaaabbbbaaa"}$, 
    \begin{equation*}
        \Split(x_2; R) = \rust{"aaa"} \brk \rust{"bbb"} \brk \rust{"a"} \brk \rust{"aaaaa"} \brk \rust{"bbbb"} \brk \rust{"aaa"} \text{.}
    \end{equation*}
    As we see in our example, matches cannot extend across segment boundaries.

    \item BPE ($\Merge$): Given a merge list $M$ and a list of pretokens $x$, apply BPE to each pretoken of $x$, by repeatedly applying merges in order of priority left to right (see \Cref{def:bpe}).
    For example, if $x_3 = \rust{"aaabbbaaa"}$ and $M = \set{(a, b) \mapsto X, (a, a) \mapsto A, (b, b) \mapsto B}$,
    \begin{equation*}
        \Merge(x_3; M) = \rust{"AXBAa"} \text{.}
    \end{equation*}
\end{enumerate}

Tokenization, given special tokens $S^*$, regex $R$, and merge list $M$, is defined
\begin{equation*}
    \Tok(x; S^*, R, M) = \Merge(\Split(\Seg(x; S^*); R); M) \text{.}
\end{equation*}
Note that each token $\tau$ in the output corresponds to a unique string $x$ denoted $\raw(\tau)$ where $\Merge(x) = \tau$.
For example, using the same $R$ and $M$ as above,
\begin{equation*}
    \Tok(x_3; \emptyset, R, M) = \Merge(\rust{"aaa"} \brk \rust{"bbb"} \brk \rust{"aaa"}) = \rust{"AaBbAa"} \text{.}
\end{equation*}
That is, merges cannot merge across pretoken boundaries, and the merged pretokens are concatenated.
In the offline setting, efficient algorithms are now well known.
For example, $\Seg$ and $\Split$ have classical linear time algorithms \citep{reps1998maximalmunch}.
The naive algorithm for $\Merge$ runs in $O(n \merges)$ time,
which is already efficient when $\merges = |M|$ is small.
When $\merges$ is large, there are well known $\tO{n}$ time heap-based algorithms (see e.g. the \python{tiktoken} implementation \citep{tiktoken}).

We now define \emph{online} $\Seg$, $\Split$, $\Merge$, and $\Tok$. First, for two strings $w,w'$ in an alphabet $\Sigma$, if $w$ is a prefix of $w'$, we show it by $w \preceq w'$.
Consider a function $f: \Sigma^* \rightarrow (\Sigma')^*$ where $\Sigma \subseteq \Sigma'$ are alphabets (e.g. characters and tokens). In our context, $f$ could be any of the $\Seg$, $\Split$, $\Merge$, or $\Tok$ functions.
We say that $z\in (\Sigma')^*$ is a \emph{stable prefix} with respect to $x$ if $z \preceq f(xy)$ for any $y \in \Sigma^*$.
It is easy to show that stable prefixes are totally ordered (see \Cref{clm:stable-prefix-compare}).
The \emph{maximum stable prefix} is then denoted by $f_{\preceq}(x)$.
We define online tokenization (analogously each of the three sub-routines) as $\Tok_{\preceq}$.
In particular, given that input $x$ is all we have observed so far, online tokenization outputs $\Tok_{\preceq}(x)$, the longest possible prefix of tokens that can be safely emitted, regardless of what characters appear in the future.
This is exactly the definition that allows us to stream input characters and output tokens.

Surprisingly, our first result establishes that online tokenization, unlike its offline counterpart, is $\NP$-Hard.
Given that it is unlikely that even a polynomial (let alone linear time) algorithm for online tokenization exists, a natural relaxation is to design efficient approximation algorithms. 
We introduce a natural approximation of online tokenization which we call \emph{delay} (see \Cref{def:online-tokenization}).
An online tokenization with delay $k$ produces on input $x = x_1 \dots x_{n}$, output $z$ satisfying:
\begin{equation*}
    \Tok_{\preceq}(x_{1} \dots x_{n - k}) \preceq z \preceq \Tok_{\preceq}(x) \text{.}
\end{equation*}
In particular, an algorithm with small delay produces a stable prefix that might not be the maximum stable prefix, but would have been the maximum stable prefix not too long ago.
A trivial algorithm achieves delay $n$: output the empty string $\epsilon$.
So we focus on answering the following question:

\begin{center}
    {\it What is the minimum delay $k$ with efficient algorithms for online tokenization?}
\end{center}

\subsection{Our Contributions}
\label{sec:contributions}

We give optimal algorithms for online tokenization, tightly characterizing when efficient online tokenization is possible, resolving the main open question of \citet{jiang2025incrementalbpe}.
First, our lower bound rules out the existence of efficient algorithms with small delay.
For any list of special tokens $S^*$, let $\ell_{\infty}(S^*) = \max_{t \in S^*} |t|$ be the longest special token and $\ell_1(S^*) = \sum_{t \in S^*} |t|$ be the sum of the lengths of the special tokens.
For any merge list $M$, let $\merges = |M|$ and $\longest = \max \set{|x| \given |\Merge(x)| = 1}$ denote the length of the longest token. 

Note $\ell_{\infty}(S^*)$, $\ell_{1}(S^*)$, $|R|$ are constants independent of $x$, so we omit them from bounds for clarity.
While $\merges$ (and therefore $\longest$) is also a constant (since the BPE dictionary is independent of the input string), since $\merges$ can be quite large in practice (around 100,000 for CL100K in \python{tiktoken}), we will account for $\merges, \longest$ in space and run-time, keeping consistent with prior work \citep{jiang2025incrementalbpe}.

\begin{restatable}[Hardness of Online Tokenization]{theorem}{OnlineLB}
    \label{thm:np-delay-lb}
    Unless $\P = \NP$, there is no algorithm for exact (zero-delay) online tokenization with $\poly(\ell_1(S^*), |R|, \merges, \longest)$ preprocessing and total time $\poly(n, \ell_1(S^*), |R|, \merges, \longest)$, even when $S^* = \emptyset$.
    Furthermore, the lower bounds hold for algorithms with $o(\merges \longest)$ delay.
\end{restatable}

A natural attempt at an online tokenization algorithm is to compose online algorithms for special token, pretokenization, and BPE.
We show that this idea of composition does in fact yield an efficient online tokenization algorithm with non-trivial delay.

The performance of this algorithm, which we call composition, is parameterized by the following quantity which depends on the pretokenization regular expression.
Given a regular language $L$, we define the lookback of the regular language, denoted $\lookback(L)$, to be the smallest integer (or $\infty$ if no such number exists) such that every string $x \in L$ of length $|x| > \lookback$ has a strict prefix $y \prec x$ such that $y \in L$, and $|y| \geq |x| - \lookback$ (see \Cref{def:lookback}), i.e. $y$ is in the language and $y$ is shorter than $x$ by at most $\lookback$.
For example, if $R = a^*$, we have $\lookback(R) = 1$.
On the other hand, if $R = a^*b$, $\lookback(R) = \infty$.

\begin{restatable}[Online Tokenization via Composition]{theorem}{CompositionUB}
    \label{thm:linear-online-tokenization-w-delay}
    Let $S^* \subset \Sigma^*$ be a set of special tokens, $R$ a total regular expression over $\Sigma$, and $M$ a merge list.
    Then, there is an $O(n \log^2 \longest)$ time and $O(\lookback(R) + \merges \longest)$ space algorithm for online tokenization with $\lookback(R) + (\merges +1)\longest + \ell_{\infty}(S^*)$ delay.
\end{restatable}

\begin{remark}As part of our algorithm for \Cref{thm:linear-online-tokenization-w-delay}, we design a simple, novel linear time algorithm for online pretokenization (see \Cref{lem:pretokenize-alg}).
In contrast, the previous algorithm of \citet{li2026static} incurs preprocessing time and space $2^{\Theta(\lookback(R))}$.
In particular, we achieve online pretokenization without delay in linear time for all regular expressions under the longest-match convention, whereas the previous algorithm was only efficient when $\lookback(R) = O(\log n)$ is small.
\end{remark}
\begin{remark}
    While we have efficient algorithms for online BPE \citep{jiang2025incrementalbpe} and pretokenization via \Cref{lem:pretokenize-alg}, \Cref{thm:np-delay-lb} still establishes that online tokenization (with no delay) is $\NP$-hard.
    Our algorithm for online pretokenization has $0$ delay.
    Known algorithms for online BPE \citep{jiang2025incrementalbpe} have delay $O(\merges \longest)$ (see proof of \Cref{lem:bpe-merge-alg}) which suffices for our purposes.
    We leave as an interesting open question whether zero delay is achievable for online BPE.
\end{remark}

While our first algorithm does not achieve optimal delay, we show that its space is optimal (up to word size).
In particular, $\Theta(\lookback(R))$ bits is necessary for online tokenization.

\begin{restatable}[Space Lower Bound]{theorem}{SpaceLB}
    \label{thm:online-tokenization-space-lb}
    For every integer $f$, 
    there exists $R$ where $\lookback(R) = 2f$ and online tokenization requires $\Omega(\lookback(R))$ bits in space.
\end{restatable}

While our algorithm requires $O(\lookback(R) + \merges \longest)$ space, we note that $\merges \longest$ is a constant independent of the input $x$, while $\lookback(R)$ can be arbitrarily large (see e.g. $R = a^* b$ above).
We note that \Cref{thm:bpe-space-lb} shows that $\Omega(\merges )$ bits of space are necessary for online BPE, and therefore online tokenization, essentially showing that our space usage is optimal (up to a $\longest$ factor and word size).

Finally, we discuss delay-optimal algorithms for online tokenization.
Note that \Cref{thm:linear-online-tokenization-w-delay} incurs delay $O(\lookback(R) + \merges \longest)$\footnote{Note that the size of $S^*$ is a constant independent of all other parameters. We focus on $R$ and $\merges $ parameters here as $S^*$ is typically not a bottleneck in practice, and we set $S^* = \emptyset$ in the construction of \Cref{thm:np-delay-lb}} and thus incurs additional $O(\lookback(R))$ delay over the lower bound of \Cref{thm:np-delay-lb}.
We show that this is not an artifact of our analysis: composition alone will always incur delay $\Omega(\lookback(R))$ on some instances.

\begin{restatable}[Composition Delay Bound]{theorem}{CompositionDelayLB}
    \label{thm:composition-delay-lb}
    For every integer $f \geq 2$ 
    there exists alphabet $\Sigma$, a total regular expression $R$ over $\Sigma^*$, merge list $M$ and some input $x$ where $\lookback(R) = f$ and
    \begin{equation*}
        \Merge_{\preceq}(\Split_{\preceq}(\Seg_{\preceq}(x))) \prec \Tok_{\preceq}(x[:-(f - 2)]) \text{.}
    \end{equation*}
\end{restatable}

Thus, if we want to obtain a delay-optimal algorithm, it cannot be from composition alone.
Using these insights, our main result establishes an algorithm that achieves near linear time, optimal space, and optimal delay.

\begin{restatable}[Online Tokenization with Constant Delay]{theorem}{MTDelay}
    \label{thm:online-tokenization-const-delay}
    Let $S^* \subset \Sigma^*$ be a set of special tokens, $R$ a total regular expression over $\Sigma$, and $M$ a merge list.
    There is an online tokenization algorithm with preprocessing time $O(\merges  \longest)$, total time $O(n (\log^2 \longest + \log \merges  + \log \min(n, \lookback)))$ and space $O(\min(n, \lookback) + \merges \longest)$ with $O(\merges \longest)$ delay.
\end{restatable}

In particular, we characterize (up to constant factors) for what delay efficient online tokenization is possible.
When delay is $O(\merges \longest)$, there is a near linear time algorithm.
On the other hand, when delay is $o(\merges \longest)$, there is no polynomial time algorithm (unless $\P = \NP$).

%% file: empirical.tex
We implement a practical variant and obtain an efficient tokenizer that streams an input file of characters to an output file of tokens.
We evaluate our tokenizer against \python{tiktoken} and \python{tokenizers} as a drop-in replacement on code (Github), English, and Chinese datasets.
See \Cref{sec:experimental-setup} for experimental setup, dataset composition, as well as implementation details and optimizations.
Our code is available 
\ificlrfinal
  at \url{https://github.com/morganstanley/MSML/tree/main/papers/Online_Tokenization_with_Optimal_Delay_in_Near_Linear_Time}.
\else
  in the supplementary material.
\fi
Our experiments reveal that our implementation exhibits a significant advantage over all evaluated baselines with respect to latency and memory consumption. 
We also describe some simple optimizations that lead to significant throughput advantages over \python{tiktoken} and \python{tokenizers}.

\subsection{Latency Benchmark}

We measure latency as time to first token (TTFT).
Since \python{tiktoken} and \python{gigatoken} do not natively support TTFT measurement, we make minimal changes to their code base to allow TTFT measurement, which is measured as the time from when the input is provided to the tokenizer to the time that the first token is formed.
\Cref{tbl:ttft-speedup} shows that our tokenizer achieves $7.9\times$-$34.8\times$ TTFT advantage over the best of \python{tiktoken} and \python{gigatoken} for every encoding and dataset, exhibiting the significant latency advantage of online tokenization.

\begin{table}[ht]
\caption{Time to first token (TTFT) speedup of our tokenizer over the better of \python{tiktoken} and \python{gigatoken} on each corpus.
}
\label{tbl:ttft-speedup}
\begin{center}
\begin{tabular}{|l|c|c|c|}
\hline
{\bf Encoding} & {\bf Github} & {\bf English} & {\bf Chinese} \\
\hline
\python{r50k} & 21.7$\times$ & 17.8$\times$ & 14.3$\times$ \\
\python{p50k} & 19.2$\times$ & 20.3$\times$ & 13.8$\times$ \\
\python{cl100k} & 34.8$\times$ & 11.1$\times$ & 7.9$\times$ \\
\python{o200k} & 17.8$\times$ & 11.8$\times$ & 12.5$\times$ \\
\hline
\end{tabular}
\end{center}
\end{table}

\subsection{Memory Benchmark}

We measure peak RSS (Resident Set Size) as input size increases.
Our implementation exhibits approximately constant peak RSS over the range of input sizes, while existing tokenizers retain the full input or output in memory.
See \Cref{fig:github-peak-rss} for Github, \Cref{fig:english-peak-rss} for English, and \Cref{fig:chinese-peak-rss} for Chinese.
Our tokenizer reports lower memory usage than \python{gigatoken} across all encodings, datasets, and input sizes.
Compared to \python{tiktoken}, our tokenizer reports higher memory usage on small inputs due to initialization overhead but lower memory usage on large inputs.
The crossover point is between $2$ and $32$ MiB-sized inputs depending on the encoding and dataset.

\subsection{Throughput Benchmark}

Finally, \Cref{tbl:combined-speedup} exhibits the throughput comparison of our tokenizer against \python{tiktoken} and \python{tokenizers} baselines.
Depending on the encoding and dataset, our tokenizer exhibits a throughput advantage of $2.4\times$-$7.0\times$ over \python{tiktoken} and $15.8\times$-$38.1\times$ over \python{tokenizers}.

Our main optimization replaces the default regular expression of each encoding with a simplified regular expression that allows for faster matching, coupled with a simple, efficient cleanup step that ensures equivalence with the reference regular expression. 
The optimizations of our streaming
algorithms (e.g. simplified regular expressions) are \emph{provably} verified in \Cref{sec:experimental-setup}.
Due to their simplicity, the optimizations are portable across different encodings, achieving consistent speedups for all evaluated tokenizers that fall within the framework considered in this work. 
See \Cref{sec:experimental-setup} for implementation details and practical optimizations.

\begin{table}[ht]
\caption{Throughput speedup of our tokenizer over \python{tiktoken}
(top) and \python{tokenizers} (bottom).
}
\label{tbl:combined-speedup}
\begin{center}
\begin{tabular}{|l|c|c|c|}
\hline
{\bf Encoding} & {\bf Github} & {\bf English} & {\bf Chinese} \\
\hline
\python{r50k} & 3.9$\times$ & 2.7$\times$ & 2.5$\times$ \\
\python{p50k} & 4.2$\times$ & 2.7$\times$ & 2.4$\times$ \\
\python{cl100k} & 4.4$\times$ & 2.9$\times$ & 2.8$\times$ \\
\python{o200k} & 7.0$\times$ & 4.8$\times$ & 3.7$\times$ \\
\hline
\python{GPT2} & 38.1$\times$ & 36.0$\times$ & 32.2$\times$ \\
\python{Mistral} & 32.8$\times$ & 27.9$\times$ & 23.1$\times$ \\
\python{Llama4} & 31.3$\times$ & 24.2$\times$ & 19.8$\times$ \\
\python{Qwen-3} & 19.2$\times$ & 19.7$\times$ & 15.8$\times$ \\
\hline
\end{tabular}
\end{center}
\end{table}

\paragraph{Recent Work}
Recently, the \python{gigatoken} tokenizer presents a parallelized tokenizer that achieves significant throughput speedups over previous tokenizers \citep{roed2026gigatoken}.
While \python{gigatoken} achieves high throughput, our implementation, due to its practical streaming design, achieves lower measured TTFT and approximately constant peak RSS over the tested input sizes.
We remark that some optimizations of \python{gigatoken} are compatible with online tokenization, and observe that a combination of the two can achieve improved throughput for some encodings and datasets, especially on the Chinese dataset (see \Cref{sec:gigatoken} for details).\footnote{Some \python{gigatoken} optimizations, such as more aggressive caching of observed pretokens can introduce additional
memory usage that grows with the number and size of distinct pretokens
encountered.}
\citet{verma2026component} give an online BPE algorithm with total time $O(n \log \longest)$, improving upon the $O(n \log^2 \longest)$ of \citet{jiang2025incrementalbpe}.
This immediately implies a runtime improvement saving a $\log \longest$ factor to our composition algorithm (\Cref{thm:linear-online-tokenization-w-delay}).

%% file: iclr/iclr_overview.tex
We now describe the main ideas behind our algorithms.

\paragraph{Efficient Online Pretokenization}
We begin with \Cref{thm:linear-online-tokenization-w-delay}.
Our algorithm is a simple composition of online segmentation, pretokenization, and BPE.
Thus, in order for composition to be efficient, we need efficient implementations of each online component.
Given the efficient online BPE algorithm \citep{jiang2025incrementalbpe}, it suffices for us to design a linear time online pretokenization algorithm with zero delay.

\begin{lemma}[Online Pretokenization (Informal \Cref{lem:pretokenize-alg})]
    \label{lem:pretokenize-informal}
    There is an algorithm that computes online pretokenization in $O(n)$ time and $O(\lookback(R))$ space.
\end{lemma}

Given a string $x$, let $F(x)$ denote the longest prefix of $x$ that matches $R$.
For simplicity, suppose $R$ can be recognized by a DFA $\calA$ of size $\poly(|R|)$ (in \Cref{lem:pretokenize-alg} we describe how to handle NFAs).
Let $\calQ$ denote the set of states and $\accept$ the accepting states of $\calA$.
For a string $w \in \Sigma^+$, let $\delta(q, w)$ denote the state computed if we begin at state $q$ and read $w$, and let $\last(q, w)$ be the length of the longest prefix $w' \preceq w$ such that $\delta(q, w') \in \accept$; i.e. reading $w'$ from $q$ ends up at an accepting state.
A state $q \in \calA$ is $\dead$ if there is no $y \in \Sigma^+$ such that $\delta(q, y) \in \accept$.\footnote{An accepting state can be $\dead$, for example, if every transition leads to a $\dead$ state.}

Consider the following simple algorithm: 
initialize the variable $\state$ to $\start$ and on input $x_i$ update $\state \gets \delta(\state, x_i) = \delta(\start, x_{\leq i})$.
Whenever $\state \in \accept$, we know that the current prefix $x_{\leq i} \preceq F(x)$.
That is, $x_{\leq i}$ is a prefix of the longest prefix of $x$ that matches $R$, so we can emit $x_{\leq i}$ and set $\last \gets i = \last(\start, x_{\leq i})$.
Whenever $\state \in \dead$, we know that the last accepting prefix (denoted $x_{\leq \last}$) is maximal and therefore $F(x) = x_{\leq \last}$.
Since we have already emitted $x_{\leq \last}$, we can emit $\brk$, and restart the algorithm on input $x_{\last + 1}$. 

We now analyze the performance of this simple algorithm.
On each input $x_{i}$, our algorithm needs to maintain (1) $\state = \delta(\start, x_{\leq i})$ to check whether $\state \in \accept$ is accepting or $\state \in \dead$ is dead and (2) $\last = \last(\start, x_{\leq i})$ to track the longest accepting prefix so far.
By restarting the algorithm at $x_{\last + 1}$, naively scanning the input again requires $O(i - \last)$ time.
Since every accepting path through the DFA reaches an accepting state once every $\lookback := \lookback(R)$ steps we have $i \leq \last + \lookback$ (\Cref{lem:live-lookahead}).
Thus, our algorithm backtracks by at most $\lookback$ steps on every output, yielding a running time of $O(n\lookback)$.
If we could instead recompute $\delta(\start, x_{\last + 1} \dots x_{i})$ and $\last(\start, x_{\last + 1} \dots x_{i})$ efficiently, then we would have an efficient algorithm for online pretokenization.

\citet{li2026static} address this by precomputing $\delta, \last$ for all possible strings of length $\lookback$, resulting in preprocessing time and space that is exponential in $\lookback$.
In this work, we instead give a simple data structure $\cqueue$ that supports exactly the required operations.

\begin{restatable}{lemma}{ConcatQueue}
    \label{lem:concat-queue}
    There is a data structure $\cqueue$ supporting the following operations on a string $w$:
    \begin{enumerate}
        \item (Push) Append a character $c$ to $w$ i.e. update $w \gets wc$.
        \item (Pop) Pop a prefix of $w$, i.e. return $w_1 \dots w_{p}$ and update $w \gets w_{p + 1} \dots w_{\ell}$ for some $p > 0$.
        \item (Query) Given a state $q$, compute $\delta(q, w)$ and $\last(q, w)$.
    \end{enumerate}
    Furthermore, the data structure supports the operations in amortized $O(1)$ time.
\end{restatable}

Now, to implement the online pretokenization algorithm efficiently, it suffices to observe that the string $x_{\last + 1} \dots x_{i}$ which we want to query only undergoes queue updates.
In particular, we push $x_{i}$ on each input, and whenever $\state \in \dead$, we pop up to $x_{\last}$.
Thus, by supporting the required operations in amortized $O(1)$ time, we obtain a total time $O(n)$ algorithm for online pretokenization.

\paragraph{Composition Delay Bound}
We now argue that composition already achieves non-trivial delay $O(\lookback + \merges \longest)$.
Composition always outputs a stable prefix (and this in fact holds for composition of any two functions $f, g$ see \Cref{clm:composition-is-stable}).
If $x' \preceq g(xy)$ for any continuation $y$, and $z \preceq f(x'y')$ for any continuation $y'$, then $z \preceq f(g(xy))$ for any continuation $y$, since $g(xy) = x'y'$ for some $y'$.

To bound the delay, it suffices to bound $|x| - |\raw(z)|$ where $z$ is the output of our algorithm (see \Cref{clm:coverage}).
Let $z_1 = \Seg_{\preceq}(x)$.
In \Cref{lem:segment-alg}, we show $|\raw(z_1)| \geq |x| - \ell_{\infty}(S^*)$ where $\ell_{\infty}(S^*) = \max_{\tau \in S^*} |\tau|$ is the length of the longest special token. 
Now, let $z_2 = \Split_{\preceq}(z_1)$.
As discussed in the pretokenization algorithm, there are two cases when we can safely emit a stable prefix of $\Split$: either $\state \in \accept$ and we can emit the accepted prefix, or $\state \in \dead$ and we can emit $\brk$ and rewind to the last accept.
Since the DFA must reach a $\accept$ state every $\lookback = \lookback(R)$ steps, $\Split_{\preceq}$ outputs all but the last $\lookback(R)$ characters of $z_1$ (see \Cref{lem:live-lookahead}), and $|\raw(z_2)| \geq |\raw(z_1)| - \lookback$.
Finally, let $z = \Merge_{\preceq}(z_2)$ be the output of our algorithm.

We use the key property of BPE which states that for all inputs $x$,
\begin{equation}
    \label{eq:online-bpe-bound}
    |\raw(\Merge_{\preceq}(x))| \geq |x| - (\merges + 1) \longest \text{.}\footnote{While known online BPE algorithms do not provably compute $\Merge_{\preceq}$, the algorithms nevertheless guarantee that the above bound holds (\Cref{lem:bpe-merge-alg})}
\end{equation}
Combining all three bounds we can upper bound the delay by
\begin{equation*}
    |x| - |\raw(z)| \leq |x| - |\raw(z_2)| + (\merges +1)\longest \leq |x| - |\raw(z_1)| + \lookback + (\merges +1)\longest \leq \ell_{\infty}(S^*) + \lookback + (\merges +1) \longest \text{.}
\end{equation*}

\paragraph{Limitations of Composition}
While composition produces a stable prefix with non-trivial delay, it falls short of an optimal delay algorithm for online tokenization. 
While the additive $\ell_{\infty}(S^*)$ term is an input-independent constant (and usually not a large one in practice), the additive $\lookback(R)$ term can be arbitrarily large.
\Cref{thm:composition-delay-lb} shows that this delay is inherent to composition.

Consider the following simple example of a tokenizer: $R = a^{n} b \mid a \mid b$ and $M = [bb \mapsto B]$; and an input $x = a^{n}$.
On the one hand, $\lookback(R) = n$ and indeed $\Split_{\preceq}(x) = a$, since two possible continuations produce $\Split(x) = a \brk \dots \brk a$ and $\Split(xb) = a^{n} b$ with a longest common prefix of $a$.
On the other hand, since no pretoken contains $bb$, then regardless of the pretokenization, no merge ever occurs. 
For our specific example, we in fact have $\Tok_{\preceq}(y) = y$ for any input $y$.
Concretely, we argue the delay is greater than $n - 2$ as the algorithm outputs $a$ and
\begin{equation*}
    \Tok_{\preceq}(x[:-(n-2)]) = \Tok_{\preceq}(x_{\leq 2}) = \Tok_{\preceq}(aa) = aa \not\preceq  a \text{.}
\end{equation*}

From the above example, we observe that composition will always produce an output $P$ where $\raw(P) \preceq x_{\leq \last}$, where $x_{\leq \last}$ is the longest accepting prefix in the pretokenization.
However, a zero-delay online tokenization algorithm might produce output $P'$ that is strictly longer than $_{\leq \last}$, i.e. $x_{\leq \last} \prec P'$.
This is because different continuations may produce different pretokenization outputs, but they nevertheless can produce the same tokenization at the end.

\paragraph{Optimal Delay Algorithm} %
We are now ready to describe our algorithm with optimal delay.
Suppose our algorithm proceeds inductively, beginning with an empty output which is a stable prefix of all possible outputs.
Consider some stable prefix $P$ of the online tokenization output, i.e. $P \preceq \Tok_{\preceq}(x)$ which we have computed so far.
Since $\Tok_{\preceq}(x) \preceq \Tok(x)$, any stable prefix $P' \preceq \Tok_{\preceq}(x)$ must be a prefix of $\Tok(x)$.
Let $\theta$ be the token following $P$ in $\Tok(x)$, i.e. $P \theta \preceq \Tok(x)$.
Then, $P$ is not a maximum stable prefix if and only if $P \theta \preceq \Tok_{\preceq}(x)$: $\theta$ is the only candidate token to extend $P$ to a longer stable prefix.
It thus suffices to decide whether $P \theta \preceq \Tok_{\preceq}(x)$, equivalently, whether for all $y \in \Sigma^+$ we have $P \theta \preceq \Tok(xy)$.

Recall that $\last = \last(\start, x)$ is the index of the last accepting state in the pretokenization DFA.
Denote $\Split(x) = z_1 \brk \dots \brk z_{p}$ so that $z_1 = x_{1} \dots x_{\last}$\footnote{Note that a different continuation $xy$ may have different outcome $\Split(xy)=z_1\brk z_{2}'\brk \dots \brk z_{p'}'$.} and 
\begin{equation*}
    \Tok(x) = \Merge(z_1) \circ \Merge(z_2) \circ \dots \circ \Merge(z_{p}) \text{.}
\end{equation*}
Let $\out$ be the index of the last output character so that $\raw(P) = x_{1} \dots x_{\out}$.
Let $\last(0) = 0$ and $1 \leq q \leq p$, let $\last(q)$ be the index such that $z_{q} = x_{\last(q - 1) + 1} \dots x_{\last(q)}$.
Note that $\last = \last(1)$.

We now describe our optimal $O(\merges \longest)$ delay tokenization algorithm.
Let $x$ be the current input of length $i$ and $P$ be the output so far (stable by assumption).
In the appendix, we will describe how to use offline tokenization to reduce $\state \in \dead$ to the $\state \not\in \dead$ case. 
Thus, without loss of generality, assume $\state = \delta(\start, x) \not\in \dead$. 
First, if $i - \out \leq 3 \merges \longest$, then we already have that $P$ has delay $O(\merges \longest)$.
Thus, assume $i - \out > 3 \merges \longest$.

Now, since $i - \out > 3 \merges \longest$, we have
\begin{equation*}
    |x_{\out + 1} \dots x_{i} y| \geq |x_{\out + 1} \dots x_{i}| \geq i - \out > 3 \merges \longest \text{.}
\end{equation*}
From \Cref{eq:online-bpe-bound}, we can conclude that $\tau := \phi(x_{\out + 1} \dots x_{i} y)$ is independent of $y$, where $\phi(w)$ denotes the first token in $\Merge(w)$ for any $w$. %
In fact, $\tau = \phi(x_{\out + 1} \dots x_{i} y) \preceq \Merge_{\preceq}(x_{\out + 1} \dots x_{i})$ can be computed with online BPE applied to a suffix of the current input.
Of course, if $\tau = \theta$, then we have $P \theta \preceq \Tok_{\preceq}(x)$ so that we can safely emit $\theta$ and proceed inductively.
On the other hand, if $\tau \neq \theta$, the fact that $\state \not\in \dead$ already tells us that $P = \Tok_{\preceq}(x)$ is a maximum stable prefix, since there is a continuation $y$ such that
\begin{equation*}
    \Tok(xy) = \Merge(xy) = P \circ \Merge(x_{\out + 1} \dots x_{i} y) \succeq P \tau \text{.}
\end{equation*}

To implement the above algorithm efficiently, we need a few additional data structures.
First, we need to compute $\theta$, the next token in $\Tok(x)$, denoted by $\phi(x_{\out + 1} \dots x_{\last(q)})$ where $x_{\last(q - 1) + 1} \dots x_{\last(q)} = z_{q}$ %
is the pretoken in $\Split(x)$ with $\last(q - 1) \leq \out < \last(q)$.
Second, we need to compute $\tau := \phi(x_{\out + 1} \dots x_{i})$, the next token in the event that a future continuation extends the current pretoken match.
To support these operations efficiently, we implement two useful data structures.
Note that both data structures support pushing string characters and $\brk$ to denote the beginning of a new pretoken (pushing $\brk$ does not contribute to the length of the maintained string).

The first data structure $\incbpe$ allows us to compute $\theta$ efficiently.
 \begin{restatable}[Incremental BPE]{lemma}{incBPELem}
    \label{lem:incremental-bpe}
    Given a BPE merge list $M$ with longest token length $\longest$, there is a data structure $\incbpe$ that maintains implicitly an append-only history string $h$, an index $0 \leq I \leq |h|$, and explicitly string $w = h_{> I}$ under the following operations:
    \begin{enumerate}
        \item (Push) append $c$ to $h \gets hc$. Note that we can push $\brk$ to denote a pretoken split.
        \item (Pop) if $I < |h|$, increment $I$ by $1$.
        \item (First Token) if $\Merge(h_{\leq I}) \prec \Merge(h)$, compute $\phi(w)$. 
        Otherwise return $\bot$.
    \end{enumerate}
    The data structure has $O(\merges \longest)$ preprocessing, space $O(\merges \longest + |w|)$, worst-case update time $O(\log^2 \longest + \log \merges)$ and query time $O(\log \longest + \log |w|)$.
\end{restatable}

The second data structure $\bpequeue$ allows us to compute $\tau$ efficiently.
\begin{restatable}[BPE Queue]{lemma}{BPEQLem}
    \label{lem:bpe-queue}
    Given a BPE merge list $M$, there is a data structure $\bpequeue$ supporting the following operations on a string $w$:
    \begin{enumerate}
        \item (Push) Append a character $c$ to $w$ i.e. update $w \gets wc$.
        \item (Pop) Pop $w_1$, i.e. return $w_1$ and update $w \gets w_{2} \dots w_{\ell}$.
        \item (First Token) If $|w| \geq 3 \merges  \longest$, then compute $\phi(w)$. Otherwise, return $\bot$.
    \end{enumerate}
    Furthermore, the data structure has $O(\merges \longest)$ preprocessing, $O(\merges  \longest + |w|)$ space, and supports updates and queries in amortized time $O(\log^2 \longest)$. 
\end{restatable}

Furthermore, \Cref{lem:pretok-q} gives a data structure $\cutQ$ that allows us to efficiently compute $\last(q)$ satisfying $\last(q - 1) \leq \out < \last(q)$.
We note that $\last(q)$ is monotonically increasing: since we are always looking for the longest match, additional characters can only extend a pretoken.

We now sketch our final algorithm on input $x_i$ (see \Cref{thm:online-tokenization-const-delay} for the full algorithm and \Cref{fig:alg-overview} for an illustration of the data structures in our algorithm).
As before, we will assume $\delta(\start, x) \not\in \dead$.
Initialize data structures $\incbpe$ and $\bpequeue$.
Denote data structure queries by $\incbpe.\phi$ and $\bpequeue.\phi$.
\begin{mdframed}
    Sketch of {\bf $\python{OptimalDelayTokenization}(x_{i})$}
    \begin{enumerate}
        \stepnames
        \item Push $x_i$ onto $\bpequeue$.
        Ensure $\incbpe$ maintains $x_{\out + 1} \dots x_{\last(q)}$. 
        \label{step:push-overview}

        \item While $\incbpe.\phi = \bpequeue.\phi \neq \bot$ (i.e. first token queries agree), output $\incbpe.\phi$ and pop $\incbpe.\phi$ from $\incbpe$ and $\bpequeue$.
        \label{step:while-match-overview}

        \item If $\incbpe = \epsilon$ and $\out < i$: add $q \gets q + 1$, push $\brk x_{\last(q - 1) + 1} \dots x_{\last(q)}$ to $\incbpe$, and return to \Cref{step:while-match-overview}.
        
        \Comment{(When $\incbpe$ is empty, we have processed the current pretoken $z_{q}$, so increment $q$ and proceed to the next pretoken.)}
        \label{step:next-cut}
        \end{enumerate}
\end{mdframed}
Our algorithm maintains the following invariants (see \Cref{fig:alg-overview}):
\begin{enumerate}
    \item if $\out < i$, $q$ satisfies $\last(q - 1) \leq \out < \last(q)$. 
    \item $\incbpe$ maintains $x_{\out + 1} \dots x_{\last(q)}$.
    \item $\bpequeue$ maintains $x_{\out + 1} \dots x_{i}$.
\end{enumerate}
Observe that whenever $i - \out > 3 \merges \longest$, $\incbpe$ is guaranteed to produce the correct output.
We only need $\push$ operations in \Cref{step:push-overview}, since $\last(q)$ is monotonically increasing.
We will show in the proof of \Cref{thm:online-tokenization-const-delay} that every query to $\incbpe$ is valid.
Following the above discussion, the displayed algorithm obtains a delay bound of $O(\merges \longest)$.
The running time follows from the near linear running time of both data structures.

\input{tikz/alg_overview}

\paragraph{Hardness of Exact Online Tokenization}
Our algorithm does not achieve exact online tokenization (zero delay), and we give a lower bound ruling out any polynomial algorithm unless $\P = \NP$. 
We will reduce $\kSAT{3}$ to online tokenization.
Consider $\kSAT{3}$ on $d$ variables.
First, we construct an alphabet $\Sigma$, regular expression $R$, and merge list $M$ depending only on $d$. 
We will ensure that the alphabet size, the size of a DFA recognizing the regular expression, and the length of the merge list all have size $\poly(d)$.
Then, we give a reduction from any $d$-variable $\kSAT{3}$ instance $\psi$ to input strings $z_{\psi}$ of length $\poly(d)$ such that there is a specific alphabet character $\alpha \in \Sigma$ that appears in $\Tok_{\preceq}(x)$ if and only if $\psi$ is unsatisfiable (see \Cref{thm:np-delay-lb} for details and proof with delay).

The high level idea of our approach is as follows.
For each of the $O(d^{3})$ possible width 3 clauses $c \in \calC$, define a symbol $\ell_{c}$.
Define $\Sigma_{\calC} = \set{\ell_{c}}_{c \in \calC}$.
Define the alphabet
\begin{equation*}
    \Sigma = \set{a_{1}^{0}, a_{1}^{1}, b_1, \dots, a_{d}^{0}, a_{d}^{1}, b_{d}} \cup \set{\alpha, \beta, \gamma, \delta} \cup \Sigma_{\calC} \text{.}
\end{equation*}
For every assignment $x \in \set{0, 1}^{d}$, we encode $x$ as $E(x) = a_{1}^{x_{1}} b_1 \dots a_{d}^{x_{d}} b_{d}$.
For a clause $c$, let $\calE_{c} = \set{E(x) \given c(x) = 0}$ denote the encodings of all assignments that falsify $c$ and build the regular expression $R_{c} = \ell_{c} \Sigma_{\calC}^* \alpha \beta \calE_{c} \gamma$.

To define the tokenization, fix $S^* = \emptyset$ and $R = R_{\calC} \mid \alpha \beta \calE \gamma \delta \mid \Sigma$ where $R_{\calC}$ is the alternation of all $R_{c}$, $\calE$ is the alternation of all encodings, and $\Sigma$ is the alternation of all characters added to ensure $R$ is total.
The merge rules will be chosen so that within any pretoken, $\alpha$ is merged if and only if the suffix of the pretoken beginning with $\alpha$ has even length.
Observe that this occurs if and only if $\alpha$ occurs in the pretoken $\alpha \beta \calE \gamma \delta$ (since $\Sigma, R_{\calC}$ both have odd length suffixes beginning with $\alpha$).
Then, given a $\kSAT{3}$ instance $\psi = c_1 \wedge \dots \wedge c_{k}$, construct the input $z = \ell_{c_{1}} \dots \ell_{c_{k}} \alpha$.

If $\psi$ is satisfiable (say by $x^*$), we will argue that the continuation $z \beta E(x^*) \gamma \delta$ corresponding to the satisfying assignment $x^*$ has no match with any $R_{\calC}$, and thus matches $\alpha \beta \calE \gamma \delta$, in which case $\alpha$ is merged and does not appear in $\Tok_{\preceq}(z)$.
On the other hand, when $\psi$ is unsatisfiable, we argue that any match involving $\alpha$ matches $R_{\calC}$ corresponding to some unsatisfied clause, so that $\alpha$ is never merged and always appears in $\Tok_{\preceq}(z)$.

%% file: tikz/alg_overview.tex
\begin{figure}[ht]\centering
\begin{tikzpicture}

  \fill[green!25]  (1,-0.16) rectangle (11,0.16);
  \fill[blue!25]   (6,0) rectangle (11,0.16);
  \fill[orange!30] (6,-0.16) rectangle (7,0);

  \draw[blue!70!black, thick]
    (6,-0.98) -- (6,-1.22) -- (11,-1.22) -- (11,-0.98);
  \node[below] at (8.5,-1.22) {\bpequeue};

  \draw[orange!80!black, thick]
    (6,-1.28) -- (6,-1.52) -- (7,-1.52) -- (7,-1.28);
  \node[below] at (6.5,-1.52) {\incbpe};

  \draw[green!60!black, thick]
    (4,0.38) -- (4,0.62) -- (11,0.62) -- (11,0.38);
  \node[above] at (7.5,0.62) {\cutQ};

  \draw[black!60!black, thick]
    (0,-1.22) -- (5.95,-1.22) -- (5.95,-0.98);
  \node[below] at (3.5,-1.22) {$P$};

  \draw[thick] (0,0) -- (12,0);

  \foreach \x/\label in {
      1/{\last},
      6/{\out},
      11/{i}
  }{
    \draw[black, thick] (\x,-0.18) -- (\x,0.18);
    \node[below=3pt] at (\x,-0.14) {$\label$};
  }

  \foreach \x/\label in {
      2.5/{\last(2)},
      4/{\last(3)},
      7/{\last(4)},
      9/{\last(5)}
  }{
    \draw[gray!55, thick] (\x,-0.14) -- (\x,0.14);
    \node[below=3pt, text=gray!65] at (\x,-0.14) {$\label$};
  }

\end{tikzpicture}
\caption{Illustration of $O(\merges \longest)$ delay algorithm. 
Characters up to $\out$ are output in $P$.
$\last, \last(2), \dots, \last(5), i$ are endpoints of pretokens if the current input $x_{\leq i}$ terminates. 
Square brackets describe the substring of the input maintained by each data structure.
Here, we have $q = 4$.}
\label{fig:alg-overview}
\end{figure}
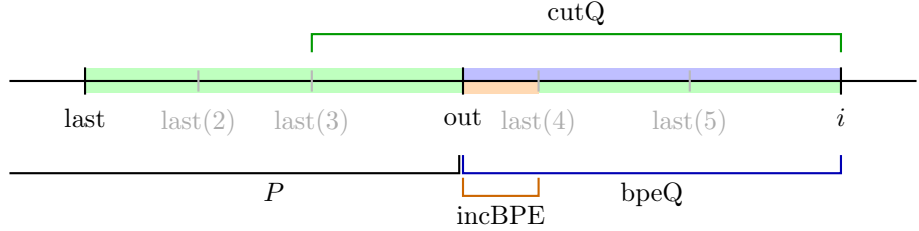

%% file: prelims.tex
\section{Preliminaries}
\label{sec:prelims}

Let $\Sigma$ denote an alphabet, $\Sigma^{k}$ the set of length $k$ strings over $\Sigma$, and $\Sigma^* = \bigcup_{k = 0}^{\infty} \Sigma^{k}$ the set of (any length) strings over $\Sigma$.
Let $[n] = \set{1, \dots, n}$.
Let $\epsilon$ denote the empty string and $\circ$ concatenation.
Let $|x|$ denote the length of a string.
Let $x[i:j]$ or $x_{i} \dots x_{j}$ denote the substring of $x$ starting at the $i$-th index and ending at the $j$-th index (inclusive and indexed at $1$).
We let $x_{\leq i} = x_{1} \dots x_{i}$ and $x_{\geq i} = x_{i} \dots x_{n}$ where $n = |x|$.
$x_{<i}$ and $x_{>i}$ are defined analogously.
Let $x[:-k]$ denote the prefix of $x$ obtained by deleting the last $k$ characters and $x[-k:]$ the suffix containing the last $k$ characters.
We say $x \preceq y$ if $x$ is a prefix of $y$.
We use $y \succeq x$ to mean that $y$ contains $x$ as a prefix, not that $y$ is a suffix of $x$.
We use $a \addgets b$ to denote $a \gets a + b$.

\paragraph{Regular Expressions}

The regular languages over an alphabet $\Sigma$ consist of the smallest collection $\calR = \calR(\Sigma)$ of languages $L \subseteq \Sigma^*$ that satisfy the following rules: 
\begin{enumerate}
    \item $\emptyset \in \calR$, $\set{\epsilon} \in \calR$, and $\set{\sigma} \in \calR$ for all $\sigma \in \Sigma$.
    \item (Concatenation) If $L_1, L_2 \in \calR$, then $L_1 \circ L_2 = \set{x_1 \circ x_2 \mid x_1 \in L_1, x_2 \in L_2} \in \calR$.
    \item (Alternation) If $L_1, L_2 \in \calR$, then $L_1\mid L_2 = L_1 \cup L_2 \in \calR$.
    \item (Star) If $L \in \calR$, then $L^* \in \calR$.
\end{enumerate}
When convenient, we use $L_1L_2$ to denote $L_1 \circ L_2$.
A regular expression $R$ is \emph{total} if every non-empty $x \in \Sigma^*$ has a non-empty prefix matched by $R$ (possibly including $x$ itself).
Note that this implies $\Sigma \subseteq L(R)$.
Throughout this work, we consider only total regular expressions.

Regular languages are closed under complement, union and finite intersection.
For our theoretical results, a match of regular expression $R$ in $x$ is the longest prefix of $x$ that matches $R$, or equivalently the longest prefix of $x$ in $L(R)$.
Since we assume $R$ is total, this is a non-empty prefix.

For a total regular expression $R$ and non-empty input $x$, define the longest match of $R$ in $x$, denoted $F(x; R)$ as $x_{\leq J(x; R)}$ where $J(x; R) = \max \set{j > 0 \given x_{\leq j} \in L(R)}$.
When the context is clear, we simply write $F(x), J(x)$.
We note that for any regular expression $R$, there is a non-deterministic finite automata (NFA) that recognizes exactly $L(R)$ of size $O(|R|)$ (see e.g. \citet{aho2006compilers}).
In \Cref{sec:experimental-setup}, we will separately discuss examples of regular expressions used in practice, which use regular expressions with priority rules and lookahead.

Given a DFA or an NFA $\calA$, let $\calQ$ be its states.
Let $\accept \subseteq \calQ$ be the accepting states, and $\dead \subseteq \calQ$ be the set of states that have no \emph{non-trivial} path to a state in $\accept$.
Note that a $\dead$ state can be a $\accept$ state,
since we do not count the trivial transition on the empty string $\epsilon$. 
For our purposes of simulating a regular expression with a DFA, if a state is $\dead$ and in $\accept$, it is a maximal match.
Let $\start$ denote the start state.
For a state $q$ and word $w$, $\delta(q, w)$ denotes the state computed if we start on $q$ and read $w$.
$\last(q, w)$ denotes the maximum $i$ where $\delta(q, w_{\leq i}) \in \accept$, or $\bot$ if no such $i$ exists.
For an NFA, we let $Q$ be a set of states and $\delta(Q, w)$ be the set of states reachable if starting at any $q \in Q$ and reading $w$ and $\last(Q, w)$ the maximum $i$ where $\delta(Q, w_{\leq i}) \cap \accept \neq \emptyset$, or $\bot$ if no such $i$ exists.

For a merge list $M$, we use merge list and BPE encoding interchangeably.
For a given token $\tau$, let $\raw(\tau)$ denote the expansion of $\tau$, i.e. the raw string $x$ such that $\Merge(x) = \tau$.
For a sequence of tokens $\tau_1 \dots \tau_k$, let $\raw(\tau_1 \dots \tau_k)$ denote the concatenation of $\raw(\tau_1), \dots, \raw(\tau_k)$.
Let the length of a token or a sequence of tokens be defined $|\tau| = |\raw(\tau)|$.

%% file: online_tokenization.tex
\section{Online Tokenization}
\label{sec:online-tokenization}

We begin by defining tokenization in the standard offline batch setting.

\begin{definition}[Special Tokens]
    \label{def:special-tokens}
    Let $\Sigma$ be a finite alphabet and $S^* \subset \Sigma^*$ a finite set of special tokens.
    For each string $x \in S^*$, $s^*(x)$ is a single special character corresponding to $x$.
    Given a string $x \in \Sigma^*$, the segmenting of $x$ by $S^*$ is given by
    \begin{equation*}
        \Seg(x; S^*) = \begin{cases}
            \epsilon & \text{if $x = \epsilon$} \\
            \brk s^*(x_{\leq j}) \brk \Seg(x_{> j}; S^*) & \text{if $j = \max \set{i \given x_{\leq i} \in S^*} \geq 1$} \\
            x_{1} \circ \Seg(x_{>1}; S^*) & \otherwise
        \end{cases} \text{.}
    \end{equation*}
    The output consists of $\brk$ symbols with strings between. 
    Each string is called a special token if it is $s^*(x)$ for $x \in S^*$, and a segment otherwise.
    Let $\ell_{1}(S^*) = \sum_{x \in S^*} |x|$ and $\ell_{\infty}(S^*) = \max_{x \in S^*} |x|$.
\end{definition}

\begin{definition}[Pretokenization]
    \label{def:pretokenization}
    Let $\Sigma$ be a finite alphabet and $R$ a total regular expression over $\Sigma$.
    Then, the pretokenization of $x$ with respect to $R$ is 
    \begin{equation*}
        \Split(x; R) = \begin{cases}
            \epsilon & x = \epsilon \\
            F(x; R) \brk \Split(x_{>|F(x; R)|}; R) & \otherwise
        \end{cases} \text{.}
    \end{equation*}
    Given a list of segments and special tokens $y = x_1 \brk \dots \brk x_m$, define $\Split(y; R)$ to be $\Split$ applied to each segment with the results concatenated.
    Note that $\Split$ is not applied to special tokens and trailing $\brk$ is removed.
    The output of $\Split$ on a single segment consists of $\brk$ symbols with strings between; we call these strings pretokens.
    Thus, on a list of segments and special tokens, $\Split$ produces a list of pretokens and special tokens.
\end{definition}

We formally define bounded lookback, a quantity which will characterize both the delay and space usage of our algorithms.

\begin{definition}[Lookback]
    \label{def:lookback}
    A regular language $L$ has $\lookback$-lookback if for any $x \in L$ with $|x| > \lookback$,
    \begin{equation*}
        \max \set{|y| \given y \prec x, y \in L } \geq |x| - \lookback 
    \end{equation*}
    with the convention $\max \emptyset = - \infty$.
    That is, every word in the regular language of length greater than $\lookback$ has a strict prefix in the language that is shorter by at most $\lookback$.
    We let $\lookback(L)$ denote the minimum $\lookback$ such that $L$ has $\lookback$-lookback.
    $L$ has unbounded lookback if $\lookback(L) > c$ for any finite $c$.
\end{definition}

Lookback ensures that if we do not find an accepting prefix within $\lookback$ characters, the current match (if any) cannot extend by reading more characters of the input.

\begin{definition}[Byte-Pair Encoding]
    \label{def:bpe}
    Let $0 < n_0 < n$ be positive integers.
    Let $\Sigma = [n_0]$ be a finite alphabet.
    Let $M$ be a trained Byte-Pair Encoding (BPE) list of merges with corresponding dictionary $D = n_0 + [M]$ (i.e. a list of ordered pairs $(m_1, m_2) \mapsto d$ for $d \in D$ and $0 < m_1, m_2 < d$).
    Let $\Sigma_{D} = \Sigma \cup D = [n]$ denote the full BPE dictionary.
    We assume the given merge list $M$ is canonical: $\Merge(\raw(\tau);M)=\tau$ for every dictionary token $\tau \in \Sigma_{D}$, as in \citet[Section~3.3]{jiang2025incrementalbpe}.

    Let $m = (m_1, m_2) \mapsto d$ be a merge.
    Given a string $x \in [d - 1]^*$, the BPE of $x$ with respect to $m$ is 
    \begin{equation*}
        \Merge(x; m) = \begin{cases}
            x & |x| \leq 1 \\
            d \circ \Merge([x_3, \dots]; m) & |x| \geq 2 \text{ and } (x_1, x_2) = (m_1, m_2) \\
            x_1 \circ \Merge([x_2, \dots]; m) & \text{otherwise}
        \end{cases} \text{.}
    \end{equation*}
    For a string $x \in \Sigma^*$, the BPE of $x$ with respect to $M$ is 
    \begin{equation*}
        \Merge(x; M) = \begin{cases}
            x & M = \emptyset \\
            \Merge(\Merge(x; \mu_1) ; M \setminus \mu_1) & M = (\mu_1, \dots)
        \end{cases} \text{.}
    \end{equation*}
    Given a list of pretokens and special tokens $y = x_1 \brk \dots \brk x_p$, we define $\Merge(y; M)$ to be $\Merge$ applied to each pretoken with the results concatenated.
    Note that on a special token, $\Merge$ outputs the special token identically.
    We call each element in the output list (i.e. number) a token.
\end{definition}

We note that a BPE encoding is determined by its list of merges $M$ alone.
For succinctness, we use $M$ alone to denote a BPE.
We define the longest token.

\begin{definition}[Longest Token]
    \label{def:longest-token}
    Define $\longest(M) = \max\set{|x| \given |\Merge(x; M)| = 1}$ to be the length of the longest token.
    When the context is clear, we write $\longest = \longest(M)$.
\end{definition}

Throughout, we use the $\brk$ symbol to denote hard boundaries in the input string through which no merges can occur.
The tokenization problem is then a composition of the above three subroutines.

\begin{definition}[Tokenization]
    \label{def:tokenization}
    We define the Tokenization problem as follows.
    Let $\Sigma$ be a finite alphabet.
    Let $S^* \subset \Sigma^*$ be a set of special tokens, $R$ a total regular expression over $\Sigma$, and $M$ a BPE encoding.
    Given a string $x \in \Sigma^*$, the tokenization of $x$ is 
    \begin{equation*}
        \Tok(x; S^*, R, M) = \Merge(\Split(\Seg(x; S^*); R); M) \text{.}
    \end{equation*}
\end{definition}

In running times, we use $R$ to denote $|R|$ and $\merges$ to denote $|M|$.
We now define tokenization in the online setting.
We begin with the notion of a stable output prefix.

\begin{definition}[Stable]
    \label{def:stable}
    Let $S^* \subset \Sigma^*$ be a set of special tokens, $R$ a regular expression over $\Sigma$, and $M$ a BPE encoding.
    A list of tokens $z$ is \emph{stable} with respect to a string $x \in \Sigma^*$ if $z$ is a prefix of $\Tok(x \circ y; S^*, R, M)$ for every string $y \in \Sigma^*$.
    That is, $\Tok$ outputs $z$ as a prefix for every continuation of $x$.
    Let $\calS(x; S^*, R, M)$ denote the set of stable prefixes with respect to $x$.
    When the context is clear, we use $\calS(x)$ to denote the set of stable prefixes.
\end{definition}

We show that any two stable prefixes are comparable, and therefore the maximal/maximum stable prefix is unique.
For two lists $z_1, z_2$, we define the partial order $z_1 \preceq z_2$ if $z_1$ is a prefix of $z_2$.

\begin{claim}
    \label{clm:stable-prefix-compare}
    If $z_1, z_2$ are stable with respect to $x$, then $z_1 \preceq z_2$ or $z_2 \preceq z_1$.
    In particular, $\calS(x; S^*, R, M)$ has a unique greatest element $z^*$ under $\preceq$, and every stable prefix is a prefix of $z^*$.
\end{claim}

\begin{proof}
    We begin with the following fact: if $z \preceq L$ and $u \preceq L$, then $z \preceq u$ or $u \preceq z$.
    Assume without loss of generality $|z| \le |u|$. 
    Since both are prefixes of $L$ and $|z| \le |u| \le |L|$, we have $z \preceq u$.

    Define $\Tok(x) := \Tok(x; S^*, R, M)$. 
    Then, since $z_1, z_2 \in \calS(x)$, we have $z_1, z_2 \preceq \Tok(x)$ so that either $z_1 \preceq z_2$ or $z_2 \preceq z_1$.
    Thus, $\preceq$ is in fact a total order, so that there exists a unique greatest element $z^*$.
\end{proof}

Thus, we may define the following Online Tokenization problem.
A string is incrementally updated if it is updated only by appending new characters.
An algorithm outputs $y$ incrementally if it outputs by appending to the previous output.

\begin{definition}[Online Tokenization]
    \label{def:online-tokenization}
    Let $S^* \subset \Sigma^*$ be a set of special tokens, $R$ a total regular expression over $\Sigma$, and $M$ a BPE encoding.
    Online tokenization, given an incrementally updated string $x$, incrementally outputs $\Tok_{\preceq} (x) = \max_{\preceq} \calS(x; S^*, R, M)$.
    An algorithm for online tokenization with delay $k$ incrementally outputs $z$ such that
    \begin{equation*}
        \Tok_{\preceq}(x[:-k]) \preceq z \preceq \Tok_{\preceq}(x) \text{.}
    \end{equation*}
    Exact online tokenization is the problem of online tokenization with delay $0$.
\end{definition}

We argue that given an incrementally updated string $x$, $\Tok(x)$ also updates incrementally, and thus (at least in principle) an algorithm could maintain the output efficiently.

\begin{claim}
    \label{clm:inc-output}
    If $x \preceq y$, then $\calS(x) \subseteq \calS(y)$.
    Consequently, $(\max_{\preceq} \calS(x)) \preceq (\max_{\preceq} \calS(y))$.
\end{claim}

\begin{proof}
    Let $z \in \calS(x)$ be a stable prefix.
    Then, $z \preceq \Tok(x x')$ for any $x' \in \Sigma^*$.
    Write $y = x a$ for some $a \in \Sigma^*$.
    Then, $z \preceq \Tok(x a y')$ for any $y' \in \Sigma^*$ since $a y' \in \Sigma^*$ is a continuation of $x$.
    Hence, $z \in \calS(y)$, as desired.
\end{proof}

\input{tokenization_lb}

\subsection{Online Tokenization via Composition}

Given \Cref{thm:np-delay-lb}, our goal is to design fast algorithms with minimal delay.
In particular, we would like to design algorithms running in almost linear time (e.g. $n^{1 + o(1)}$).
We begin with a simple algorithm that composes online segmentation, pretokenization, and BPE merging.

We first define online segmentation, pretokenization, and BPE merging.
Note that \Cref{clm:stable-prefix-compare} and \Cref{clm:inc-output} do not use anything specific about the tokenization process, and can be shown for any function that takes an (incremental) list as input and outputs an (incremental) list.
Thus, again in principle, we can hope for efficient online segmentation, pretokenization, and BPE merging.
Since pretokenization takes a list of lists as an input and outputs a list of strings, we say a list of strings is updated incrementally if it is updated by either appending characters to the current last string, or starting a new string.
Similarly, an algorithm outputs a list of strings incrementally by either appending characters to the current last string, or starting a new string.
We let $\calS_{\Seg}(x; S^*)$, $\calS_{\Split}(L; R)$, and $\calS_{\Merge}(L; M)$ denote the set of stable prefixes with respect to $S^*$, $R$, $M$ respectively, where $x$ is a string input to segment with $S^*$, and $L$ is a list of strings input to pretokenization or BPE.

\begin{definition}[Online Segmentation]
    \label{def:online-segmentation}
    Let $S^* \subset \Sigma^*$ be a set of special tokens.
    Online Segmentation, given an incrementally updated string $x$, incrementally outputs $\Seg_{\preceq} (x; S^*) = \max_{\preceq} \calS_{\Seg}(x; S^*)$.
\end{definition}

\begin{definition}[Online Pretokenization]
    \label{def:online-pretokenization}
    Let $R$ be a total regular expression over $\Sigma$.
    Online Pretokenization, given an incrementally updated list of strings $L$, incrementally outputs $\Split_{\preceq}(L; R) = \max_{\preceq} \calS_{\Split}(L; R)$.
\end{definition}

\begin{definition}[Online Byte-Pair Encoding]
    \label{def:online-bpe}
    Let $M$ be a BPE encoding.
    Online BPE, given an incrementally updated list of strings $L$, incrementally outputs $\Merge_{\preceq}(L; M) =  \max_{\preceq} \calS_{\Merge}(L; M)$.
\end{definition}

As before, we omit $S^*, R, M$ when convenient.
We observe that for a list $L = a_1 \brk \dots \brk a_{s}$, $\Split_{\preceq}(L) = \Split(a_1) \brk \dots \brk \Split(a_{s - 1}) \brk \Split_{\preceq}(a_s)$.
Indeed, any continuation of $L$ has the first $s - 1$ segments identical to $L$ and the $s$-th segment is a continuation of $a_s$.
Similarly, $\Merge_{\preceq}(L) = \Merge(a_1) \circ \dots \circ \Merge(a_{s - 1}) \circ \Merge_{\preceq}(a_{s})$.

We give a linear time algorithm for online tokenization.

\CompositionUB*

Accounting only for $n$ on large inputs, our algorithm runs in $O(n)$ time and $O(\lookback)$ space with $\lookback + O(1)$ delay.

\begin{proof}[Proof of \Cref{thm:linear-online-tokenization-w-delay}]
    As subroutines, we require algorithms $\Segment$, $\Pretokenize$, $\BPEMerge$ with the following guarantees.

    \begin{restatable}[Online Segmentation]{lemma}{OnlineSeg}
        \label{lem:segment-alg}
        There is an algorithm $\Segment$ that computes online segmentation in $O(n)$ time and $O(1)$ space.
        Furthermore, for any $x \in \Sigma^*$, $|\raw(\Seg_{\preceq}(x))| \geq \max(0, |x| - \ell_{\infty}(S^*))$.
    \end{restatable}

    \begin{restatable}[Online Pretokenization (Formal \Cref{lem:pretokenize-informal})]{lemma}{OnlineSplit}
        \label{lem:pretokenize-alg}
        There is an algorithm $\Pretokenize$ that computes online pretokenization in $O(n)$ time and $O(\lookback(R))$ space.
        Furthermore, for any list of segments $L$, $\Split_{\preceq}(L)$ outputs all but the last $\lookback(R)$ characters of $L$.
    \end{restatable}

    \begin{restatable}[Online BPE \citep{jiang2025incrementalbpe}]{lemma}{OnlineBPE}
        \label{lem:bpe-merge-alg}
        There is an algorithm $\BPEMerge$ that in $O(n \log^2 \longest)$ total time and $O(\merges \longest)$ space incrementally outputs $y$ satisfying
        \begin{equation*}
            \Merge(x[:-C]; M) \preceq y \preceq \Merge_{\preceq}(x; M)
        \end{equation*}
        for some $C = C(x, M) \leq (\merges + 1)\longest$.
    \end{restatable}

    First, we argue that the above suffices to prove \Cref{thm:linear-online-tokenization-w-delay}.
    Our algorithm simply outputs 
    \begin{equation*}
        \BPEMerge(\Pretokenize(\Segment(x; S^*); R); M) \preceq \Merge_{\preceq}(\Split_{\preceq}(\Seg_{\preceq}(x))) \text{.}
    \end{equation*}
    Note that each of the three algorithms are online (i.e. handles incremental input and incremental output) so that the algorithm is well defined.
    The running time and space bounds follow trivially, so it suffices to prove correctness.
    Let $x$ be the current input and $\Tok_{\preceq}(x)$ the desired output.
    First, we prove that our algorithm is sound: it outputs a stable prefix.

    \begin{claim}
        \label{clm:composition-is-stable}
        For any $S^*, R, M$ and input string $x$,
        \begin{equation*}
            \Merge_{\preceq}(\Split_{\preceq}(\Seg_{\preceq}(x; S^*); R); M) \preceq \Tok_{\preceq}(x; S^*, R, M) \text{.}
        \end{equation*}
    \end{claim}

    \begin{proof}
        Throughout, we omit $S^*, R, M$ for clarity.
        Let $x_1 = \Seg_{\preceq}(x)$, $x_2 = \Split_{\preceq}(x_1)$, and $x_3 = \Merge_{\preceq}(x_2)$ denote the outputs of our algorithm.
        Fix any continuation $y \in \Sigma^*$.
        Then, by correctness of $\Segment$, $x_1 = \Seg_{\preceq}(x) \preceq \Seg(xy)$ so that $\Seg(xy) = x_1 w$ for some $w$.
        By correctness of $\Pretokenize$, $x_2 \preceq \Split(x_1w) = \Split(\Seg(xy))$.
        Then, $\Split(\Seg(xy)) = x_2 v$ for some $v$.
        By correctness of $\BPEMerge$, we have $x_3 \preceq \Merge(x_2 v) = \Merge(\Split(\Seg(xy))) = \Tok(xy)$.
        Thus, $x_3$ is a stable prefix.
    \end{proof}

    Next, we prove completeness: our algorithm outputs the maximum stable prefix with delay $\lookback(R) + O(\merges \longest)$.
    We require the following lemma, which states that any accepted word must reach an accepting state every $\lookback(R)$ characters.

    \begin{lemma}
        \label{lem:live-lookahead}
        Let $\calA$ be a DFA recognizing $L(R)$.
        For $u \in \Sigma^+$, let $\calA(u)$ denote whether $\calA$ is in an accepting state after reading $u$.
        Let $\last(u) := \max \set{|w| \given w \preceq u\text{, } \calA(w) = 1}$ denote the last accepting state when reading $u$ (possibly including $u$).
        If $\last(u) < |u| - \lookback(R)$, then $\calA(ux) = 0$ for all $x \in \Sigma^*$.
    \end{lemma}

    \begin{proof}
        Suppose for contradiction $\calA(ux) = 1$ for some continuation $x \in \Sigma^*$.
        Let $\lookback := \lookback(R)$.
        Since $R$ is total, we have $\last(u) \geq 1$ so that $|u| > \lookback$.
        Set $y_0 = ux$.
        By \Cref{def:lookback}, for any $|y_i| > \lookback$, there exists some $y_{i + 1} \prec y_i$ such that $|y_{i + 1}| \geq |y_i| - \lookback$ and $\calA(y_{i + 1}) = 1$.
        Let $i^*$ be the largest $i$ such that $u \preceq y_i$ so that $y_{i^* + 1} \prec u$.
        Then, $\last(u) \geq |y_{i^* + 1}| \geq |y_{i^*}| - \lookback \geq |u| - \lookback$, a contradiction.
    \end{proof}

    Let $\lookback := \lookback(R)$.
    Let $\Seg_{\preceq}(x) = a_1 \brk \dots \brk a_{s - 1} \brk a'$ where $x = a_1 a_2 \dots a_{s - 1} a_s$ and $a_s$ (possibly empty) contains no special token.
    We can assume this without loss of generality, as if $x$ ends on a special token, then $a_s = \epsilon$ is empty.
    Note that we can thus assume $a_{s - 1}$ is a special token.
    
    Since $\Seg_{\preceq}(x)$ only holds back characters of $a_s$ if a suffix of $a_s$ is a prefix of a special token in $S^*$, we conclude that $a_s[:-\ell_{\infty}(S^*)] \preceq a' \preceq a_s$ since $a_s[-\ell_{\infty}(S^*):] \not\in S^*$ and therefore is not a prefix of any special token $t \in S^*$.
    Then, we write
    \begin{equation*}
        \Split_{\preceq}(\Seg_{\preceq}(x)) = \Split(a_1) \brk \dots \brk \Split(a_{s - 1}) \brk \Split_{\preceq}(a') = P \brk \Split_{\preceq}(a')
    \end{equation*}
    where $P = \brk_{i = 1}^{s - 1} \Split(a_i)$.
    If $a' = \epsilon$ is empty, then $\Seg_{\preceq}(x) = a_1 \brk \dots \brk a_{s - 1} \brk$ so that 
    \begin{equation*}
        \Merge(\Split(a_1 \brk \dots \brk a_{s - 1} \brk)) \preceq \Tok_{\preceq}(x) \text{.}
    \end{equation*}
    In other words, since $a_{s - 1}$ is a special token, we have
    \begin{equation*}
        \Tok_{\preceq}(a_1 \dots a_{s - 1}) \preceq \Tok(a_1 \dots a_{s - 1}) = \Merge(\Split(a_1 \brk \dots \brk a_{s - 1} \brk)) \preceq \Tok_{\preceq}(x)
    \end{equation*}
    so that the algorithm has delay $\ell_{\infty}(S^*)$, since $a' = \varepsilon$ implies $|a_s| \leq |a'| + \ell_{\infty}(S^*) \leq \ell_{\infty}(S^*)$.
    
    Thus, assume $a' \neq \epsilon$ is not empty.
    Let $\Split_{\preceq}(a') = b_1 \brk \dots \brk b_{t - 1} \brk b'$ where $a' = b_1 \dots b_t$ and $b' \preceq b_t$.
    \Cref{lem:pretokenize-alg} guarantees that $|b'| \geq |b_t| - \lookback(R)$.
    
    Then, our algorithm outputs
    \begin{align*}
        \BPEMerge(P \brk \Split_{\preceq}(a')) &= \Merge(P) \circ \Merge(b_{1}) \circ \dots \circ \Merge(b_{t - 1}) \circ \BPEMerge(b') \\
        &= Q \circ \BPEMerge(b') 
    \end{align*}
    where $Q = \Merge(P) \circ \left(\bigcirc_{i = 1}^{t - 1} \Merge(b_i) \right)$.
    Recall that we proved above
    \begin{equation*}
        |a_1 \dots a_{s - 1} b_1 \dots b_{t - 1} b'| \geq |a_1 \dots a_{s - 1} a'| - \lookback(R) \geq |x| - (\ell_{\infty}(S^*) + \lookback(R)) \text{.}
    \end{equation*}

    We have the following claim.

    \begin{claim}
        \label{clm:coverage}
        Suppose $y \preceq x$ and $z \preceq \Tok_{\preceq}(x)$, and $|z| \geq |y|$.
        Then, $\Tok_{\preceq}(y) \preceq z$.
    \end{claim}

    \begin{proof}[Proof of \Cref{clm:coverage}]
        Note $\Tok_{\preceq}(y) \preceq \Tok_{\preceq}(x)$, so $\Tok_{\preceq}(y), z$ are prefixes of $\Tok_{\preceq}(x)$.
        Thus, if $|z| \geq |y|$, we have the desired claim.
    \end{proof}

    Our algorithm outputs $z = Q \circ \BPEMerge(b') \preceq \Tok_{\preceq}(x)$.
    From \Cref{lem:bpe-merge-alg} we have
    \begin{align*}
        |z| &\geq |a_1 \dots a_{s - 1} b_1 \dots b_{t - 1}| + |b'| - (\merges + 1)\longest \\
        &\geq |x| - \ell_{\infty}(S^*) - \lookback - (\merges +1)\longest \text{.}
    \end{align*}
    Thus, \Cref{clm:coverage} implies the desired delay bound.
    It remains to prove the required lemmas.

    \OnlineSeg*

    \begin{proof}[Proof of \Cref{lem:segment-alg}]
        Consider a trie $\calT$ that recognizes strings in $S^*$.
        This trie has size $O(\ell_1(S^*))$.
        Regard the trie as a deterministic automaton with an additional rejecting sink. 
        Missing child transitions lead to the sink, which loops on every character. Accepting nodes correspond to special tokens. 
        The function $\last(\start, w)$ returns the length of the longest accepting prefix of $w$, or $\bot$ if none exists.
        
        We note some properties of $\Seg_{\preceq}(x)$.
        The input $x$ consists of the characters in the committed output $\Seg_{\preceq}(x)$ and (possibly empty) buffer $u$ such that $x = \raw(\Seg_{\preceq}(x)) \circ u$.
        We note that $u$ is a (non-maximal) trie prefix.
        Otherwise, if it is not a trie prefix, it is not (and cannot be extended to a special token), so that we can commit at least one character of $u$ to output.
        If it is a maximal trie prefix, it is a special token (and cannot be extended to a longer one), so we can commit $s^*(u)$ as a special token.
        Thus, we have 
        \begin{equation*}
            |\raw(\Seg_{\preceq}(x))| \geq |x| - |u| \geq |x| - \ell_{\infty}(S^*) \text{.}
        \end{equation*}
        For some partition $x = y_1 \dots y_s u$ where we may assume without loss of generality that (possibly empty) $y_s$ is not a special token, we have
        \begin{equation}
            \label{eq:seg-prec-output}
            \Seg_{\preceq}(x) = \begin{cases}
                y_1 \brk \dots \brk y_{s} \brk &\text{if $u$ contains a special token as a prefix} \\
                y_1 \brk \dots \brk y_{s} &\otherwise
            \end{cases} \text{.}
        \end{equation}
        Note that each $y_{j}$ is either a special token (corresponding to a trie match) or a segment of alphabet characters. 
        Only the latter is provided as input to pretokenization and BPE in later stages.
        
        In the first case, since $u$ contains a special token as a prefix, any continuation of $x$ contains a special token after $y_{s}$, so the prefix contains $\brk$ after $y_{s}$.
        In the latter case, we do not know if $u$ will extend to a special token.
        In particular, we cannot yet commit any character of $u$, since if $u$ is extended to a match (e.g. $\tilde{u} = u u' \in S^*$) then $\Seg(x u')$ ends with $\dots y_{s} \brk s^*(\tilde{u})$.

        We now give a simple algorithm.
        Intuitively, $\buffer$ will maintain the uncommitted suffix of the input.
        Initialize $\state \gets \start$ to the root of $\calT$, $\last \gets \bot$, and $\buffer \gets \epsilon$.
        Recall that $\delta$ denotes the transition function and $\last$ the last acceptance.
        In the following, reset means $\state \gets \start, \last \gets \bot, \buffer \gets \epsilon$.
        Our algorithm never outputs two consecutive $\brk$; this can be checked by examining the last output.
        On input $x_i$, do the following:
        \begin{mdframed}
            {\bf $\python{SimpleMatch}(x_{i})$}
            \stepnames
            \begin{enumerate}
                \item Update $\state \gets \delta(\state, x_i)$ and push $x_i$ to $\buffer$.

                \item While $\buffer$ is not a trie prefix: 
                \label{step:dead-not-accept}
                
                \begin{enumerate}
                    \item If $\last = \bot$, output $\buffer_{1}$ and update $\buffer \gets \buffer_{> 1}$.
                    
                    \Comment{(Prefix of special token unmatched. Flush first character.)}
                    
                    \item If $\last \neq \bot$, output and pop $\brk s^*(\buffer_1 \dots \buffer_{\last}) \brk$.

                    \Comment{(Backtrack to the last special token. Restore state.)}

                    \item Update $\state \gets \delta(\start, \buffer)$ and $\last \gets \last(\start, \buffer)$.
                \end{enumerate}

                \item If $\state \in \accept$ and $\buffer$ is a maximal trie prefix, output $\brk s^*(\buffer) \brk$. Reset.
    
                \Comment{(Found maximal special token.)}
                
                \item Else if $\state \in \accept$ set $\last \gets |\buffer|$. 
                If $\last \neq \bot$, output $\brk$.
                \label{step:short-special}

                \Comment{(Found special token, not maximal.)}
            \end{enumerate}
        \end{mdframed}

        Observe that if we make the assumption that $S^*$ is prefix-free (as is often the case in practice), the algorithm and ensuing analysis can be simplified by ignoring \Cref{step:short-special}.
        The space of the algorithm is $O(|\buffer|) = O(\ell_{\infty}(S^*))$, since $\buffer$ contains exactly the uncommitted input $u$.
        The time of the $\state \in \accept$ cases is $O(|\buffer|)$, $O(|\buffer|)$ in the maximal case and $O(1)$ otherwise.
        We bound the time of \Cref{step:dead-not-accept}.
        We observe that each iteration of the while loop takes $O(|\buffer|)$ time (e.g. to recompute $\state, \last$).
        Furthermore, there are at most $O(|\buffer|)$ iterations since in each iteration, the size of $\buffer$ decreases by at least $1$.
        Thus, the overall time per update is $O(|\buffer|^2) = O(\ell_{\infty}(S^*)^2)= O(1)$.
        With a slightly tighter analysis, we can achieve $O(n |\buffer|)$ total time, by observing that every step that requires $O(|\buffer|)$ time produces at least one character of output, so there are at most $O(n)$ such steps.

        We now prove correctness, inducting on $x_i$ with the following hypothesis.
        Let $u$ be the uncommitted input after processing $x_{i - 1}$.
        \begin{enumerate}
            \item $\buffer$ contains the uncommitted input.
            \item $\state = \delta(\start, \buffer)$ and $\last = \last(\start, \buffer)$.
            \item The output is $\Seg_{\preceq}(x_{\leq i})$ (satisfying \Cref{eq:seg-prec-output}).
        \end{enumerate}
        The base case is trivial.
        On input $x_i$, either $x_i$ extends $u$ as a prefix of a special token, or it does not.
        In the former case, if $u x_i$ is a maximal special token, our algorithm has $\state \in \accept$ and $\buffer$ is maximal, so we output $\brk s^*(u x_i) \brk$ and reset, satisfying the inductive hypothesis.
        Otherwise, $u x_i$ is not a maximal special token. 
        If it is a special token, we output $\brk$ since $\state \in \accept$, push $x_i$ to $\buffer$ and update $\state, \last$, again satisfying the inductive hypothesis.
        Otherwise, our algorithm pushes $x_i$ to $\buffer$ and updates $\state$, again satisfying the inductive hypothesis.

        If $x_i$ does not extend $u$, our algorithm enters the case where $\buffer$ is not a trie prefix.
        Whenever $\buffer$ is not a trie prefix, we pop some prefix of it (either one character or a prefix of $\buffer$ that is a special token).
        Eventually, $\buffer$ (which is possibly empty) must be a trie prefix.
        If this prefix is maximal, we output it, leaving $\buffer$ empty and satisfying the first condition.
        If not, we retain $\buffer$, again satisfying the first condition, as this is the uncommitted input.
        It is easily verified that whenever $\buffer$ is updated, $\state$ and $\last$ are recomputed appropriately, satisfying the second condition.
        In each iteration, we output any prefix of $\buffer$ that is a special token (that cannot be extended), and also outputting $\brk$ if the $\buffer$ contains a trie prefix after exiting the while loop, maintaining the correctness of the output as desired.
    \end{proof}

    \OnlineSplit*

    We note that in the offline setting, there is a linear time pretokenization algorithm (called maximal munch tokenization) of \citet{reps1998maximalmunch}. 
    In the worst case, their algorithm requires linear space.
    We design an algorithm parameterized by $\lookback := \lookback(R)$.
    Consider a DFA $\calA$ that recognizes exactly $L(R)$ with starting state $\start$.

    Observe that a simple modification of \Cref{lem:segment-alg} gives an algorithm for online pretokenization that runs in $O(n \lookback)$ time and $O(\lookback)$ space:
    Simulate $\calA$.
    On each accepting state, emit the current match; otherwise buffer the current input in memory.
    On each dead state, backtrack to the last accepting state, emit a $\brk$, and simulate $\calA$ on the buffer in memory from the starting state.
    \Cref{lem:live-lookahead} guarantees that a dead state must be encountered within $O(\lookback)$ characters of the last accept, so that we backtrack at most $O(\lookback)$ characters, resulting in $O(n\lookback)$ time.
    While in practice $\lookback$ is not large, in the worst case this can lead to a quadratic time algorithm (e.g. when $\lookback$ is sufficiently large).
    To guarantee linear time in general, we require the following useful data structure:

    \ConcatQueue*

    The data structure is initialized with two empty stacks: the push stack and the pop stack.

    \begin{mdframed}
        $\push(c)$ for character $c \in \Sigma$:
        \begin{enumerate}
            \item Compute $\delta(q, c)$ and $\last(q, c)$ for all $q \in \calQ$.
            \item Let $w_{\push}$ be the top element of the push stack ($\epsilon$ if empty).
            Compute and push
            \begin{align*}
                \delta(q, w_{\push}c) &\gets \delta(\delta(q, w_{\push}), c) \text{,} \\
                \last(q, w_{\push}c) &\gets \begin{cases}
                    \last(q, w_{\push}) & \text{ if } \last(\delta(q, w_{\push}), c) = \bot \\
                    |w_{\push}| + \last(\delta(q, w_{\push}), c) & \otherwise
                \end{cases} \text{.}
            \end{align*}
        \end{enumerate}
    \end{mdframed}

    \begin{mdframed}
        $\pop(p)$ for integer $p \geq 1$:
        \begin{enumerate}
            \item If the pop stack has $k \geq p$ elements, pop $p$ elements from the pop stack.
            \item If the pop stack has $k < p$ elements:
            \begin{enumerate}
                \item Pop $k$ elements from the pop stack. 
                Let the push stack contain $w_1 \dots w_{\ell}$ (with $w_{\ell}$ on top).
                \item If $\ell + k \leq p$, empty both stacks.
                \item If $\ell + k > p$, reverse the elements of the push stack onto the pop stack, so that the pop stack contains $w_{\ell} \dots w_1$ (with $w_1$ on top). 
                For each $i \in [\ell]$ in decreasing order, compute:
                \begin{align*}
                    \delta(q, w_{i} \dots w_{\ell}) &\gets \delta(\delta(q, w_{i}), w_{i + 1} \dots w_{\ell}) \text{,} \\
                    \last(q, w_{i} \dots w_{\ell}) &\gets \max \begin{cases}
                        \last(q, w_{i}) \\
                        1 + \last(\delta(q, w_{i}), w_{i + 1} \dots w_{\ell}) 
                    \end{cases} \text{.}
                \end{align*}
                
                \Comment{(For comparison purposes, if $\last(q, w) = \bot$, then $\bot = - \infty$.)}
                
                Then, pop $p - k$ elements from the pop stack.
            \end{enumerate}
        \end{enumerate}
    \end{mdframed}

    \begin{mdframed}
        $\query(q)$ for state $q \in \calQ$:
        \begin{enumerate}
            \item Let the pop stack contain $w_{j} \dots w_{1}$ (with $w_{1}$ on top) and the push stack contain $w_{j + 1} \dots w_{\ell}$ (with $w_{\ell}$ on top).
            \item Compute:
            \begin{align*}
                    \delta(q, w_{1} \dots w_{\ell}) &\gets \delta(\delta(q, w_{1} \dots w_{j}), w_{j + 1} \dots w_{\ell}) \text{,} \\
                    \last(q, w_{1} \dots w_{\ell}) &\gets \max \begin{cases}
                        \last(q, w_{1} \dots w_{j}) \\
                        j + \last(\delta(q, w_{1} \dots w_{j}), w_{j + 1} \dots w_{\ell}) 
                    \end{cases} \text{.}
                \end{align*}
        \end{enumerate}
    \end{mdframed}

    We prove that the data structure satisfies the desired guarantees.

    \begin{proof}[Proof of \Cref{lem:concat-queue}]
        We maintain the following invariant: every element maintains $\delta$ and $\last$ for the string obtained by concatenating all characters beneath it and itself. 
        This consists of prefixes of the string on the push stack and suffixes of the string on the pop stack.
        Formally,

        \begin{claim}[Below Consistency]
            \label{clm:below-consistency}
            Let the pop stack contain $w_{j} \dots w_{1}$ (with $w_{1}$ on top) and the push stack contain $w_{j + 1} \dots w_{\ell}$ (with $w_{\ell}$ on top).
            Our algorithm maintains for every $q \in \calQ$:
            \begin{enumerate}
                \item (Pop) $\delta(q, w_{i} \dots w_{j})$ and $\last(q, w_{i} \dots w_{j})$ for all $1 \leq i \leq j$,
                \item (Push) $\delta(q, w_{j + 1} \dots w_{i})$ and $\last(q, w_{j + 1} \dots w_{i})$ for all $j + 1 \leq i \leq \ell$.
            \end{enumerate}
        \end{claim}

        \begin{proof}[Proof of \Cref{clm:below-consistency}]
             Note that for any word $w_{1} \dots w_{j} w_{j + 1} \dots w_{\ell}$, we have by the definition of a DFA for all $j \in [\ell]$,
            \begin{equation}
                \label{eq:trans-compose}
                \delta(q, w_{1} \dots w_{\ell}) = \delta(\delta(q, w_{1} \dots w_{j}), w_{j + 1} \dots w_{\ell}) \text{.}
            \end{equation}
            Similarly, for all $j \in [\ell]$,
            \begin{align}
                \label{eq:last-compose}
                \last(q, w_{1} \dots w_{\ell}) &= \max \begin{cases}
                    \last(q, w_{1} \dots w_{j}) \\
                    j + \last(\delta(q, w_{1} \dots w_{j}), w_{j + 1} \dots w_{\ell}) 
                \end{cases} \\
               &=  \begin{cases}
                    \last(q, w_{1} \dots w_{j}) & \text{ if } \last(\delta(q, w_{\leq j}), w_{j + 1} \dots w_{\ell}) = \bot \\
                    j + \last(\delta(q, w_{\leq j}), w_{j + 1} \dots w_{\ell}) & \otherwise
                \end{cases} \text{.} \nonumber
            \end{align}
            The above equality follows as either the DFA reaches an accepting state in the second half of the computation when scanning $w_{j + 1} \dots w_{\ell}$ (in which case the second case triggers) or it does not, in which case it suffices to find an accepting state in the first half when scanning $w_{1} \dots w_{j}$ (which is handled by the first case).

            We now argue that the invariant is maintained on the push and pop stacks.
            This follows directly from \Cref{eq:trans-compose} and \Cref{eq:last-compose} and the update rule of $\push$ and $\pop$.
        \end{proof}

        Given \Cref{clm:below-consistency}, we observe that \Cref{eq:trans-compose} and \Cref{eq:last-compose} ensure that $\query$ returns the correct values.
        We now argue that our data structure has the desired performance.
        Consider a single element: it is pushed onto the push stack once, moved from the push to pop stack (at most) once, and popped from the pop stack (at most) once.
        It then suffices to argue that each operation requires $O(1)$ time.
        For $\push$, the computation of the transitions $\delta$ and $\last$ can be computed in $O(1)$ time given $\delta(q, w_{\push})$ and $\last(q, w_{\push})$.
        When moving from the push stack to the pop stack, we compute $\delta(q, w_{i})$ and $\last(q, w_{i})$ in $O(1)$ time and then query $\delta(\cdot, w_{i + 1} \dots w_{\ell})$ and $\last(\cdot, w_{i + 1} \dots w_{\ell})$ in $O(1)$ time.
        Popping an element of the pop stack requires $O(1)$ time, completing the proof.
    \end{proof}
    
    Using the given data structure, run the following algorithm.
    We begin by initializing an empty $\cqueue$ (\Cref{lem:concat-queue}) and $\state \gets \start$.
    Then, for each character $x_i$ in the input, do the following.
    
    \begin{mdframed}
        {\bf $\FastPretokenization(x_{i})$}
        
        \begin{enumerate}
        \stepnames
            \item Push $\cqueue.\push(x_i)$ and update $\state \gets \delta(\state, x_i)$.
            
            \item If $\state \in \accept$, output $\cqueue.\pop(|\cqueue|)$.
            \label{step:acc-prefix}

            \Comment{(Found a match, so emit the current buffer and reset $\cqueue$.)}

            \item While $\state \in \dead$:

            \label{step:pretok-dead-loop}

            \Comment{(The unresolved suffix cannot continue the current match, so $\brk$ at the last $\accept$.)}
            
            \begin{enumerate}
                \item If $\cqueue = \epsilon$, output $\brk$, set $\state \gets \start$ and exit.
                \item Query $\delta(\start), \last(\start) \gets \cqueue.\query(\start)$.
                \item Output $\brk$ and $\cqueue.\pop(\last(\start))$.
                \item Update $\state \gets \delta(\start)$.
                \label{step:brk-acc-prefix}
            \end{enumerate}
            
        \end{enumerate}
    \end{mdframed}

    We prove that our algorithm satisfies the desired properties.

    \begin{proof}[Proof of \Cref{lem:pretokenize-alg}]
        We begin by analyzing the correctness of our algorithm.
        We proceed by induction, and argue that for all $i \geq 0$, after witnessing $x_{\leq i}$, the algorithm maintains the following:
        \begin{enumerate}[label=\Alph*]
            \item The output is $\Split_{\preceq}(x_{\leq i})$.
            The output $\Split_{\preceq}(x_{\leq i})$ consists of a string consisting of pretokens and separators $\brk$.
            We let $P_i$ denote the output before the last separator in $\Split_{\preceq}(x_{\leq i})$.
            Thus, we observe that the output has one of the following forms
            \begin{itemize}
                \item $P_i \brk x_{b(i) + 1} \dots x_{\ell(i)}$ for $\ell(i) = b(i) + \last(\start, x_{b(i) + 1} \dots x_{i})$, for some $b(i)$ indicating the current ``live" pretoken, 
                \item $P_{i} \brk$, in which case we write $b(i) = \ell(i) = i$.
            \end{itemize}
            \label{it:outer-out-correct}
            
            \item $\cqueue$ maintains $x_{\ell(i) + 1} \dots x_{i}$.
            \label{it:outer-queue-correct}
            
            \item $\state$ maintains $\delta(\start, x_{b(i) + 1} \dots x_{i})$. 
            \label{it:outer-state-correct}
        \end{enumerate}
        We call the above induction the outer induction, which inducts over each input $x_{i}$ provided incrementally to the algorithm.
        The base case is trivial. 
        If $i = 0$, we have $b(0) = \ell(0) = 0$ so the initial output is $\epsilon$, $\cqueue$ is empty, and $\state \gets \start = \delta(\start, \epsilon)$.

        We now proceed to the inductive case.
        Consider input $x_{i}$.
        To prove the outer induction above, we will use an inner induction over the while loop within $\FastPretokenization$. 
        To prove the output of the $i$-th iteration of the outer induction has Properties \ref{it:outer-out-correct}, \ref{it:outer-queue-correct}, \ref{it:outer-state-correct}, we assume the output of the $(i-1)$-th iteration starting with $P_{i-1}\brk$ satisfies Properties \ref{it:outer-out-correct}, \ref{it:outer-queue-correct}, \ref{it:outer-state-correct} after processing $x_{i - 1}$.
        Note that $P_{i - 1} \preceq P_{i}$ since we can only incrementally extend the output of $\Split_{\preceq}(x_{\leq i})$.
        We now claim that after the $j$-th iteration of the while loop, the following hold: 
        \begin{enumerate}
            \item The current output is $P_{i - 1} \brk Q_{j} \preceq \Split_{\preceq}(x_{\leq i})$ where 
            \begin{equation*}
                Q_{j} = x_{b(i - 1) + 1} \dots x_{r(0)} \brk x_{r(0) + 1} \dots x_{r(1)} \brk \dots \brk x_{r(j - 1) + 1} \dots x_{r(j)}
            \end{equation*}
            for $r(j) = r(j - 1) + \last(\start, x_{r(j - 1) + 1} \dots x_{i})$ with the convention $r(-1) = b(i - 1)$. 
            \item $\cqueue$ maintains $x_{r(j) + 1} \dots x_{i}$. 
            Furthermore, $\state$ maintains $\delta(\start, x_{r(j - 1) + 1} \dots x_{i})$.
            \label{it:state-cqueue-correct}
            \item If $\state \not\in \dead$, then the output is $\Split_{\preceq}(x_{\leq i})$.
            \label{it:state-live-output-correct}
            \item $\state \not\in \accept$ if and only if $r(j) < i$.
        \end{enumerate}
        First, we claim that if the inner induction claim holds and the while loop terminates, then the outer induction holds.
        Towards \Cref{it:outer-out-correct}, we observe that the while loop only terminates when $\state \not\in \dead$ and thus \Cref{it:state-live-output-correct} ensures that whenever $\state \not\in \dead$, the output is $\Split_{\preceq}(x_{\leq i})$ as desired.
        Towards \Cref{it:outer-queue-correct} and \Cref{it:outer-state-correct}, we set $b(i) = r(j - 1)$ and $\ell(i) = r(j)$ after the termination of the while loop so that \Cref{it:state-cqueue-correct} ensures the desired properties.
        Thus, it suffices to prove (1) the inner induction holds and (2) that the algorithm does indeed terminate after each input.

        \paragraph{Base Case.}
        We begin with the base case.
        By convention, we say that the $0$-th iteration of the while loop ends after \Cref{step:acc-prefix}, i.e. immediately before the first iteration of the while loop.
        When input $x_i$ arrives, the output is $\Split_{\preceq}(x_{<i}) \preceq \Split_{\preceq}(x_{\leq i})$ by \Cref{clm:inc-output}.

        If $\state \not\in \accept$, then the algorithm does not output or modify $\cqueue$ in \Cref{step:acc-prefix}.
        Furthermore, since
        \begin{equation*}
            \state = \delta(\delta(\start, x_{b(i - 1) + 1} \dots x_{i - 1}), x_i) = \delta(\start, x_{b(i - 1) + 1} \dots x_{i}) \not\in \accept \text{,}
        \end{equation*}
        we have $r(0) \leq \ell(i - 1)$, satisfying the first, second, and fourth conditions.
        
        If $\state \in \accept$, then we have
        \begin{equation*}
            \state = \delta(\start, x_{b(i - 1) + 1} \dots x_{i}) \in \accept \text{.}
        \end{equation*}
        In particular, since for any continuation $y$,
        \begin{align*}
            \Split(x_{\leq i} y) \succeq \Split_{\preceq}(x_{\leq i}) \succeq \Split_{\preceq}(x_{< i}) \succeq P_{i - 1} \brk
        \end{align*}
        we have $\Split(x_{\leq i} y) = P_{i - 1} \brk \Split(x_{b(i - 1) + 1} \dots x_{i} y)$ so that%
        \begin{equation*}
            x_{b(i - 1) + 1} \dots x_{i} = F(x_{b(i - 1) + 1} \dots x_{i}) \preceq F(x_{b(i - 1) + 1} \dots x_{i} y) \preceq \Split(x_{b(i - 1) + 1} \dots x_{i} y) 
        \end{equation*}
        where $F(x)$ represents the first match of $x$ (see \Cref{def:pretokenization}).
        Thus, the first, second, and fourth conditions are satisfied with $r(0) = i$.

        It remains to check the third condition.
        Suppose $\state \not\in \dead$, we claim the output is $\Split_{\preceq}(x_{\leq i})$.
        Since $\state \not\in \dead$, there is some continuation $|y| > 0$ where 
        \begin{equation*}
            \delta(\start, x_{b(i - 1) + 1} \dots x_{i} y) = \delta(\delta(\start, x_{b(i - 1) + 1} \dots x_{i}), y) = \delta(\state, y) \in \accept \text{.}
        \end{equation*}
        If $\state \not\in \accept$, then above we had $r(0) = \ell(i - 1) < i$,
        \begin{align*}
            \Split(x_{\leq i}) &= P_{i - 1} \brk x_{b(i - 1) + 1} \dots x_{\ell(i - 1)} \brk \Split(x_{\ell(i - 1) + 1} \dots x_{i}) \\
            \Split(x_{\leq i} y) &= P_{i - 1} \brk x_{b(i - 1) + 1} \dots x_{i} y
        \end{align*}
        so that $\Split_{\preceq}(x_{\leq i}) \preceq  P_{i - 1} \brk x_{b(i - 1) + 1} \dots x_{\ell(i - 1)} = P_{i - 1} \brk x_{b(i - 1) + 1} \dots x_{r(0)}$.

        If $\state \in \accept$, after \Cref{step:acc-prefix} the output is $\Split(x_{\leq i})$ since we output
        \begin{equation*}
            P_{i - 1} \brk x_{b(i - 1) + 1} \dots x_{i} = P_{i - 1} \brk F(x_{b(i - 1) + 1} \dots x_{i}) = \Split(x_{\leq i}) \text{.}
        \end{equation*}
        We claim this is $\Split_{\preceq}(x_{\leq i})$ as
        \begin{align*}
            \Split(x_{\leq i}) &= P_{i - 1} \brk x_{b(i - 1) + 1} \dots x_{i} \\
            \Split(x_{\leq i} y) &=  P_{i - 1} \brk x_{b(i - 1) + 1} \dots x_{i} y
        \end{align*}
        have a maximum common prefix of $P_{i - 1} \brk x_{b(i - 1) + 1} \dots x_{i}$.
        Thus, in both cases we have that the output is $\Split_{\preceq}(x_{\leq i})$.

        \paragraph{Inductive Case.}
        Now, let $j \geq 1$.
        We check that all conditions are maintained through each iteration of the while loop.
        From the second condition, we have $\cqueue$ maintains $x_{r(j - 1) + 1} \dots x_{i}$ so that we compute
        \begin{equation*}
            \delta(\start) = \delta(\start, x_{r(j - 1) + 1} \dots x_{i}) \quad \text{    and     } \quad \last(\start) = \last(\start, x_{r(j - 1) + 1} \dots x_{i}) \text{.}
        \end{equation*}

        Let $\delta_0$ denote $\state$ at the end of the previous iteration.
        If $\delta_0 \in \accept$, the fourth condition implies that $r(j - 1) = i$ and the second implies $\cqueue$ maintains $\epsilon$.
        Furthermore, since $\delta_0 \in \dead$ (otherwise this iteration would never have run), then for every continuation $y$
        \begin{equation*}
            \Split(x_{\leq i} y) = \Split(x_{\leq i}) \brk \Split(y)
        \end{equation*}
        so that $\Split_{\preceq} (x_{\leq i}) = P_{i - 1} \brk Q_{j - 1} \brk$.
        Our algorithm outputs exactly $\brk$ in the $j$-th iteration.
        Furthermore, $\state \gets \delta(\start) = \start \not\in \dead$, so the while loop terminates, completing the inductive step.

        Thus, in the following assume $\delta_0 \not\in \accept$, so that $r(j - 1) < i$ and $r(j) > r(j-1)$.
        By induction, we output $P_{i - 1} \brk Q_{j}$.
        We argue $P_{i - 1} \brk Q_{j} \preceq \Split_{\preceq}(x_{\leq i})$.
        Since $\delta_0 \in \dead$, we have for any continuation $y$,
        \begin{equation*}
            \Split(x_{\leq i} y) \succeq P_{i - 1} \brk Q_{j - 1} \brk \Split(x_{r(j - 1) + 1} \dots x_{i} y) \text{.}
        \end{equation*}
        Since $x_{r(j - 1) + 1} \dots x_{r(j)} = F(x_{r(j - 1) + 1} \dots x_{i})$, we have
        \begin{equation*}
            P_{i - 1} \brk Q_{j} = P_{i - 1} \brk Q_{j - 1} \brk F(x_{r(j - 1) + 1} \dots x_{i}) \preceq \Split(x_{\leq i} y) \text{.}
        \end{equation*}
        The algorithm pops and outputs $x_{r(j - 1) + 1} \dots x_{r(j)}$, satisfying the first and second conditions.
        If $\delta(\start) \not\in \accept$, we have $r(j) < i$, satisfying the fourth condition.

        We verify the third condition.
        Suppose $\state \not\in \dead$, so there is some continuation $|y| > 0$ where 
        \begin{equation*}
            \delta(\start, x_{r(j - 1) + 1} \dots x_{i} y) = \delta(\delta(\start, x_{r(j - 1) + 1} \dots x_{i}), y) = \delta(\state, y) \in \accept \text{.}
        \end{equation*}
        Then, we have
        \begin{align*}
            \Split(x_{\leq i}) &= P_{i - 1} \brk Q_{j} \brk \Split(x_{r(j) + 1} \dots x_{i}) \\
            \Split(x_{\leq i} y) &= P_{i - 1} \brk Q_{j - 1} \brk x_{r(j - 1) + 1} \dots x_{i} y 
        \end{align*}
        so that the maximum common prefix is at most $P_{i - 1} \brk Q_{j - 1} \brk x_{r(j - 1) + 1} \dots x_{r(j)} = P_{i - 1} \brk Q_{j}$, as desired.
        
        \paragraph{Termination Condition.}
        We now argue that the while loop terminates.
        Above we have argued that whenever $\state \in \accept$, the while loop terminates in the next iteration.
        On the other hand, whenever $\state \not\in \accept$, we have $r(j) < i$, and the first condition ensures that $r(j) > r(j - 1)$ since $\calA$ is total so that $\last(\start, x_{r(j - 1) + 1} \dots x_{i}) > 0$, as every non-empty string has a non-empty prefix that matches the regular expression.

        \paragraph{Performance.}
        Finally, we bound the amortized running time of our algorithm.
        From \Cref{lem:concat-queue}, it suffices to argue that the running time of our algorithm is dominated by operations of $\cqueue$ (up to constant factors).
        To do so, we can simply charge the check of $\delta(\start) \in \dead$ and outputting $\brk$ in each iteration of the while loop to query from $\cqueue$.
        To bound the space, we use \Cref{lem:live-lookahead}.
        Suppose $|\cqueue| > \lookback$ maintains more than $\lookback$ elements e.g. $x_{i - \lookback} \dots x_{i}$.
        Recall $\ell(i) = b(i) + \last(\start, x_{b(i) + 1} \dots x_{i})$ is the last output index, corresponding to the last index where $\state \in \accept$ on the current attempted match.
        Then, we have $\ell(i) < i - \lookback$ since we pop from $\cqueue$ and output whenever an accepting state is encountered.
        In particular, \Cref{lem:live-lookahead} ensures that $\delta(\start,x_{b(i)+1}\dots x_i)\in\dead$, so that we enter the while loop and output $\brk$.
        Since $\calA$ is total, on the new match we have $\ell' \geq \ell(i) + 1$, so that the algorithm outputs $\brk$ and pops (at least) $x_{\ell(i) + 1}$ from $\cqueue$.
        Thus, whenever $|\cqueue| > \lookback$, the algorithm pops at least one element from $\cqueue$ in each iteration, ensuring that the size of the string maintained by $\cqueue$ does not grow.
        Thus, $\cqueue$ requires space at most $O(\lookback)$.

        \paragraph{Extension to NFA}
        So far, we have assumed that $\calA$ is a DFA.
        However, given a regex $R$, a DFA might have size $\exp(|R|)$, while an NFA is promised to have size $O(|R|)$.
        We thus describe how to extend our algorithm to handle NFAs.
        It suffices to describe how to modify the data structure given by \Cref{lem:concat-queue} to handle queries on NFAs.
        Note that we only ever query the data structure with the $\start$ state.
        In addition, in the algorithm of \Cref{lem:pretokenize-alg}, we modify $\state \subseteq \calQ$ to be a subset of states and in the first step, modify the update rule $\state \gets \cup_{q \in \state} \delta(q, x_i)$.
        We say a subset $\state$ is $\dead$ if $\state \subseteq \dead$.

        We modify \Cref{lem:concat-queue} as follows.
        For each state $q$, we maintain $\delta(q, w) \subseteq \calQ$ which is now a subset of states.
        For any $w = u v$, we have the identities (analogous to \Cref{eq:trans-compose} and \Cref{eq:last-compose}):
        \begin{align*}
            \delta(q, w) &= \bigcup_{a \in \delta(q, u)} \delta(a, v) \\
            \last(q, w) &= \max \set{\last(q, u), |u| + \max_{a \in \delta(q, u)} \last(a, v)} \text{.}
        \end{align*}
        Above, $\max$ of an empty set or $\bot$ is $-\infty$.
        Given $\delta(\cdot, u), \delta(\cdot, v)$ and $\last(\cdot, u), \last(\cdot, v)$ note that $\delta(\cdot, w), \last(\cdot, w)$ can be computed in $\poly(|\calQ|) = O(1)$ time.
        Thus, the amortized time of each operation of the $\cqueue$ data structure is still $O(1)$, as desired.

        The analyses of correctness and performance follow identically.
        Note that all of the modifications described incur a $\poly(|R|) = O(1)$ blowup in runtime.
        This concludes the proof of \Cref{lem:pretokenize-alg}.
    \end{proof}

    Finally, for online BPE we use the algorithm of \citet{jiang2025incrementalbpe}.

    \OnlineBPE*

    \begin{proof}[Proof of \Cref{lem:bpe-merge-alg}]
        We take the algorithm of \citet{jiang2025incrementalbpe} (see their eager output algorithm and accompanying analysis).
        Their algorithm has time $O(n \log^2 \longest)$ and space $O(\merges \longest)$.
        While they do not explicitly analyze the delay of their algorithm (indeed it does not compute zero-delay online BPE, see \Cref{clm:online-bpe-delay}), we show that their algorithm has delay at most $(\merges +1)\longest$.
        
        Let $d(x) \leq \longest$ denote the depth of the \citet{aho1975efficient} automaton state after reading $x$.
        The \cite{jiang2025incrementalbpe} algorithm outputs the longest common prefix of $\Merge(x[:-j])$ for $0 \leq j \leq d(x)$.
        Every input in this common-prefix computation extends $x[:-\longest]$ so that $\Merge_{\preceq}(x[:-\longest])$ is a prefix of their output.
        In particular, $\Merge_{\preceq}$ on input $x$ tokenizes all but a $\merges \longest$ length suffix of $x$.
    
        \begin{lemma}[Theorem 2 of \citet{berglund2023formalizing}]
            \label{lem:merge-prec-progress}
            For any merge list $M$ and input $x \in \Sigma^*$, we have
            \begin{equation*}
                \Merge(x[:-C]) \preceq \Merge_{\preceq}(x)
            \end{equation*}
            where $C := C(x, M) \leq \merges \longest$ and $\longest = \longest(M)$.
        \end{lemma}

        The bound on $C$ is tight up to constant factors (see \Cref{thm:bpe-space-lb} for details).
        We remark that while $C(x, M)$ may vary on input $x$, it is always upper bounded by $\merges \longest$.
        By \Cref{lem:merge-prec-progress}, the output emits all but at most $(\merges + 1) \longest$ characters.
        Alternatively, we can also use the algorithm of \citet{mamouras2026efficient}, which obtains similar guarantees.
    \end{proof}
    
    This concludes the proof of \Cref{thm:linear-online-tokenization-w-delay}.
\end{proof}
\subsection{Space Lower Bounds and Limitations of Composition}

We give a lower bound showing that the space usage of \Cref{thm:linear-online-tokenization-w-delay} is optimal.

\SpaceLB*

\begin{proof}
    Fix the alphabet $\Sigma = \set{a, b, c}$.
    Fix a function $f$.
    Let $S^* = \emptyset$, 
    \begin{equation*}
        R = (aa \mid bb)^f c \mid a \mid b \mid c
    \end{equation*}
    and $M = [aa \mapsto A, bb \mapsto B]$.
    
    First, we claim that $\lookback(R) = 2f$.
    Consider any string of the form $(aa \mid bb)^{f} c \in L(R)$.
    Then, the only proper
    prefixes that are in $L(R)$ are $a, b$ which are length $2f$ shorter than the original string. 
    Thus, we have $\lookback(R) \geq 2f$.
    To prove the upper bound, we observe that all singletons are in the language, and the longest word in the language is length $2f + 1$.
    
    We now show that online tokenization requires $\Omega(\lookback(R))$ space. 
    Consider the following inputs, where $x$ is a string $(aa \mid bb)^{f}$.
    Upon seeing this, there are (at least) two possibilities.
    The input ends, in which case, $\Split$ partitions $x$ into single characters, so that $\Tok(x) = x$.
    On the other hand, if the next character is a $c$ so that the entire string $xc$ becomes one pretoken, we apply the merge rules to the consecutive $aa$ and $bb$ pairs, so the output is $\Tok(xc) \in (A \mid B)^{f} c$.
    In particular, since the two outputs have no common prefix we have $\Tok_{\preceq}(x) = \epsilon$.
    When the next character is a $c$ and the input ends, the algorithm must output a string in $(A \mid B)^{f} c$.
    Note that the output is purely a function of the internal state of the algorithm, since the final character $c$ is constant across all $2^{f}$ possible inputs. 
    Since there are $2^{f}$ such strings, the algorithm requires $\Omega(f)$ space, as desired.
\end{proof}

Using similar ideas, we show that our space dependence on $\merges$ and the bound in \Cref{lem:merge-prec-progress} are both optimal.

\begin{theorem}[BPE Space Bound]
    \label{thm:bpe-space-lb}
    For every pair of functions $f, g$ with $f = \Omega(\log g)$, there exists $M$ where the size of the merge list $\merges  = \Theta(f)$ and $\longest = \Theta(g)$ 
    and online BPE merging (and therefore online tokenization) requires $\Omega(\merges)$ bits of space.
    Furthermore, there exists some $x \in \Sigma^*$ where 
    \begin{equation*}
        \Merge_{\preceq}(x) \prec \Merge(x[:-\Theta(\merges \longest)]) \text{.}
    \end{equation*}
\end{theorem}

\begin{proof}[Proof of \Cref{thm:bpe-space-lb}]
    We give an instance of BPE, which can easily be extended to a tokenization instance by setting $S^* = \emptyset$ and $R = \Sigma^*$.
    Let $t = f$.
    Let $\Sigma = \set{a, b, c_0, c_1, \dots, c_{t}}$.
    We build the list of merges as follows, in order of priority.
    For $s = \lceil \log g \rceil$ we define
    \begin{align*}
        aa \mapsto A_1, A_1A_1 \mapsto A_2, \dots, A_{s - 1}A_{s - 1} \mapsto A \\
        bb \mapsto B_1, B_1B_1 \mapsto B_2, \dots, B_{s - 1}B_{s - 1} \mapsto B
    \end{align*}
    so that with $L = 2^{s}$ we have $\Merge(a^{L}) = A$ and $\Merge(b^{L}) = B$.
    Next, define the following rules, where the priority order is given first by row and then left-to-right in each row:
    \begin{align*}
        & \quad &A c_t \mapsto R_t \\
        &c_{t - 1} A \mapsto L_{t - 1} \quad &A c_{t - 1} \mapsto R_{t - 1} \\
        &\dots \quad &\dots \\
        &c_{1} A \mapsto L_{1} \quad &A c_1 \mapsto R_1 \\
        &c_{0} A \mapsto L_{0} \text{.}
    \end{align*}
    Analogously, we define $B c_i \mapsto R_i'$ and $c_iB \mapsto L_i'$, placing each immediately after its corresponding $A$-rule in the above priority order.

    We now describe the inputs.
    For any $v \in \set{a, b}^{t}$, we can build the input
    \begin{equation*}
        x = c_0 v_{1}^{L} c_1 \dots c_{t - 1} v_{t}^{L} \text{.}
    \end{equation*}
    Where $v=(v_1,\ldots,v_t)$. For $i \in [t]$, we define $X_{i} = L_{i - 1}$ if $v_i = a$ and $L_{i - 1}'$ otherwise and $Y_i = R_{i}$ if $v_i = a$ and $R_i'$ otherwise.
    Then, consider the $\epsilon$ (empty) continuation and continuation $y = c_t$. We have
    \begin{align*}
        \Merge(x) &= X_1 \dots X_{t} \\
        \Merge(xy) &= c_0 Y_1 \dots Y_t \text{.}
    \end{align*}
    Thus, $\Merge_{\preceq}(x) = \epsilon \prec c_0 = \Merge(x_{\leq 1}) = \Merge(x[:-\Theta(fg)])$.
    Note $|x| = t (2^{s} + 1) = \Theta(f g)$.
    As in \Cref{thm:online-tokenization-space-lb}, after viewing $x$, the output is $\Merge_{\preceq}(x) = \epsilon$, but the algorithm must produce $2^{t}$ distinct outputs after receiving the continuation $y$.
    Then, the number of bits of space is at least $\Omega(t) = \Omega(f)$.
\end{proof}

Next, we can extend the example of \Cref{thm:bpe-space-lb} to show that the Online BPE (eager output) algorithm of \citet{jiang2025incrementalbpe} in fact incurs delay $\Omega(\longest)$.
Since in our application, delay $O(\merges \longest)$ suffices (and is optimal), we leave as an open question to develop an online BPE algorithm with zero delay.

\begin{claim}
    \label{clm:online-bpe-delay}
    For every function $g$, there exists $M$ where $\longest = \Theta(g)$ and the online BPE algorithm of \citet{jiang2025incrementalbpe} incurs delay $\Omega(\longest)$.
\end{claim}

We begin by recapping the \citet{aho1975efficient} automaton, a crucial ingredient in describing the \citet{jiang2025incrementalbpe} algorithm.
For any node $z$ in the trie, define $\fail(z)$ to be the longest proper suffix of $z$ that is a trie prefix.
Note $\fail(z)$ could be the root corresponding to $\epsilon$.
We use strings to denote their corresponding nodes in the trie.
Aho and Corasick \citep{aho1975efficient} give an efficient linear time algorithm for computing $\fail(z)$ for all $z \in V(\calT)$ and the Aho-Corasick automaton, denoted $\calA(\calT)$, which has states $V(\calT)$ and transitions:
\begin{equation*}
    \delta(z, c) = \begin{cases}
        zc & zc \in V(\calT) \\
        \delta(\fail(z), c) & \otherwise \text{.}
    \end{cases}
\end{equation*}
That is, transition to the child if it exists, otherwise find the longest suffix of $zc$ (possibly empty) that is a prefix of the trie.
The \emph{depth} of a state is the depth of the node in the trie.

\begin{proof}[Proof of \Cref{clm:online-bpe-delay}]
    Let $\Sigma = \set{a}$ and $M = \set{aa \mapsto A^{1}, A^{1}A^{1} \mapsto A^{2}, \dots, A^{k}A^{k} \mapsto A^{k + 1}}$ for $k = \ceil{\log g}$ so that $2 g \leq \longest = 2^{k + 1} \leq 4 g$.
    Consider the input $x = a^{2 \longest - 1}$ so that $\Merge_{\preceq}(x) = A^{k + 1}$.
    Indeed, we have $\Merge_{\preceq}(x_{\leq \longest}) = A^{k + 1}$.
    
    Note that the depth of the Aho-Corasick automaton on the trie consisting of the strings corresponding to all tokens $\set{a, A^{1}, \dots, A^{k + 1}}$ after scanning $x$ is $\longest$.
    Thus, the \citet{jiang2025incrementalbpe} algorithm outputs the longest common prefix of $\Merge(x), \dots, \Merge(x[:-\longest])$.
    In particular,
    \begin{align*}
        \Merge(x) &= A^{k + 1} A^{k} \dots A^{1} a \\
        \Merge(x[:-\longest]) &= A^{k} \dots A^{1} a \text{.}
    \end{align*}
    so the longest common prefix is $\epsilon$, the empty output.
    This fails the requirement for the algorithm to have delay $\longest - 1$, thus establishing that the delay is $\Omega(\longest)$.
\end{proof}

So far, we have given an algorithm that is essentially optimal in both time and space.
However, the worst-case delay of the algorithm is still $\Omega(\lookback)$.
On the other hand, \Cref{thm:np-delay-lb} does not rule out efficient algorithms with delay $O(\merges \longest)$, which can be much smaller than $\lookback$.
Below, we show that this is an artifact of composition: naively composing online algorithms for $\Segment$, $\Split$, and $\Merge$ cannot achieve online tokenization with delay less than lookback.

\CompositionDelayLB*

\begin{proof}[Proof of \Cref{thm:composition-delay-lb}]
    Let $\Sigma = \set{a, b}$.
    For any function $f$, we provide an example where composition of segmenting, pretokenization, and merging alone achieve delay at least $f - 1 = \lookback(R) - 1$.
    Let $S^* = \emptyset$, $R = a^{f}b \mid a \mid b$ and $M = [bb \mapsto B]$.
    Observe that the lookback is $f$ since the only strict prefix of $a^{f} b$ that is matched by the regular expression $R$ is $a$, so the lookback is at least $f$.
    To upper bound the lookback, note singletons are accepted and the longest pretoken has length $f + 1$.
    
    Consider the input $a^{f}$.
    Note $\Tok_{\preceq}(a^{f}) = a^{f}$ since no pretoken contains $bb$ as a substring and no merge ever occurs.
    However, $\Split_{\preceq}(a^{f}) \preceq a$ since it must be a common prefix of
    \begin{align*}
        \Split(a^{f}) &= a \brk a \brk \dots \brk a \\
        \Split(a^{f} b) &= a^{f} b \brk \text{.}       
    \end{align*}
    Thus, $\Merge_{\preceq}(\Split_{\preceq}(a^{f})) \preceq \Merge_{\preceq}(a) = a \prec aa = \Tok_{\preceq}(aa)$.
    In particular, the delay is at least $f - 1$, as desired.
\end{proof}

While the above example is not hard for tokenization (in fact, $\Tok_{\preceq}(x) = \Tok(x) = x$ for any $x$), it shows that any algorithm achieving $o(\lookback)$ delay against arbitrary pretokenization splitting and BPE merging cannot naively compose algorithms for segmenting, pretokenization, and BPE merging.

%% file: tokenization_lb.tex
\subsection{Exact Online Tokenization}

We begin with a conditional lower bound showing that it is likely hard to design an efficient algorithm for online tokenization without delay.

\OnlineLB*

\begin{proof}[Proof of \Cref{thm:np-delay-lb}]
    We begin with the lower bound against algorithms for exact online tokenization with zero delay.
    Suppose we have some algorithm running in time $\poly(n, \ell_1(S^*), R, \merges )$.
    We will construct a reduction from $\kSAT{3}$ on $d$ variables.
    We construct the following tokenization instance.
    There are $O(d^3)$ possible clauses of width $3$ on $d$ variables.
    Let $\calC$ be the set of all such clauses.
    Define $\Sigma_{\calC} = \set{\ell_{c}}_{c \in \calC}$.
    The alphabet consists of
    \begin{equation*}
        \Sigma = \set{a_1^{0}, a_{1}^{1}, b_{1}, \dots, a_{d}^{0}, a_{d}^{1}, b_{d}} \cup \set{\alpha, \beta, \gamma, \delta} \cup \Sigma_{\calC} \text{.}
    \end{equation*}
    Let $S^* = \emptyset$.
    For every bit-string $x \in \set{0, 1}^{d}$, define
    \begin{equation*}
        E(x) = a_{1}^{x_1} b_1 \dots a_{d}^{x_{d}} b_d \text{.}
    \end{equation*}
    For each clause $c$, let $\calE_{c} = \set{E(x) \given c(x) = 0}$ be the set of assignments on $d$ variables that falsify $c$. 
    Observe that $\calE_{c}$ can be recognized by a DFA of size $O(d)$: there is $1$ assignment to the three variables in $c$ that falsifies $c$, and the remaining variables are unconstrained.
    Similarly, $\calE = \set{E(x) \given x \in \set{0, 1}^{d}}$ is the set of all assignments and can be recognized by a DFA of size $O(d)$.
    Consider the regular expression
    \begin{equation*}
        R_{c} = \ell_{c} \Sigma_{\calC}^* \alpha \beta \calE_{c} \gamma 
    \end{equation*}
    where we abuse notation and use $\calE_{c}$ to denote the regular expression that recognizes $\calE_{c}$.
    Let $R_{\calC}$ recognize the alternation of all $R_{c}$ for $c \in \calC$.
    Then, the full regular expression is
    \begin{equation*}
        R = R_{\calC} \mid \alpha \beta \calE \gamma \delta \mid \Sigma \text{.}
    \end{equation*}
    Note that $\Sigma$ is the alternation of all single characters, and ensures that $R$ is total.
    Observe also that all rules that match strings of length $>1$ have distinct starting characters, so that any longest match automatically chooses the match of highest priority.
    In particular, longest match already prioritizes matches of length $>1$ over single character matches.
    Thus, $R$ can be recognized by a DFA of size $O(d^4)$.

    The merge rules are described below in order of priority:
    \begin{align*}
        &(\gamma, \delta) \mapsto \Gamma \\
        &(b_{d}, \gamma) \mapsto B_{d} \\
        &(a_{d}^{0}, b_{d}) \mapsto A_d^{0} \quad &(a_d^{1}, b_d) \mapsto A_d^{1} \\
        &(b_{d - 1}, a_{d}^{0}) \mapsto B_{d - 1}^{0} \quad &(b_{d - 1}, a_{d}^{1}) \mapsto B_{d - 1}^{1} \\
        &(a_{d - 1}^{0}, b_{d - 1}) \mapsto A_{d - 1}^{0} \quad &(a_{d - 1}^{1}, b_{d - 1}) \mapsto A_{d - 1}^{1} \\
        &\vdots \\
        &(b_{1}, a_{2}^{0}) \mapsto B_{1}^{0} \quad &(b_{1}, a_{2}^{1}) \mapsto B_{1}^{1} \\
        &(a_{1}^{0}, b_{1}) \mapsto A_1^{0} \quad &(a_1^{1}, b_1) \mapsto A_1^{1} \\
        &(\beta, a_1^{0}) \mapsto \Delta^{0} \quad &(\beta, a_1^{1}) \mapsto \Delta^{1} \\
        &(\alpha, \beta) \mapsto \Phi \text{.}
    \end{align*}
    Note that $\merges = 4d + 3$.
    Thus, in the preprocessing phase, the assumed algorithm for online pretokenization given $S^*, R, M$ requires $\poly(d)$ time.
    
    Consider now a $3$-CNF formula $\psi = c_{1} \wedge \dots \wedge c_{k}$ on $d$ variables and $k = \poly(d)$ clauses.
    Construct the input
    \begin{equation*}
        x_{\psi} = \ell_{c_1} \ell_{c_2} \dots \ell_{c_k} \alpha \text{.}
    \end{equation*}
    Note that $|x_{\psi}| = \poly(d)$.
    We claim that $\Tok_{\preceq}(x_{\psi})$ decides whether $\psi$ is satisfiable.

    \begin{claim}
        \label{clm:sat-tok-equiv}
        For any $3$-CNF formula $\psi$, we have the following:
        \begin{enumerate}
            \item (Completeness) If $\psi \in \SAT$, then $\Tok_{\preceq}(x_{\psi}) = \ell_{c_1} \dots \ell_{c_{k}}$.
            \item (Soundness) If $\psi \not\in \SAT$, then $\Tok_{\preceq}(x_{\psi}) = \ell_{c_1} \dots \ell_{c_{k}} \alpha$.
        \end{enumerate}
    \end{claim}

    Thus, if an algorithm computed online tokenization (without delay) in $\poly(n, \ell_1(S^*), R, \merges )$ time, we would have an algorithm that determines satisfiability of $\poly(d)$ size formulas in $\poly(d)$ time.
    In particular, since $\kSAT{3}$ is $\NP$-complete, we would have $\P = \NP$, a contradiction.
    It remains to prove \Cref{clm:sat-tok-equiv}.

    \begin{proof}[Proof of \Cref{clm:sat-tok-equiv}]
        First, we argue $\Tok_{\preceq}(\ell_{c_1} \ell_{c_2} \dots \ell_{c_k}) = \ell_{c_1} \ell_{c_2} \dots \ell_{c_k}$.
        This follows simply as $\ell_{c}$ are not involved in any merge rules, and thus are output regardless of the choice of continuation.
        Next, note that $\Tok(x_{\psi}) = x_{\psi}$ since every pretoken that contains more than one character ends with either $\gamma$ or $\delta$, which does not occur in $x_{\psi}$.
        Thus, $\Split$ splits $x_{\psi}$ into single characters, and no merges occur.
        
        Suppose $\psi \in \SAT$ is satisfiable.
        Let $z \in \set{0, 1}^{d}$ be a satisfying assignment so $\psi(z) = 1$.
        Consider the continuation $\beta E(z) \gamma \delta$.
        Then, we claim that no substring of $x_{\psi} \beta E(z) \gamma \delta$ matches $R_{c}$ for any $c \in \calC$.
        Suppose for contradiction some substring of $x_{\psi} \beta E(z) \gamma \delta$ matches $R_{c}$ for some $c$.
        Since there is only one $\alpha, \beta, \gamma$ in the input, the match must take the form $\ell_{c_{i}} \dots \ell_{c_{k}} \alpha \beta E(z) \gamma$ for some $i \in [k]$.
        We thus have $E(z) \in \calE_{c_{i}}$ for some clause $c_i \in \psi$, a contradiction since $z$ is a satisfying assignment.

        Since $R_{\calC}$ does not match any substring of $x_{\psi} \beta E(z) \gamma \delta$, we have
        \begin{equation*}
            \Tok(x_{\psi} \beta E(z) \gamma \delta) = \Merge(\ell_{c_1} \brk \dots \brk \ell_{c_k} \brk \alpha \beta E(z) \gamma \delta) = \ell_{c_1} \ell_{c_2} \dots \ell_{c_k} \Phi A_1^{z_1} \dots A_{d}^{z_{d}} \Gamma \text{.}
        \end{equation*}
        Thus, $\Tok_{\preceq}(x_{\psi}) \preceq \ell_{c_1} \ell_{c_2} \dots \ell_{c_k}$, establishing equality.

        On the other hand, suppose $\psi \not\in \SAT$ is unsatisfiable.
        Our goal is to show that for any continuation $w \in \Sigma^*$, $x_{\psi} \preceq \Tok(x_{\psi} w)$.
        Observe that $\alpha$ appears in a pretoken only when followed by a $\beta$, a string in $\calE$ and a $\gamma$.
        Thus, it suffices to only consider continuations of the form $\beta E(z) \gamma w$ where $z \in \set{0, 1}^{d}$ and $w \in \Sigma^*$.
        Otherwise, no rule $R_{c}$ matches a pretoken of length $> 1$ containing $\alpha$, so $\alpha$ is emitted as its own token.
        Since $\psi$ is not satisfiable, we have $\psi(z) = 0$, or $c_i(z) = 0$ for some $i \in [k]$.
        Take the smallest such $i$ with $E(z) \in \calE_{c_{i}}$ and since $R_{c_{i}}$ gets matched first, we have
        \begin{align*}
            \Tok(x_{\psi} \beta E(z) \gamma w) &= \Merge(\ell_{c_1} \brk \dots \brk \ell_{c_{i - 1}} \brk \ell_{c_{i}} \dots \ell_{c_{k}} \alpha \beta E(z) \gamma) \circ \Tok(w) \\
            &= \ell_{c_1} \dots \ell_{c_{i - 1}} \circ \Merge(\ell_{c_{i}} \dots \ell_{c_{k}} \alpha \beta E(z) \gamma) \circ \Tok(w) \\
            &= \ell_{c_1} \dots \ell_{c_{k}} \alpha \Delta^{z_1} B_{1}^{z_2} \dots B_{d - 1}^{z_{d}} B_{d} \circ \Tok(w) \text{.}
        \end{align*}
        Thus, we have $\Tok_{\preceq}(x_{\psi}) = \ell_{c_1} \dots \ell_{c_{k}} \alpha$, as desired.
    \end{proof}

    \paragraph{Extension to $o(\merges \longest)$ Delay.}
    Now, we describe how to modify the lower bound against algorithms with small delay.
    We add a character $\zeta$ to the alphabet so that
    \begin{equation*}
        \Sigma = \set{a_{1}^{0}, a_{1}^{1}, b_{1}, \dots, a_{d}^{0}, a_{d}^{1}, b_{d}} \cup \set{\alpha, \beta, \gamma, \delta, \zeta} \cup \Sigma_{\calC} \text{.}
    \end{equation*}
    The special token set $S^* = \emptyset$ is identical.
    Note that in the previous construction, $\longest = \longest(M) = 2$.
    The merge rules now begin with
    \begin{equation*}
        (\zeta, \zeta) \mapsto \zeta^{(1)}, (\zeta^{(1)}, \zeta^{(1)}) \mapsto \zeta^{(2)}, \dots, (\zeta^{(p - 1)}, \zeta^{(p - 1)}) \mapsto \zeta^{(p)} \text{.}
    \end{equation*}
    Then, we have for every $1 \leq i \leq d$, the merge rule $(b_{i}, \zeta^{(p)}) \mapsto \beta_{i}$.
    The remaining merge rules are as essentially identical, replacing $b_{i}$ with $\beta_{i}$.
    \begin{align*}
        &(\gamma, \delta) \mapsto \Gamma \\
        &(\beta_{d}, \gamma) \mapsto B_{d} \\
        &(a_{d}^{0}, \beta_{d}) \mapsto A_d^{0} \quad &(a_d^{1}, \beta_d) \mapsto A_d^{1} \\
        &(\beta_{d - 1}, a_{d}^{0}) \mapsto B_{d - 1}^{0} \quad &(\beta_{d - 1}, a_{d}^{1}) \mapsto B_{d - 1}^{1} \\
        &(a_{d - 1}^{0}, \beta_{d - 1}) \mapsto A_{d - 1}^{0} \quad &(a_{d - 1}^{1}, \beta_{d - 1}) \mapsto A_{d - 1}^{1} \\
        &\vdots \\
        &(\beta_{1}, a_{2}^{0}) \mapsto B_{1}^{0} \quad &(\beta_{1}, a_{2}^{1}) \mapsto B_{1}^{1} \\
        &(a_{1}^{0}, \beta_{1}) \mapsto A_1^{0} \quad &(a_1^{1}, \beta_1) \mapsto A_1^{1} \\
        &(\beta, a_1^{0}) \mapsto \Delta^{0} \quad &(\beta, a_1^{1}) \mapsto \Delta^{1} \\
        &(\alpha, \beta) \mapsto \Phi \text{.}
    \end{align*}
    
    Observe that now as long as $p \leq d$, we have $\merges = p + d + 4d + 3 = p + 5d + 3$.
    For example, if we choose $p = \lceil \log d \rceil$, then $\longest \leq 2^{p} + 2 \leq 2d + 2$.
    As argued above, we have $|\Sigma| = O(d^3)$, $\merges = O(d)$, and $\longest = O(d)$.
    We reduce instead from $\kSAT{3}$ on $d/2$ variables (assuming without loss of generality that $d$ is even).
    Consider the fixed string 
    \begin{equation*}
        F = a_1^{0} b_1 \zeta^{2^{p}} \dots a_{d/2}^{0} b_{d/2} \zeta^{2^{p}} \text{.}
    \end{equation*}
    For every assignment $z \in \set{0, 1}^{d/2}$ on $d/2$ variables, define 
    \begin{equation*}
        E(z) = a_{d/2 + 1}^{z_{1}} b_{d/2 + 1} \zeta^{2^{p}} \dots a_{d}^{z_{d/2}} b_{d} \zeta^{2^{p}} \text{.}
    \end{equation*}
    For every possible clause $c$ of width $3$ on $d/2$ variables, we let $\calE_{c} = \set{E(z) \given c(z) = 0}$ denote the assignments on $d/2$ variables that falsify $c$, and we write $R_{c} = \ell_{c} \Sigma_{\calC}^* \alpha \beta F \calE_{c} \gamma$.
    By our choice of $p = \lceil \log d \rceil$, $R_{c}$ can be recognized by a DFA of size $O(d^2)$ since we may replace the at most $d$ states for recognizing $b_i$ for any $1 \leq i \leq d$ with a path of length at most $O(d)$ that recognizes $b_i \zeta^{2^{p}}$.
    We let $R_{\calC}$ be the alternation of all $R_{c}$ for $c \in \calC$.
    Then, if we let $\calE = \set{E(z) \given z \in \set{0, 1}^{d/2}}$, we define the full regular expression analogously as $R = R_{\calC} \mid \alpha \beta F \calE \gamma \delta \mid \Sigma$.
    Following similar arguments as before, and noting that once again all pretoken matching strings of length $> 1$ have distinct starting characters, $R$ can be recognized by a DFA of size $O(d^{5})$, so that the total preprocessing time of the assumed online tokenization algorithm is $\poly(d)$.

    For a $3$-CNF formula $\psi = c_{1} \wedge \dots \wedge c_{k}$ on $d/2$ variables and $k = \poly(d)$ clauses, we construct the input
    \begin{equation*}
        x_{\psi} = \ell_{c_1} \ell_{c_2} \dots \ell_{c_k} \alpha \beta F \text{.}
    \end{equation*}
    Note that $|x_{\psi}| = \poly(d)$.
    We claim that $\Tok_{\preceq}(x_{\psi})$ decides whether $\psi$ is satisfiable, even with delay.
    
    \begin{claim}
        \label{clm:sat-tok-delay-equiv}
        Let $f = |F| = \frac{d}{2}(2^{p} + 2)$.
        For any $3$-CNF formula $\psi$, we have the following:
        \begin{enumerate}
            \item (Completeness) If $\psi \in \SAT$, then for all $y \preceq x_{\psi}$ with $|y| \geq |x_{\psi}| - f$, $\Tok_{\preceq}(y) = \ell_{c_1} \ell_{c_2} \dots \ell_{c_k}$.
            \item (Soundness) If $\psi \not\in \SAT$, then for all $y \preceq x_{\psi}$ with $|y| \geq |x_{\psi}| - f$, $\Tok_{\preceq}(y) = \ell_{c_1} \ell_{c_2} \dots \ell_{c_k} \alpha$.
        \end{enumerate}
    \end{claim}

    \begin{proof}[Proof of \Cref{clm:sat-tok-delay-equiv}]
        As before, we have $\Tok_{\preceq}(\ell_{c_1} \ell_{c_2} \dots \ell_{c_k}) = \ell_{c_1} \ell_{c_2} \dots \ell_{c_k}$ since no $\ell_{c}$ participates in any merge rule and for any $y \preceq x_{\psi}$, we have $\Tok(y) = y$, since $x_{\psi}$ contains no $\gamma, \delta$.

        Suppose $\psi \in \SAT$ is satisfiable, so that there exists an assignment $z \in \set{0, 1}^{d/2}$ with $\psi(z) = 1$.
        Fix any $y \preceq x_{\psi}$ of length at least $|x_{\psi}| - f$.
        Then, $y = \ell_{c_1} \ell_{c_2} \dots \ell_{c_k} \alpha \beta F_{\leq i}$ for some $0 \leq i \leq f$.
        Consider the extension $F_{>i} E(z) \gamma \delta$.
        Following similar arguments to \Cref{clm:sat-tok-equiv}, we have
        \begin{equation*}
            \Tok(y F_{>i} E(z) \gamma \delta) = \Tok(x_{\psi} E(z) \gamma \delta) = \ell_{c_1} \ell_{c_2} \dots \ell_{c_k} \Phi A_{1}^{0} \dots A_{d/2}^{0} A_{d/2 + 1}^{z_1} \dots A_{d}^{z_{d/2}} \Gamma \text{.}
        \end{equation*}
        Thus, $\Tok_{\preceq}(y) \preceq \ell_{c_1} \ell_{c_2} \dots \ell_{c_k}$, establishing equality.

        Suppose $\psi \not\in \SAT$.
        Again, fix $y = \ell_{c_1} \ell_{c_2} \dots \ell_{c_k} \alpha \beta F_{\leq i}$ for some $0 \leq i \leq f$.
        Since every pretoken of length $>1$ involving $\alpha$ follows $\alpha$ with $\beta$, $F$, a string in $\calE$, and $\gamma$, it suffices to consider extensions of the form $F_{>i} E(z) \gamma w$ where $z \in \set{0, 1}^{d/2}$ and $w \in \Sigma^*$.
        Otherwise, the pretoken does not match $\alpha \beta F_{\leq i}$, so that $\alpha \beta F_{\leq i}$ is not merged and output as single characters.
        Since $\psi$ is unsatisfiable, there exists a clause $c_j$ in $\psi$ such that $E(z) \in \calE_{c_{j}}$.
        Following similar arguments to \Cref{clm:sat-tok-equiv}, we have
        \begin{align*}
            \Tok(y F_{>i} E(z) \gamma w) &= \Tok(x_{\psi} E(z) \gamma w) \\
            &= \ell_{c_1} \ell_{c_2} \dots \ell_{c_k} \alpha \Delta^{0} B_{1}^{0} \dots B_{d/2 - 1}^{0} B_{d/2}^{z_{1}} \dots B_{d - 1}^{z_{d/2}} B_{d} \circ \Tok(w) \text{.}
        \end{align*}
        Thus, $\Tok_{\preceq}(y) = \ell_{c_1} \ell_{c_2} \dots \ell_{c_k} \alpha$.
    \end{proof}

    Suppose for contradiction that we have an algorithm with delay $o(\merges \longest)$.
    Let $z$ be the algorithm's output on $x_\psi$, and let $k$ be its delay.
    Hence
    \begin{equation*}
        \Tok_{\preceq}(x_\psi[:-k]) \preceq z \preceq \Tok_{\preceq}(x_\psi) \text{.}
    \end{equation*}
    Since $\merges = \Theta(d)$, $\longest = \Theta(d)$, and $f = \Theta(d^2)$, we have $k \leq f$ for sufficiently large $d$.
    By \Cref{clm:sat-tok-delay-equiv}, $\Tok_{\preceq}(x_{\psi}[:-k]) = \Tok_{\preceq}(x_{\psi})$ determines whether $\psi$ is satisfiable.
    Thus $z$ also determines whether $\psi$ is satisfiable.
    Since the input $x_{\psi}$ has length $\poly(d)$, we decide whether $\psi$ is satisfiable in $\poly(d)$ time, contradicting the assumption $\P \neq \NP$.
\end{proof}

%% file: improved_delay.tex
\section{Improved Delay Beyond Composition}
\label{sec:improved-delay}

We show that a more careful algorithm can achieve delay $O(\merges \longest)$ independent of $\lookback$.
\Cref{thm:np-delay-lb} establishes that this delay is optimal for polynomial time algorithms.
We begin with some notation and useful properties. 
Recall that for any string $x$, $\phi(x)$ is the first token in $\Merge(x)$.

\begin{claim}[Theorem 2 of \citet{berglund2023formalizing}]
    \label{clm:first-token-equal}
    Let $M$ be a merge list with $\merges$ pairs and $\longest$ the length of the longest token.
    For any two input strings $x, y$, if $x_{\leq (\merges  + 1)\longest} = y_{\leq (\merges +1)\longest}$, then $\phi(x) = \phi(y)$.
\end{claim}

Throughout this section, we will view BPE as the following process on an input string $x$.
Every pair of adjacent characters initially has a cut between them, that is, every character begins as a single token.
Pretokenization fixes certain cuts in place (i.e. the cuts that represent cuts between pretokens denoted by $\brk$).
For example, the list of pretokens 
\begin{equation*}
    \rust{"aaa"} \brk \rust{"bbb"} \brk \rust{"aaaa"}
\end{equation*}
can be viewed as list of characters with cuts
\begin{equation*}
    a \mid a \mid a \brk b \mid b \mid b \brk a \mid a \mid a \mid a
\end{equation*}
where $\brk$ represents fixed cuts given by the pretokenization phase and $\mid$ represents unfixed cuts (or simply cuts) that can be removed by the BPE merge process.
BPE merge on each pretoken then removes cuts, i.e. every merge removes a cut between two adjacent tokens.
Following the above example, the merge rule $(a, a) \mapsto A$ removes cuts so that the output is 
\begin{equation*}
    a a \mid a \brk b \mid b \mid b \brk a a \mid a a 
\end{equation*}
corresponding to the output string \rust{"AabbbAA"}.

An ordered pair of tokens $(\tau, \mu)$ is consistent, denoted $\tau \rightarrow \mu$, if
\begin{equation*}
    \Merge(\raw(\tau) \raw(\mu)) = \tau \mu \text{.}
\end{equation*}
We show that consistency extends to sequences of tokens.

\begin{claim}
    \label{clm:stable-merge-seq}
    Let $\tau_1, \dots, \tau_{k}$ be a sequence of tokens.
    Then
    \begin{equation*}
        \Merge(\raw(\tau_1) \dots \raw(\tau_k)) = \tau_1 \dots \tau_k \Longleftrightarrow \tau_i \rightarrow \tau_{i + 1}~\forall~1 \leq i < k \text{.}
    \end{equation*}
    We say $\tau_1, \dots, \tau_k$ is a consistent sequence of tokens if they satisfy the left hand side.
    Consequently, for any token $\mu$ and consistent sequence of tokens $\tau_1, \dots, \tau_k$, we have $\mu, \tau_1, \dots, \tau_k$ is consistent iff $\mu \rightarrow \tau_1$.
\end{claim}

Following the above convention of viewing BPE as a process of removing cuts between characters, a sequence of tokens $\tau_{1} \dots \tau_{k}$ can be viewed as the raw string $\raw(\tau_{1}) \dots \raw(\tau_{k})$ with the cuts surviving the BPE merge process exactly denoting the cuts in the final output.
\begin{equation*}
    \raw(\tau_{1}) \mid_{1} \raw(\tau_{2}) \mid_{2} \dots \mid_{k - 1} \raw(\tau_{k}) 
\end{equation*}
where $\mid_{i}$ denotes the surviving cut between token $\tau_{1}$ and $\tau_{2}$.

\begin{proof}[Proof of \Cref{clm:stable-merge-seq}]
    Suppose $\tau_1, \dots, \tau_k$ is consistent. 
    Consider the interval $\raw(\tau_i) \raw(\tau_{i + 1})$ for any $1 \leq i < k$ with endpoint cuts $\mid_{i - 1}, \mid_{i + 1}$.
    Observe that for two strings containing this interval, if after applying merge $m \in M$, neither endpoint cut is merged, then the cuts internal to the interval are identical after applying merge $m$.
    Then, consider the strings $\raw(\tau_1) \dots \raw(\tau_k)$ and $\raw(\tau_i) \raw(\tau_{i + 1})$.
    Since $\tau_1, \dots, \tau_k$ is consistent, the endpoint cuts $\mid_{i - 1}, \mid_{i + 1}$ are never merged.
    The same statement holds trivially over the string $\raw(\tau_i) \raw(\tau_{i + 1})$.
    Thus, since the cuts internal to the interval are identical, we have that $\tau_i \rightarrow \tau_{i + 1}$.

    Conversely, suppose every pair $\tau_i \rightarrow \tau_{i + 1}$ is consistent.
    If every cut between $\tau_{i}, \tau_{i + 1}$ is not merged, then we have $\Merge(\raw(\tau_1) \dots \raw(\tau_k)) = \tau_1 \dots \tau_k$ by the above argument.
    Thus, assume some cut between $\tau_{i}, \tau_{i + 1}$ is merged.
    Let $m$ be the highest priority rule that removes a boundary cut, and among the boundary cuts removed by $m$, let $\mid_{i}$ be the one of minimum index.
    We claim that at the instant $\mid_{i}$ is removed, neither endpoint cut $\mid_{i - 1}$ nor $\mid_{i + 1}$ has been removed.
    Indeed, by the choice of $m$, no rule of higher priority than $m$ removes a boundary cut, and every rule of lower priority is applied only after $m$; so it suffices to rule out that $m$ itself removes an endpoint cut before $\mid_{i}$.
    For $\mid_{i - 1}$, this holds by the minimality of $i$.
    For $\mid_{i + 1}$, note that minimality does not apply, since $i + 1 > i$; instead, if $m$ also removes $\mid_{i + 1}$, then it does so only after removing $\mid_{i}$, as $\mid_{i}$ lies to the left of $\mid_{i + 1}$.
    Consequently, the cuts internal to the interval $\raw(\tau_i) \raw(\tau_{i + 1})$ are identical to those in the $\Merge$ process on $\raw(\tau_i) \raw(\tau_{i + 1})$ alone, so $\mid_{i}$ is removed there as well.
    This contradicts $\tau_i \rightarrow \tau_{i + 1}$.
\end{proof}

We now present our algorithm.

\MTDelay*

\begin{proof}[Proof of \Cref{thm:online-tokenization-const-delay}]
    We present an algorithm that computes $(\Merge \circ \Split)_{\preceq}$ with delay $O(\merges \longest)$.
    Following similar arguments to \Cref{thm:linear-online-tokenization-w-delay}, we can obtain an online tokenization algorithm with delay $O(\merges \longest) + \ell_{\infty}(S^*) = O(\merges \longest)$, since $\Seg_{\preceq}$ outputs all but the last $\ell_{\infty}(S^*) = O(1)$ characters of the input.
    We begin with the following useful data structures.

    \paragraph{BPE Queue}
    On any string $w$, let $\phi(w; M) = \Merge(w; M)_{1}$ be the first token in $\Merge(w)$.
    As before, we use $\phi(w)$ when the merge list $M$ is clear.

    \BPEQLem*

    Note that queries are only required to be correct when $|w|$ is sufficiently large ($\gg \merges \longest$).

    \begin{proof}[Proof of \Cref{lem:bpe-queue}]
        We require a modification of the incremental algorithm of \citet{jiang2025incrementalbpe}.
        
        \begin{lemma}
            \label{lem:batch-bpe}
            Given BPE merge list $M$, there is an algorithm $\Phi$ with $O(\merges  \longest)$ preprocessing and space that given any string $x$, computes $\phi(x_{\geq i})$ for all $i$ in $O(|x| \log^2 \longest)$ time.
        \end{lemma}
        
        Note that \citet{jiang2025incrementalbpe} give an algorithm that computes in $O(n \log^2 \longest)$ time the \emph{last} token of every \emph{prefix} of $x$.
        In contrast, we compute the \emph{first} token of every \emph{suffix} of $x$.
        Since the data structure essentially follows from \citet{jiang2025incrementalbpe}, we defer the proof to \Cref{app:omitted}.

        Finally, we can use a standard trick to convert the batch algorithm into an algorithm with efficient amortized update and query time.
        We maintain two queues $U$ and $V$ such that the maintained string $w$ satisfies $w = U V$.
        Initially, both $U, V = \epsilon$.
        We maintain the invariant that:
        \begin{center}
            If $|w| \geq 3 \merges \longest$, then $2\merges \longest \leq |U| \leq 3\merges \longest$ and $\phi(U_{\geq i})$ is computed for all $i \leq |U| + 1$.
        \end{center}
        On a push, we push the character to $V$.
        On a pop, we pop the character from $U$ if it is not empty. Otherwise, $V$ is not empty and we remove a character from $V$.
        After each operation, if $|U| < 2\merges \longest$ and $|w| \geq 3 \merges \longest$, then we have $|V| > \merges \longest$, in which case we let $U \gets w_{\leq 3 \merges \longest}$ and apply the algorithm of \Cref{lem:batch-bpe} to $U$. 
        
        We check that the invariant is always maintained. 
        Consider the first time $|w| \geq 3 \merges \longest$, at this time we have $|V| = 3 \merges \longest$ and $U$ is empty.
        Our algorithm detects this, and moves $w$ to $U$, applying the batch algorithm to $U$, maintaining the invariant.
        Subsequently, anytime $|U| < 2\merges \longest$ and $|w| \geq 3\merges \longest$, we again transfer $w_{\leq 3\merges \longest}$ to $U$ and apply the batch algorithm.
        Note that $|U| \leq 3\merges \longest$ always and $\phi(U_{\geq i})$ is always computed, since whenever we add elements to $U$, we recompute $\phi(U_{\geq i})$ on the whole object $U$, while popping elements from $U$ ensures that $\phi(U_{\geq i})$ is still computed on the remaining elements.
        Each application of the batch algorithm requires $O(\merges  \longest \log^2 \longest)$ time and ends with $|U| = 3\merges \longest$.
        Since we do not apply the algorithm again until $|U| < 2\merges \longest$, or after $\Omega(\merges \longest)$ operations, the amortized time is $O(\log^2 \longest)$.

        Finally, correctness follows since whenever $|w| \geq 3\merges \longest$, the invariant guarantees $|U| \geq 2\merges \longest$.
        \Cref{clm:first-token-equal} then ensures that $\phi(w) = \phi(U)$.
    \end{proof}

    \paragraph{Incremental BPE}
    We also extend the incremental BPE algorithm of \citet{jiang2025incrementalbpe}.

    \incBPELem*

    Popping is allowed only when $I < |h|$.
    First token queries are only required to return $\phi(w)$ when $\Merge(h_{\leq I}) \prec \Merge(h)$, i.e. when $I$ is a token boundary in the tokenization of the full implicit hidden string $h$.
    
    \begin{proof}[Proof of \Cref{lem:incremental-bpe}]
        \citet{jiang2025incrementalbpe} give an algorithm with $O(\merges  \longest)$ preprocessing and space (see Section 5 and 6 of \citep{jiang2025incrementalbpe} and \Cref{app:omitted} for details) $O(\log^2 \longest)$ worst-case update time supporting push operations and last token queries after any number of push operations.
        For example, given a string $h$, they can compute the last token of every prefix of $h$, i.e. the last token of $\Tok(h_{\leq i})$ for all $i$, denoted $\lambda(h_{\leq i})$.
        
        We describe how to augment their data structure with first token $\phi$ queries, as pops can trivially be maintained with a separate counter for $I$, initialized at $0$.        
        Note that \citet{jiang2025incrementalbpe} provide $\lambda(h_{\leq i})$ queries for all $i$.
        Let $J = |h|$ denote the total number of characters pushed onto $h$, including those that have been popped.
        Consider the tree $\calT$ whose nodes are indices $0 \leq i \leq J$ and node $i$ has parent $p(i) = i - |\lambda(h_{\leq i})|$.
        Intuitively, node $i$ points to its parent, the last character of the token immediately preceding $\lambda(h_{\leq i})$.
        Thus, if we wanted to compute the BPE output of $h_{\leq i}$, it suffices to follow the path from $i$ to the root of the tree $\calT$.

        Let $I$ be the number of characters popped so $I \leq J$.
        Note that $I, J$ naturally correspond to indices in the input string. 
        We abuse notation and let $I, J$ also denote the nodes in the tree $\calT$ corresponding to the indices $I, J$.
        Our data structure maintains $w = h_{> I}$, and induced forest of $\calT$ on the nodes corresponding to indices $I \leq i \leq J$, which we denote $\calT_{\geq I}$.
        Note that for any push, we only need to query $\lambda(h_{\leq i})$ from the \citet{jiang2025incrementalbpe} data structure, and add this node to $\calT_{\geq I}$ (either to its parent if it is in the tree, or as the root of a new tree otherwise).
        We maintain $\calT_{\geq I}$ using standard data structures (for example, an Euler Tour and efficient order queries using \cite{dietz1987two, bender2002two}) so that one can efficiently insert and delete nodes, and support queries of the form is $a$ a descendant of $b$, and if so, what is the next vertex on that path?
        This allows us to support first token queries, as we can query whether $J$ is a descendant of $I$.
        If so, report $\phi(w)$ as the child of $I$ on the path form $I$ to $J$ (note that this could be $J$ itself), otherwise return $\bot$.
        Finding the child of $I$ can be performed with binary search in $O(\log |w|)$ time.
        Overall, this requires space $O(|w| + \merges \longest)$, since we also use the data structure of \citet{jiang2025incrementalbpe}.

    \end{proof}

    Next, we give an efficient data structure for computing pretokens.
    Recall from the proof of \Cref{lem:pretokenize-alg}, for an incrementally updated string $x_{\leq i}$: $\Split_{\preceq}(x_{\leq i}) = Q \brk x_{b(i) + 1} \dots x_{\ell(i)}$ where $\ell(i) = b(i) + \last(\start, x_{b(i) + 1} \dots x_{i})$ or $\Split_{\preceq}(x_{\leq i}) = Q \brk$ in which case $b(i) = \ell(i) = i$.
    Here, $b(i)$ indicates the start of the current ``live'' pretoken in $x_{\leq i}$.

    \begin{lemma}
        \label{lem:pretok-q}
        Given a total regular expression $R$, there is a data structure $\cutQ$ that maintains implicitly an append-only string $x$ (of length $i$), and index $0 \leq \out \leq |x|$.
        Denote
        \begin{enumerate}
            \item $\Split(x) = z_1 \brk \dots \brk z_{\ell}$ where $z_{j} = x_{r(j - 1) + 1} \dots x_{r(j)}$ for $1 \leq j \leq \ell$.
            \item $\Split_{\preceq}(x) = z_{1} \brk \dots \brk z_{\base}$ where $z_{\base} = x_{b(i) + 1} \dots x_{\ell(i)}$.
        \end{enumerate}
        Then, $\cutQ$ supports the following operations:
        \begin{enumerate}
            \item (Push) append $c$ to $x_{\leq i} \gets x_{<i} c$.
            \item (Pop) increment $\out$ by 1.
            \item (Cut) if $\out = i$, return $i$, otherwise return $r(j)$ where $r(j - 1) \leq \out < r(j)$.
            \item (Base) return $b(i)$.
        \end{enumerate}
        Furthermore, the data structure requires $O(1)$ preprocessing, amortized update and query time $O(1)$ and $O(1 + i - \out)$ space.
    \end{lemma}

    \begin{proof}[Proof of \Cref{lem:pretok-q}]
        We will maintain a linked list $\calL \subseteq \set{(r(j - 1), r(j))}_{1 \leq j \leq \ell}$ that intuitively contains pretokens in $\Split(x_{\leq i})$ where $r(j) > \out$.
        We will also maintain a data structure $\calA$ that contains for some subset of $1 \leq j \leq \ell$, the tuple $r(j - 1)$ and $\state(j) := \delta(\start, x_{r(j - 1) + 1} \dots x_{i})$.
        For the analysis, denote $\state_{0}(j) = \delta(\start, x_{r(j - 1) + 1} \dots x_{i - 1})$.
        Let $\calQ$ be the state space of the DFA recognizing $L(R)$.
        We will ensure that $\calA$ contains at most $|\calQ|$ tuples by ensuring that for every $q \in \calQ$, we keep only $(s^*, q) \in \calA$ where $s^* = \min \set{s \mid (s, q) \in \calA}$.
        This operation, which can be computed in $O(1)$ time, is denoted $\trim(\calA)$.
        At initialization, both $\calL, \calA \gets \emptyset$.
        Upon receiving $x_{i}$, run the following algorithm.
        
        \begin{mdframed}
            $\cutQ.\push(x_{i})$

            \stepnames

            \begin{enumerate}
                \item For $(s, q) \in \calA$, update $q \gets \delta(q, x_{i})$. $\trim(\calA)$.

                \item Let $X = \set{s \mid (s, q) \in \calA, q \in \accept}$. 

                \item If $X \neq \emptyset$, let $s^* = \min X$. Remove all $(s, q)$ from $\calA$ with $s > s^*$.
                \label{step:accept-update-pretok-q}
                \begin{enumerate}
                    \item If $s^* \leq \out$, set $\calL \gets \set{(s^*, i)}$.
                    \item If $s^* > \out$, remove all $(s, e) \in \calL$ where $s \geq s^*$ and insert $(s^*, i)$ to $\calL$.
                \end{enumerate}

                \item Otherwise, if $X = \emptyset$, add $(i - 1, i)$ to $\calL$ and $(i - 1, \delta(\start, x_{i}))$ to $\calA$.

                \Comment{(Totality ensures $\delta(\start, x_{i}) \in \accept$.)}

                \item Remove $(s, q)$ from $\calA$ if $q \in \dead$.
            \end{enumerate}
        \end{mdframed}

        \begin{mdframed}
            $\cutQ.\pop()$

            \stepnames

            \begin{enumerate}
                \item Increment $\out \gets \out + 1$.
                \item Remove any $(s, e) \in \calL$ where $e \leq \out$.
            \end{enumerate}
        \end{mdframed}
        
        We claim that the following hold after each update (push or pop):
        \begin{enumerate}
            \item If $\out = i$, $\calL = \emptyset$. Otherwise, $\calL = \set{(r(j_0 - 1), r(j_0)), \dots, (r(\ell - 1), r(\ell))}$ where $r(j_0 - 1) \leq \out < r(j_0)$.
            \item $\calA$ contains for every $q \not\in \dead$ such that $s_{q} = \min \set{r(j - 1) \mid \state(j) = q}$ exists, a tuple $(s_{q}, q)$, and no other tuples.
        \end{enumerate}
        Before proving the invariants, we see that this efficiently supports the desired queries. 
        Cut queries can be supported in $O(1)$ time by examining the first element of $\calL$ (or returning $i$ if $\calL$ is empty).
        Base queries can be supported by computing $\mu \gets \min \set{s \given (s, q) \in \calA}$ in $O(|\calA|) = O(1)$ time.
        If the set is empty, we set $\mu \gets i$.
        We claim that $\mu = b(i)$, since 
        \begin{align*}
            \mu &= \min \set{s \given (s, q) \in \calA} \\
            &= \min \set{r(j - 1) \given \delta(\start, x_{r(j - 1) + 1} \dots x_{i}) \not\in \dead} = b(i) \text{.}
        \end{align*}
        When the set is empty, we have $b(i) = \mu = i$, as desired.
        It remains to prove that the above invariants hold, which we do so by induction.
        In the base case, before any operations, the invariants are trivially satisfied.

        Consider a push operation of $x_i$.
        There are two cases: either $x_i$ does not extend any existing pretoken, or it does.
        
        In the former case, we have
        \begin{equation*}
            \Split(x_{\leq i}) = z_1 \brk \dots \brk z_{\ell} \brk x_{i} \text{.}
        \end{equation*} 
        Observe that $X = \emptyset$ since every $\delta(\start, x_{r(j - 1) + 1} \dots x_{i}) \not\in \accept$, so we add $(i - 1, i)$ to $\calL$ and $(i - 1, \delta(\start, x_{i}))$ to $\calA$.
        This satisfies the first property.

        Towards the second property, consider some $q \not\in \dead$ and the minimum $r(j - 1)$ such that $\state(j) = q$.
        We claim $(r(j - 1), \state(j)) \in \calA$.
        Suppose $\state(j) \not\in \accept$, which covers all but the new singleton pretoken.
        Note that for any $j, j'$, if $\state(j) \neq \state(j')$, then $\state_{0}(j) \neq \state_0(j')$.
        Thus, at the start of the $\push$, the same $j$ must have been the minimum $j$ with state $\state_0(j)$.
        Thus, we have $(r(j - 1), \state_0(j)) \in \calA$ by induction.
        If $(r(j - 1), \state(j)) \not\in \calA$, it must have been deleted by $\trim(\calA)$ since $X = \emptyset$ and \Cref{step:accept-update-pretok-q} does not run.
        This cannot happen, since $r(j - 1)$ is the minimum over all $j$ with $\state(j) = q$.
        The unique $\state(j) \in \accept$ is the new singleton pretoken, which we insert into $\calA$.
        
        Conversely, consider a tuple $(s, q) \in \calA$.
        Any tuple in $\calA$ is $(r(j - 1), \state(j))$ for some $j$ (since all pretokens from before the $\push$ survive).
        Suppose for contradiction that there is some $j'$ with $r(j') < r(j)$ and $\state(j) = \state(j')$.
        Take the minimum such $j'$.
        By the inductive hypothesis, $(r(j' - 1), \state_0(j'))$ was in $\calA$ before the push, so that $(r(j' - 1), \state(j')) \in \calA$ is in $\calA$, implying that $(r(j - 1), \state(j))$ is removed by $\trim(\calA)$.
        This concludes the proof of the second property.

        In the latter case, for some $k \geq \base$, we have
        \begin{equation*}
            \Split(x_{\leq i}) = z_{1} \brk \dots \brk z_{k} \dots z_{\ell} x_{i} \text{.}
        \end{equation*}
        By the choice of $k$, it is the earliest old pretoken extended by $x_i$. Since $\state(k) \in \accept$, its previous state $\state_0(k)$ is not dead as reading $x_i$ provides a nonempty accepting continuation.
        We claim that $(r(k-1),\state_0(k))$ is present in $\calA$ before the push.
        Otherwise, the second invariant gives an earlier $j<k$ with $\state_0(j)=\state_0(k)$.
        Then $\state(j)=\state(k)\in\accept$, so $x_i$ would extend the earlier pretoken $z_j$, contradicting the choice of $k$.
        Consequently, $r(k-1)\in X$, and no earlier old pretoken start belongs to $X$.
        Thus $s^*=r(k-1)$.

        Towards the first property, we see that if $r(k - 1) \leq \out$, then $\calL = \set{(r(k - 1), i)}$ satisfies the invariant.
        Similarly, if $r(k - 1) > \out$, we have $\calL = \set{\dots, (r(k - 1), i)}$, as desired.

        Towards the second property, again consider some $q \not\in \dead$ and the minimum $j$ such that $\state(j) = q$.
        We claim $(r(j - 1), \state(j)) \in \calA$.
        If not, following similar arguments as the previous case, it must have been removed either by $\trim(\calA)$ or \Cref{step:accept-update-pretok-q}.
        We have already argued that $\trim(\calA)$ will not remove the tuple, so we may focus on \Cref{step:accept-update-pretok-q}.
        Here, we only remove tuples $(s, q)$ where $s > s^* = r(k - 1)$.
        Since all pretokens have $r(j - 1) \leq r(k - 1)$, this cannot happen either.

        Conversely, consider $(s, q) \in \calA$.
        If $(s, q) = (r(j - 1), \state(j))$ corresponds to a pretoken in $\Split(x_{\leq i})$, a similar argument to above proves that $r(j - 1)$ is the minimum among all tuples with state $\state(j)$.
        However, it may now be the case that $(s, q)$ corresponds to some pretoken $z_{j}$ with $j > k$.
        In this case, we have $r(j - 1) > r(k - 1)$, so that \Cref{step:accept-update-pretok-q} removes $(s, q)$, as desired.

        Finally, consider a $\pop$ operation.
        Note that $\out$ that does not affect the second invariant at all, so this remains satisfied.
        We remove all $(s, e) \in \calL$ with $e \leq \out$, so that the first interval in $\calL$ has $e > \out$, as desired.

        \paragraph{Performance}
        It remains to bound the performance.
        Preprocessing is constant, since we just need to build the DFA and initialize $\calL, \calA$.
        All queries require $O(1)$ time, since we ensure $\calA$ has at most size $|\calQ| = O(1)$.
        Each $\push$ and $\pop$ operation requires $O(1)$ time for all operations that do not involve $\calL$, again because $|\calA| = O(1)$.
        To bound the amortized update time, note that each iteration adds at most one interval, and each interval is removed at most once.
        Finally, we note that the space is $O(\calL) = O(1 + i - \out)$.
    \end{proof}

    We now use \Cref{lem:bpe-queue}, \Cref{lem:incremental-bpe}, and \Cref{lem:pretok-q} to prove \Cref{thm:online-tokenization-const-delay}, i.e. give an online tokenization algorithm with constant delay.
    We begin by describing the data structures our algorithm will use.
    We initialize $\cut, \out \gets 0$.
    We ensure that $\out$ maintains the last index in the output and $\cut$ to maintain the end of the pretoken containing out in $\Tok(x)$, i.e. the current pretoken being output if the input were to immediately terminate.

    We maintain a raw string buffer $\buffer$ that supports $\push$ and $\pop$.
    We initialize empty $\bpequeue$ and $\incbpe$ data structures.
    We initialize $\cutQ$ with $\out \gets 0$ (the same $\out$ as the counter).
    Whenever $\out$ is incremented, we accordingly pop from $\cutQ$.
    When a data structure is used as an input, this denotes reading the string maintained by the structures.
    When $\pop$ is supplied with an integer argument $p$, we $\pop$ $p$ items.
    When $\push$ is supplied with a string, we push the string character by character.
    For each character $x_i$ in the input, do the following:
    \begin{mdframed}
        {\bf $\python{OptimalDelayTokenization}(x_{i})$}
        
        \begin{enumerate}
        \stepnames
            \item Push $x_i$ to $\bpequeue$, $\cutQ$, and $\buffer$.
            \label{step:push}
            
            \item If $\cutQ.\cut > \cut$:

            \Comment{(The current pretoken containing $\out$ extends.)}
            
            ~~~~Push $\incbpe.\push(\buffer.\pop(\cutQ.\cut - \cut))$.
            
            ~~~~Update $\cut \gets \cutQ.\cut$.
            \label{step:accept}
            
            \item While $\cutQ.\base > \out$:
            \label{step:dead-loop}

            \Comment{(Cannot extend the current pretoken containing $\out$. 
            Process dead pretokens, and find the next live pretoken. 
            Compare to \Cref{step:pretok-dead-loop} of {\bf $\FastPretokenization$}.)}
            
            \begin{enumerate}
                \item Push $\incbpe.\push(\brk)$.

                \item While $\incbpe \neq \epsilon$:

                 ~~~~Output $\phi \gets \incbpe.\phi$. 
                 
                 ~~~~Update $\bpequeue.\pop(|\phi|)$, $\incbpe.\pop(|\phi|)$, and $\out \addgets |\phi|$.

                 \label{step:dead-loop-match}

                 \Comment{(Compute BPE on the current dead pretoken with $\incbpe$.)}

                 \item Push $\incbpe.\push(\buffer.\pop(\cutQ.\cut - \cut))$.
                 
                 \item Update $\cut \gets \cutQ.\cut$.

                 \Comment{(Move the next pretoken to $\incbpe$.)}    
            \end{enumerate}
            
            \item While $\incbpe.\phi = \bpequeue.\phi \neq \bot$:

            \begin{enumerate}
                \item Output $\phi \gets \incbpe.\phi$.

                \item Update $\bpequeue.\pop(|\phi|)$, $\incbpe.\pop(|\phi|)$, and $\out \addgets |\phi|$.
            \end{enumerate}
            
            \Comment{(When $\incbpe$ and $\bpequeue$ agree, $\incbpe.\phi$ extends the stable prefix.)}

            \label{step:while-match}

            \item If $\incbpe = \epsilon$ and $\out < i$:

            \label{step:inc-cut}

            \Comment{(Load the next pretoken into $\incbpe$.)}

            \begin{enumerate}
                \item Push $\incbpe.\push(\brk \buffer.\pop(\cutQ.\cut - \cut))$.
                
                \item Update $\cut \gets \cutQ.\cut$ and go to \Cref{step:while-match}.
            \end{enumerate}
        \end{enumerate}
    \end{mdframed}

    When $\out < i$, write $\Split(x_{\leq i}) = Q \brk z_{-p} \brk \dots \brk z_{0} \brk z_{1} \brk \dots \brk z_{\ell}$ where $x_{b(i) + 1} \dots x_{i} = z_{-p} \dots z_{\ell}$ with $\ell(i) = |Q| + |z_{-p}|$ and $|z_{j}| = \last(\start, z_{j} \dots z_{\ell})$.
    Above, $z_{-p}, \dots, z_{\ell}$ are the pretokens in $x_{\leq i}$ if the input were to immediately terminate. 
    We use $z_{j} \dots z_{\ell}$ to denote the concatenation of pretokens $z_{j}, z_{j + 1}, \dots, z_{\ell}$, so that $z_{j}$ is indeed the longest accepting prefix in $z_{j} \dots z_{\ell}$ for each $-p \leq j \leq \ell$. 
    Define $r(-p-1) = b(i)$ and for each $-p \leq j \leq \ell$, $r(j) = r(j - 1) + |z_{j}|$ so that $z_{j} = x_{r(j - 1) + 1} \dots x_{r(j)}$.
    Index the pretokens so that $r(-1)\leq\out<r(0)$; thus $z_0$ contains the first un-emitted character.
    When $\out=i$, no such $z_0$ is needed.
    
    We now make the following claim at after processing each input $x_i$:
    \begin{enumerate}
        \item $\buffer$ maintains $x_{\cut + 1} \dots x_{i}$.
        \item $\bpequeue$ maintains $x_{\out + 1} \dots x_{i}$.
        \item $\incbpe$ explicitly maintains $x_{\out + 1} \dots x_{\cut}$. Consequently, $\out \leq \cut$.
        \item The output $y$ of the algorithm satisfies $\raw(y) = x_{\leq \out}$.
        \item  If $\out<i$, then $r(-1)\leq\out<r(0)$ and $\cut=r(0)$. If $\out=i$, then $\cut=i$.
        \item $\out \geq b(i)$.
    \end{enumerate}
    Trivially, the conditions are satisfied at initialization.
    Observe that $\cut, \out$ are monotonically non-decreasing.
    To prove the first three properties, note that (1) $x_{i}$ is always pushed onto $\buffer, \bpequeue$, (2) we always push and pop $\cutQ.\cut - \cut$ elements from $\buffer$ to $\incbpe$ and increment $\cut$ by the same amount, and (3) we always increment $\out$ by $|\phi|$, which is popped from $\bpequeue$ and $\incbpe$.
    Similarly, the fourth property follows as every output is accompanied by a corresponding increment to $\out$.

    Next, we argue the fifth property.
    If $\out=i$, the third property and $\cut \leq i$ give $\cut=i$. 
    Otherwise, $r(-1) \leq \out < r(0)$ holds by definition.
    We argue $\cut = r(0)$.
    After \Cref{step:accept} and \Cref{step:dead-loop}, $\cut = r(0)$ by correctness of \Cref{lem:pretok-q}, since we query $\cutQ.\cut$ after any pushes and pops (i.e. increments to $\out$).
    In \Cref{step:while-match}, we increment $\out$, but as long as $\incbpe \neq \epsilon$, we still have $\out < \cut = r(0)$ by the third property, in which case $r(0)$ does not change and $\cut$ does not need to updated.
    On the other hand, whenever $\incbpe = \epsilon$, in \Cref{step:inc-cut} we recompute $\cut \gets \cutQ.\cut = r(0)$ by the correctness of \Cref{lem:pretok-q}.

    Finally, we argue $\out \geq \cutQ.\base = b(i)$.
    In \Cref{step:dead-loop}, we already have $\cutQ.\base = b(i)$.
    Thus, as long as the while loop terminates (which we will show later), we must have $\out \geq b(i)$.

    We are now ready to prove correctness, and proceed with the inductive claim that the output $y$ satisfies
    \begin{equation*}
        (\Merge \circ \Split)_{\preceq}(x_{\leq i - 3 \merges \longest}) \preceq y \preceq (\Merge \circ \Split)_{\preceq}(x_{\leq i}) \text{.}
    \end{equation*}
    We argue inductively over individual token emissions within each input iteration. 
    Thus, before considering the next emission, we may assume that the output already produced is a stable prefix.
    
    We first justify first-token queries to $\incbpe$.
    Consider a query with $\out<\cut$, and let $w=x_{\out+1}\dots x_{\cut}$ be the string maintained by $\incbpe$.
    Let $a$ be the position such that $\incbpe$ implicit hidden history is $\dots \brk x_{a + 1} \dots x_{\cut}$, possibly followed by a trailing separator at $\cut$.
    We denote $u = x_{a + 1} \dots x_{\out}$.
    
    Observe that both $a$ and $\out$ are token boundaries in the tokenization
    of the current input: $a$ is either a pretoken boundary or a boundary of the already-emitted stable output, and $\out$ is the end of that output.
    The block $uw$ lies within a single pretoken of the current $\Tok(x)$ tokenization: it is loaded within one pretoken, and later longest-match extensions can only preserve or join existing pretokens.
    Thus $a$ and $\out$ are surviving token boundaries within that pretoken, so \Cref{clm:stable-merge-seq} gives $\Merge(uw)=\Merge(u)\circ\Merge(w)$.
    Since the implicit hidden history of $\incbpe$ is $\dots \brk uw$, we have $\Merge(h_{\leq \out}) \prec \Merge(h)$ and \Cref{lem:incremental-bpe} applies. 
    In particular, since the first token query condition is satisfied,
    \begin{equation}
        \incbpe.\phi = \phi(x_{\out+1} \dots x_{\cut}).
        \label{eq:incbpe-current-query}
    \end{equation}
    
    Beginning with soundness, it is clear that if we do not output any tokens after reading $x_i$, then the output is a stable prefix by induction as
    \begin{equation*}
        y \preceq (\Merge \circ \Split)_{\preceq}(x_{<i}) \preceq (\Merge \circ \Split)_{\preceq}(x_{\leq i}) \text{.}
    \end{equation*}
    Suppose a token is output in \Cref{step:dead-loop} (i.e. $\cutQ.\base > \out$).
    By the correctness of \Cref{lem:pretok-q}, we have $\cutQ.\base = b(i)$ so that $b(i - 1) \leq \out < b(i)$.
    Following the proof of \Cref{lem:pretokenize-alg}, write
    \begin{align*}
        \Split_{\preceq}(x_{<i}) &= P_{i - 1} \brk x_{b(i - 1) + 1} \dots x_{\ell(i - 1)} \\
        \Split_{\preceq}(x_{\leq i}) &= P_{i - 1} \brk x_{b(i - 1) + 1} \dots x_{r'(0)} \brk x_{r'(0) + 1} \dots x_{r'(1)} \brk \dots \brk x_{b(i) + 1} \dots x_{\ell(i)}
    \end{align*}
    where for some $t \geq 1$, we have $r'(-1) = b(i - 1)$, $r'(t) = \ell(i)$, and $r(i) = r(i - 1) + \last(\start, x_{r(i - 1) + 1} \dots x_{i})$ for $0 \leq i \leq t$.
    Following our above notation, we have $z_{-p} = x_{r'(t - 1) + 1}\dots x_{r'(t)}$.
    Then, observe
    \begin{align*}
        \Merge_{\preceq}(\Split_{\preceq}(x_{\leq i})) &= \Merge(P_{i - 1} \brk x_{b(i - 1) + 1} \dots x_{r'(0)} \brk x_{r'(0) + 1} \dots x_{r'(1)} \brk \dots \brk) \\
        &\quad\quad\quad \circ \Merge_{\preceq}(x_{r'(t - 1) + 1} \dots x_{r'(t)}) \\
        &\preceq (\Merge \circ \Split)_{\preceq}(x_{\leq i}) \text{.}
    \end{align*}
    Consider an iteration of \Cref{step:dead-loop}. 
    Since $\out < b(i)$, the current pretoken ends at $\cut \leq b(i)$, so this pretoken is output by $\Split_{\preceq}$ and is unchanged under every continuation.
    By \eqref{eq:incbpe-current-query}, the query returns $\phi(x_{\out+1}\dots x_{\cut})$, the next token of the current pretoken.
    Since the already-emitted output is stable, this is also the next token under every continuation.
    Thus each emission preserves soundness, and the argument applies again after updating $\out$.

    Upon exiting \Cref{step:dead-loop}, we have $\out = b(i)$, and the emitted prefix is the following stable prefix:
    \begin{equation*}
        \Merge(P_{i - 1} \brk x_{b(i - 1) + 1} \dots x_{r'(0)} \brk x_{r'(0) + 1} \dots x_{r'(1)} \brk \dots \brk) \preceq (\Merge \circ \Split)_{\preceq}(x_{\leq i}) \text{.}
    \end{equation*}

    Now suppose a token $\phi$ is emitted in \Cref{step:while-match}.
    Write $z_0 = z_0'z_0''$, where $|z_0'|=\out-r(-1)$.
    By \eqref{eq:incbpe-current-query} $\phi=\incbpe.\phi=\phi(z_0'')$.
    Also, $\bpequeue$ maintains $z_0'' z_1 \dots z_{\ell}$ and $\phi = \bpequeue.\phi \neq \bot$ so $|\bpequeue| \geq 3 \merges \longest$.
    
    Fix any continuation $z$.
    If no pretoken extends, or the first pretoken
    that extends occurs strictly after $z_0$, then $z_0$
    is unchanged.
    The next token is therefore $\phi(z_0'')=\phi$.
    
    Otherwise, let $z_k$, with $k \leq 0$, be the first pretoken that extends.
    Because its longer match was unavailable before reading $z$, the new match ends strictly after position $i$.
    Consequently, this pretoken has raw string $z_k\dots z_\ell z'$ for some nonempty prefix $z'\preceq z$.
    The remaining suffix of $z$ may form further pretokens.
    
    By the inductive stability of the emitted output, the boundary at $\out$ survives within this extended pretoken. 
    In particular, we have
    \begin{equation*}
        \Merge(z_{k} \dots z_{\ell} z') = \Merge(z_{k} \dots z_{0}') \circ \Merge(z_{0}'' \dots z_{\ell} z') \text{.}
    \end{equation*}
    By \Cref{clm:first-token-equal},
    \begin{equation*}
        \phi(z_0''z_1\dots z_\ell z')
        =\phi(z_0''z_1\dots z_\ell)
        =\bpequeue.\phi
        =\phi \text{.}
    \end{equation*}
    Thus $\phi$ is the next token for every continuation and emitting it preserves soundness.
    
    Finally, we prove completeness, i.e. $(\Merge \circ \Split)_{\preceq}(x_{\leq i - 3 \merges \longest}) \preceq y$.
    As in the proof of soundness, we note that after exiting \Cref{step:dead-loop}, we have $\out \geq b(i)$.
    We may assume without loss of generality that $\out < i - 3\merges \longest$, as otherwise \Cref{clm:coverage} already bounds the delay by $3 \merges \longest$.

    Thus $|\bpequeue| = i-\out > 3\merges\longest$.
    At termination, \Cref{step:inc-cut} ensures that $\incbpe$ is nonempty, since $\out<i$.
    Define
    \begin{equation*}
        \theta = \incbpe.\phi = \phi(x_{\out+1}\dots x_{\cut}) \text{,}
        \qquad
        \tau = \bpequeue.\phi = \phi(x_{\out+1}\dots x_i) \text{.}
    \end{equation*}
    The first equality is justified by \eqref{eq:incbpe-current-query}.
    Both queries are non-bottom, so termination of \Cref{step:while-match} implies $\theta\neq\tau$.
    
    Since $b(i) \leq \out < i$, there is a live pretoken $x_{b(i) + 1} \dots x_{\ell(i)}$ with 
    \begin{equation*}
        q=\delta(\start,x_{b(i)+1}\dots x_i)
          \notin\dead \text{.}
    \end{equation*}
    Choose a nonempty continuation $z$ such that $\delta(q, z) \in \accept$.
    Note that all pretokens preceding $b(i)$ are dead, while $x_{b(i)+1}\dots x_i z$ is accepted as a single pretoken under the continuation $z$.
    
    For the empty continuation, the next token after the already-emitted output is $\theta$.
    For continuation $z$, the boundary at $\out$ survives by soundness. 
    The next token is
    \begin{equation*}
        \phi(x_{\out+1} \dots x_i z)
        =\phi(x_{\out+1}\dots x_i)
        =\tau \text{,}
    \end{equation*}
    where the first equality follows from \Cref{clm:first-token-equal}.
    Hence, if $P$ denotes the already emitted output,
    \begin{equation*}
        P \circ\theta\preceq\Merge(\Split(x_{\leq i})) \text{,}
        \qquad
        P \circ\tau\preceq\Merge(\Split(x_{\leq i}z)) \text{.}
    \end{equation*}
    Because $\theta\neq\tau$, these two tokenizations have longest common prefix exactly $P$.
    Together with soundness, this proves
    \begin{equation*}
        P = (\Merge\circ\Split)_{\preceq}(x_{\leq i}) \text{.}
    \end{equation*}
    In particular,
    \begin{equation*}
        (\Merge\circ\Split)_{\preceq}
            (x_{\leq i-3\merges\longest})
        \preceq P 
    \end{equation*}
    as required.

    \paragraph{Termination Condition.}
    It remains to prove that the algorithm terminates on each iteration.
    \Cref{step:dead-loop} terminates since $\out$ increases in every iteration of the while loop.
    In particular, since $\out < b(i)$, we have $\out < r(0)$ so that $\incbpe$ is not empty.
    \Cref{step:while-match} terminates as each 
    iteration pops at least $1$ character from both $\incbpe$ and $\bpequeue$, so the loop must terminate when $|\bpequeue| < 3 \merges \longest$.
    \Cref{step:inc-cut} must terminate whenever \Cref{step:while-match} terminates with a mismatch, in which case $\incbpe \neq \epsilon$.
    Otherwise, we increment $\cut$ until eventually $\out = i$, and the algorithm terminates.
    
    \paragraph{Performance}
    Finally, we bound the space and runtime of our algorithm.
    To bound the space, we observe that similar to \Cref{lem:pretokenize-alg}, we can bound the space of $\cutQ$ by $O(\min(n, \lookback) + \merges \longest)$.
    This follows as 
    \begin{equation*}
        \out \geq \ell(i) - 3 \merges \longest \geq i - \lookback - 3 \merges \longest \text{.}
    \end{equation*}
    The first inequality follows from the fact that if $\out < \ell(i) - 3 \merges \longest$, then $\bpequeue$ and $\incbpe$ share a prefix of length at least $3 \merges \longest$, so that \Cref{step:while-match} will output a token.
    The second inequality follows from the definition of $\lookback$.
    Similarly, the space of $\bpequeue$ and $\incbpe$ is $O(\min(n, \lookback) + \merges \longest)$.
    The space of $\buffer$ is $O(i - \cut) = O(\min(n, \lookback))$ as $\cut \geq \ell(i)$.
    
    To bound the time, our algorithm does $O(1)$ work per input in addition to various calls to the data structures.
    Combining \Cref{lem:incremental-bpe}, \Cref{lem:bpe-queue}, and \Cref{lem:pretok-q} gives an algorithm with $O(\merges  \longest)$ preprocessing and $O(n (\log^2 \longest + \log \merges  + \log \min(n, \lookback)))$ total time.
    This concludes the proof of \Cref{thm:online-tokenization-const-delay}.
\end{proof}

%% file: experiments.tex
\section{Experimental Setup}
\label{sec:experimental-setup}

We compare our end-to-end tokenizer on a selection of encodings from \python{tiktoken} \citep{tiktoken} and \python{tokenizers} \citep{huggingface}.
Our code can be found 
\ificlrfinal
at \url{https://github.com/morganstanley/MSML/tree/main/papers/Online_Tokenization_with_Optimal_Delay_in_Near_Linear_Time}.
\else
in the supplementary material.
\fi

We test our tokenizer on similar datasets as \citet{jiang2025incrementalbpe}.\footnote{We note that we do not explicitly compare our tokenizer against \citet{jiang2025incrementalbpe}, since their code supports only BPE, whereas our goal is to benchmark the full end-to-end tokenization pipeline.}
In particular, we use the following three datasets:
\begin{enumerate}
    \item English and Chinese - We utilized the first Parquet shard from the standard training split of the Wikipedia dataset (20231101 dump) \citep{wikimedia_wikipedia_2023b}.
    The filtered dataset is obtained by a strided sampling strategy.
    For English, we keep documents 0, 5, 10, and 16 in every block of 21 documents.
    For Chinese, we keep documents 0 and 6 in every block of 13 documents.
    
    \item Github - We extracted data from the RedPajama GitHub subset \citep{together_redpajama_2023b}, specifically using the file with the filename \rust{"filtered_08cdfa755e6d4d89b673d5bd1acee5f6.sampled.jsonl"}. 
    The filtered dataset contains every 17th document.
\end{enumerate}

All experiments were conducted on an Intel Xeon Platinum 8462Y+, comprising 32 cores.
All computations were performed exclusively on the CPU and executed on a single core.
The system was otherwise idle during benchmarking, with no other user processes competing for CPU resources.

\subsection{Applications: R50K, CL100K, and O200K.}
\label{sec:backtrack-bound}

We study the full $\rfk$, $\clok$, $\otk$ reference patterns below.
The three patterns: $\rfk$ (used by both $\rfk$ and $\pfk$ encodings), $\clok$, and $\otk$ are given below:

{\scriptsize
\begin{minted}{rust}
    R50K = r"'(?:[sdmt]|ll|ve|re)| ?\p{L}+| ?\p{N}+| ?[^\s\p{L}\p{N}]+|\s+(?!\S)|\s+";

    CL100K = concat!( 
        r"'(?i:[sdmt]|ll|ve|re)|[^\r\n\p{L}\p{N}]?+\p{L}++|\p{N}{1,3}+|",
        r" ?[^\s\p{L}\p{N}]++[\r\n]*+|\s++$|\s*[\r\n]|\s+(?!\S)|\s"
    );

    O200K = concat!(
        r"[^\r\n\p{L}\p{N}]?[\p{Lu}\p{Lt}\p{Lm}\p{Lo}\p{M}]*[\p{Ll}\p{Lm}\p{Lo}\p{M}]",
        r"+(?i:'s|'t|'re|'ve|'m|'ll|'d)?",
        r"|[^\r\n\p{L}\p{N}]?[\p{Lu}\p{Lt}\p{Lm}\p{Lo}\p{M}]+[\p{Ll}\p{Lm}\p{Lo}\p{M}]",
        r"*(?i:'s|'t|'re|'ve|'m|'ll|'d)?",
        r"|\p{N}{1,3}",
        r"| ?[^\s\p{L}\p{N}]+[\r\n/]*",
        r"|\s*[\r\n]+",
        r"|\s+(?!\S)",
        r"|\s+"
    );
\end{minted}
}

Each regex is an alternation of several cases.
We explain some of the illustrative cases below (focusing on $\rfk$ and $\clok$, the interested reader is encouraged to investigate in detail with a tool such as \url{https://regex101.com}).

\newsavebox{\rustbox}

\begin{lrbox}{\rustbox}
  \rust{r"\s+$"}
\end{lrbox}

\begin{enumerate}
    \item \rust{r"'(?i:[sdmt]|ll|ve|re)"} matches common apostrophed suffixes such as \rust{"'ll"}.
    \item \rust{r" ?\p{L}+"} matches consecutive letters with optional leading whitespace, e.g. \rust{" hello"}. 
    \item \rust{r" ?\p{N}+"} matches consecutive digits with optional leading whitespace, e.g. \rust{" 123"}.
    We note that $\clok$ caps the length of a run of digits at 3. 
    \item \rust{r" ?[^\s\p{L}\p{N}]+"} matches consecutive symbols with optional leading whitespace e.g. \rust{" !?!?"}.
    \item \rust{r"\s+(?!\S)"} matches consecutive white space such that the \emph{next} character is a space or end of input.
    \label{pat:leading-space}
    \item \usebox{\rustbox} matches a consecutive whitespace until the end of the file.
    Throughout, we use \$ as a special character to denote EOF.
    \label{pat:space-eof}
    \item \rust{r"\s*[\r\n]"} matches consecutive whitespace ending in a newline.
    \label{pat:newline-space}
    \item \rust{r"\s+"} and \rust{r"\s"} matches consecutive whitespace and single whitespace respectively.
    \label{pat:space}
\end{enumerate}

In these reference patterns, matching is computed in order of priority, then length.
In particular, given a regular expression $R = R_1 \mid R_{2} \mid  \dots \mid R_{i}$ the \emph{first} match of $R$ in $x$, denoted $F'(x; R)$ is defined as $x_{\leq J'(x; R)}$ where
\begin{equation*}
    J'(x; R) = \max \set{j > 0 \mid x_{\leq j} \in L(R_{i}) \text{ and no non-empty prefix of $x$ is in $L(R_{k})$ for $k < i$}} \text{.}
\end{equation*}
That is, instead of just taking the longest prefix that matches any rule, we take the highest priority rule that has any match with a prefix of $x$, and the longest match under that rule.
We observe that given the left-to-right priority of the alternation rules in the regex pattern, a string is accepted by a rule if (1) it matches the pattern of that rule and (2) it has no prefix that matches any higher priority rule.
Note that (online) pretokenization (\Cref{def:pretokenization}) is defined identically, replacing $F$ with $F_{\refT}$.

Beyond left-to-right priorities, practical patterns also implement the lookahead operations.
In particular, $\clok$ and $\otk$ contain the rule \rust{r"\s+(?!\S)"} which matches runs of whitespace that are followed by whitespace or end of input (but not runs of whitespace followed by a non-whitespace character).
In this case, a DFA may need to scan ahead by one character to accept a prefix.
For example, if $x = \rust{"\t\t\ta"}$, the DFA must scan $\rust{"\t\t\t"}$ before accepting $\rust{"\t\t"}$.
We note that even with the priority rules and lookahead operation the $\clok$ and $\otk$ examples can still be represented with simple DFAs (see \Cref{app:omitted}).
In our examples, the DFA needs to scan ahead by at most one character.

With both priority and lookahead, let $F_{\refT}(x;R)$ be the prefix of $x$ matched by $R$ incorporating all priority and lookahead rules.
Define $J_{\refT}(x;R)=|F_{\refT}(x;R)|$.

In view of priority and lookahead, we give a generalized definition of lookback (\Cref{def:lookback}).

\begin{definition}
    \label{def:ref-lookback}
    A regular expression $R$ has $\lookback_{\refT}$-generalized-lookback if for any $x \in \Sigma^*$ and $a \in \Sigma$, $|J_{\refT}(xa) - J_{\refT}(x)| \leq \lookback_{\refT}$.
\end{definition}

First, one can see that when $R$ has no priority rules (i.e. just one rule) and no lookahead, $|J_{\refT}(xa) - J_{\refT}(x)| = J(xa) - J(x) \geq 0$ is only non-zero when $xa \in L$, in which case $J(x)$ is exactly the longest match among strict prefixes, so that $\lookback_{\refT}(R) = \lookback(R)$.

For the reference patterns considered here, we adapt the simple backtracking algorithm of \Cref{lem:segment-alg} to compute pretokenization on a pattern with priority rules and one character lookahead.
When $\lookback_{\refT}(R)$ is finite, a similar proof can show that the structure of the DFAs ensure that each restart reprocesses only $O(\lookback_{\refT}(R))$ characters.
Charging each restart to progress through the input gives $O(n \lookback_{\refT}(R))$ total time.
We use this approach in practice because it is simple and avoids the overhead of the more general data structure in \Cref{lem:pretokenize-alg}.
We further justify this choice by showing that in most practical cases, $\lookback_{\refT}(R)$ is small.

\begin{claim}
    \label{clm:r50k-lookback}
    The $\rfk$ pattern satisfies $\lookback_{\refT}(\rfk) \leq 2$.
\end{claim}

\begin{proof}[Proof of \Cref{clm:r50k-lookback}]
    Appending a character can extend an existing run of letters, digits, or symbols by at $1$. 
    For runs of letters and numbers, appending a character cannot lead to a match with a higher priority rule.
    On the other hand, an initial apostrophe could be extended to a contraction e.g. \rust{"'s"}, which extends the match by at most $1$.
    Finally, note that a character can extend the match to a contraction, for example, $x = \rust{"'v"}$ and $a = \rust{"e"}$, which can extend the first match by at most $2$.
    
    For a run of whitespace, appending a non-whitespace character can shorten the first match by one character because of the lookahead rule, or extend the match to a higher priority rule by at most one character.
    Appending a whitespace character on the other hand, can extend the match by at most one character.
    This exhausts all the cases.
\end{proof}

\begin{claim}
    \label{clm:cl100k-200k-lookback}
    The $\clok$ and $\otk$ patterns have unbounded $\lookback_{\refT}$.
\end{claim}

\begin{proof}[Proof of \Cref{clm:cl100k-200k-lookback}]
    It suffices to produce a single accepted string where any proper prefix that is also accepted is much shorter.
    
    Consider $\clok$. For any $m \geq 1$, consider $x = \rust{"\n(\t)^{m}"}$ and $a$ to be any letter so that $xa$ matches the pattern \Cref{pat:newline-space} and $J_{\refT}(xa) = 1$.
    On the other hand, $x$ matches the pattern \Cref{pat:space-eof} and $J_{\refT}(x) = m + 1$.
    This establishes a gap of size $m$.

    For $\otk$, we consider a similar string $x = \rust{"\n(\t)^{m}"}$ and $a = \rust{"\n"}$.
    Note that $xa$ is accepted by the \eqref{pat:newline-space} (i.e. the fifth rule of $\otk$, but any length $> 1$ prefix of $x$ is not accepted: the first rule it matches, either \eqref{pat:leading-space} or \eqref{pat:space}, is lower priority than \eqref{pat:newline-space}, which is matched by the prefix \rust{"\n"}.
    Thus, $J_{\refT}(xa) - J_{\refT}(x) = m + 1$ as well.
\end{proof}

However, given that $\clok$ and $\otk$ only have unbounded lookback on relatively rare outputs (basically files that are all whitespace), all three patterns have lookback that are typically much smaller than input size on normal text.

\subsection{Optimization}

We now describe some optimizations in our code to improve the practical efficiency of our algorithm.

\paragraph{Chunk Size}
We read input in fixed-size chunks to amortize input overhead, observing that chunks of 64KiB (or larger) are most efficient in practice.
The read chunk size controls input batching.
For each chunk, the unresolved suffix may contain a candidate pretoken that is retained until it is resolved by reading the next chunk.

\paragraph{Pretokenization Algorithm}
We remark that in practice, the simple backtracking algorithm for online pretokenization that requires $O(n\lookback)$ time is more efficient than the $O(n)$ time algorithm given by \Cref{lem:pretokenize-alg}.
In particular, the overhead of the linear-time algorithm depends on the size of the regular expression, which is typically larger than the overhead incurred by backtracking (since in practice $\lookback$ is often very small, see e.g. \Cref{clm:r50k-lookback}).

\paragraph{BPE Algorithm.}
To favor throughput, our implementation invokes the incremental (but not online) BPE algorithm of \citet{jiang2025incrementalbpe}, which only merges after a pretoken boundary is resolved.
This is consistent with \citet{jiang2025incrementalbpe}, which observe that the online BPE algorithm is slower in practice.
The implementation therefore differs here from the theoretical composition algorithm, which emits accepted prefixes of candidate pretokens to an online BPE merging algorithm.
The practical implementation can consequently (in the worst case) retain an arbitrarily long unfinished pretoken and defer applying BPE to merge this pretoken until the pretoken is complete.
Its worst-case memory and latency can therefore depend on pretoken length, even for patterns with bounded lookback.
We note that for the implemented regular expressions, this scenario typically does not occur on natural language, evidenced by the memory and latency advantages of our implementation.

\paragraph{Linear Buffering}
Suppose that we have read in a chunk into a buffer, and our online tokenization algorithm has decided that some small prefix of the chunk can be stably tokenized and emitted as output.
We now need to remove this prefix on the chunk and recurse the tokenizer on the remaining suffix.
Naively draining the buffer (e.g. \rust{std::String::drain}) after each token is emitted requires linear time, since we shift the entire array of characters forwards.
Since tokens are typically much shorter than a chunk, this requires time that is quadratic in the chunk size.
Instead, whenever we emit a token, we simply advance a counter on the buffer to mark which characters have been successfully processed.
At the arrival of a new chunk, we keep the last trailing token (typically quite short) and drain the processed prefix in time linear in the chunk size.

\paragraph{DFA Implementation}
A key optimization of our code is a faster implementation of the regular expression matching required by pretokenization.
Indeed, regex parsing is the dominant cost of tokenization (see \citet{tiktoken_lib_rs}).
Instead of directly implementing the pretokenization patterns (e.g. $\rfk$, $\clok$) as a DFA, we use a combination of (1) a simpler pattern that can be efficiently implemented as a DFA, and (2) an auxiliary $\applyFixup$ function that ensures our algorithm exhibits identical behavior to a DFA directly implemented on the original pattern.

We note that the choice of DFA implementation (either the lazy DFA or dense DFA in \rust{regex_automata}) once the pattern is fixed does not lead to any significant differences in performance.
Furthermore, since we work exclusively with total regular expressions, we use anchored DFAs, which ensure that the first character is matched (removing the overhead caused by searching for alternative starting characters).
In \Cref{sec:data-validation}, we present theoretical and empirical evidence that our alternative implementation is in fact equivalent to a direct implementation of the DFA.

\paragraph{Segmenting}
Finally, we optimize the special token segmenting.
Since special tokens typically start with the same character (e.g. \rust{<}) and this character is relatively rare in natural text, we use \rust{memchr} \citep{burntsushi_memchr_2026}, an optimized library for string search in Rust to efficiently find instances of the starting characters.

\subsection{Data Validation: DFA Optimization}
\label{sec:data-validation}

For \texttt{tiktoken}, we provide theoretical and empirical evidence that our DFA optimizations preserve the same pretoken matching and splitting as the reference patterns. We consider the $\rfk$, $\clok$, and $\otk$, beginning with their simplified forms.
{\scriptsize
\begin{minted}{rust}
    R50K_SIMP = r"'(?:[sdmt]|ll|ve|re)| ?\p{L}+| ?\p{N}+| ?[^\s\p{L}\p{N}]+|\s+";

    CL100K_SIMP = concat!( 
        r"'(?i:[sdmt]|ll|ve|re)|[^\r\n\p{L}\p{N}]?\p{L}+|\p{N}{1,3}|",
        r" ?[^\s\p{L}\p{N}]+[\r\n]*|\s*[\r\n]|[^\S\r\n]+"
    );
\end{minted}
}
Similarly to $\rfk$, the simplified $\otk$ pattern is identical to the reference $\otk$ pattern with the \rust{r"\s+(?!\S)"} rule removed.
The $\applyFixup$ function is described as follows.
Every time the DFA finds a resolved match in the substring (in the language of \Cref{lem:pretokenize-alg}, this corresponds to outputting $\brk$), we call the $\applyFixup$ function on the candidate pretoken.
Note that if we reach the end of the current chunk in buffer and the DFA is not in an $\accept$ or $\dead$ state, we simply suspend the DFA computation until the arrival of the next chunk and do not call $\applyFixup$.

\begin{mdframed}
    {\bf $\applyFixup(x, j; R)$}
    
    \begin{enumerate}
        \item {\bf Input:} remaining input $x$ and candidate boundary $j$.
        
        \item {\bf extend:} if $x_{j}$ is a newline, $x_{\leq j}$ is all whitespace, and the remaining suffix of $x$ through EOF is all horizontal whitespace, then extend the pretoken to the end of the buffer (equivalently return $x$).

        \Comment{(extend is only applied for the $\clok$ pattern.)}

        \item {\bf trim:} if $j > 1$ and $x_{j}$ is whitespace, $j < |x|$ and next character $x_{j + 1}$ is not whitespace, then return $x_{< j}$.

        \Comment{(for the $\clok$ and $\otk$ patterns, we check if $x_{j}$ is horizontal whitespace.)}

        \item Otherwise, return $x_{\leq j}$.
    \end{enumerate}
\end{mdframed}

We now show that the simplified patterns along with $\applyFixup$ as described above is equivalent to the original pattern.

\paragraph{Proof of Equivalence}

Observe the following crucial property of $\Split$: any algorithm that successfully computes $F_{\refT}(x; R)$ immediately gives an algorithm that computes $\Split$.
The proof is essentially immediate from the definition of pretokenization (see \Cref{def:pretokenization}).
Recall that given a string $x$ and a total expression $R$, $F_{\refT}(x; R)$ is the first match of $R$ in $x$ starting at the first character.
Note $F_{\refT}(\epsilon; R) = \epsilon$.

\begin{claim}[Local Reduction]
    \label{clm:local-reduction}
    Let $R$ be a total regular expression.
    Suppose there is a function $f_0$ such that $f_0(x) = F_{\refT}(x; R)$ for all $x \in \Sigma^*$.
    Then, $f(\epsilon) = \epsilon$ and $f(x) := f_0(x) \brk f(x_{>|f_0(x)|})$ for $|x| > 0$ computes $\Split(x)$ for all $x \in \Sigma^*$.
\end{claim}

Thus, it suffices to argue that the first match is identical for all strings.
Before doing so, we formally describe our algorithm $\FastMatch(x; R)$.
For regular expression $R$ (e.g. $\rfk$), let $\tilde{R}$ denote the simplified pattern.

\begin{mdframed}
    {\bf $\FastMatch(x; R)$}
    \begin{enumerate}
        \item {\bf Input:} string $x \in \Sigma^*$.
        \item If $x = \epsilon$, return $\epsilon$.
        \item Compute $J_{\refT}(x; \tilde{R})$.
        \item Output $\applyFixup(x, J_{\refT}(x; \tilde{R}); R)$.
    \end{enumerate}
\end{mdframed}

\begin{proposition}
    \label{prop:first-match-equivalence}
    For all $x \in \Sigma^*$ and pattern $R$ in $\rfk$, $\clok$, and $\otk$, we have
    \begin{equation*}
        \FastMatch(x; R) = F_{\refT}(x; R) \text{.}
    \end{equation*}
\end{proposition}

\begin{proof}[Proof of \Cref{prop:first-match-equivalence}]
    We analyze each pattern separately, beginning with $\rfk$.

    \paragraph{R50K.}
    Note that the first four rules are identical in the original and simplified patterns.
    None of these patterns trigger the {\bf trim} subroutine of $\applyFixup$, since the last character is not a whitespace.
    Therefore, the returned match is identical, if original pattern first matches one of these rules.

    Suppose the fifth rule \rust{r"\s+(?!\S)"} matches in the original pattern. 
    Let $w$ be the maximal initial run of whitespace.
    If $w$ reaches EOF, both the original and simplified patterns match all of $w$, and {\bf trim} does not trigger because there is no following character.
    Otherwise, $w$ is followed by a non-whitespace character.
    For the fifth rule to match, $|w| \geq 2$, and the original pattern matches $w$ without its final character. The simplified pattern matches all of $w$, and {\bf trim} removes its final character, yielding the same match.

    Finally, suppose the sixth rule matches in the original pattern.
    Since the fifth rule did not match, the initial whitespace run consists of a single character followed by a non-whitespace character. 
    This whitespace cannot be a literal space, since one of the preceding letter, number, or symbol rules would then match. 
    The simplified pattern matches the same single whitespace character, and \texttt{trim} does not trigger because the candidate has length one.

    \paragraph{CL100K.}
    Consider now $\clok$. 
    The first rule is identical in the original and simplified pattern.
    For the second, third, and fourth rules, we claim that removing the possessive quantifier \rust{"+"} does not change the match.
    First, a general characterization of when possessive quantifiers change the output of a match.

    \begin{claim}
        \label{clm:possessive-equiv}
        Let $A_0, B_0$ be disjoint sets of literals in $\Sigma$.
        Consider a regular expression $R_1$, defined by \rust{r"A++B"} and $R_2$, defined by \rust{r"A+B"}.
        If $A \subseteq A_0^*$ and $B \subseteq B_0^*$, then $F_{\refT}(x; R_1) = F_{\refT}(x; R_2)$ for all $x \in \Sigma^*$.
    \end{claim}

    \begin{proof}[Proof of \Cref{clm:possessive-equiv}]
        We begin with a brief refresher on possessive quantifiers (see e.g. \citet{oracleJavaRegexQuantifiers, regularExpressionsInfoPossessive}).
        Given a pattern \rust{r"A++B"}, a possessive quantifier tries to (1) find the longest match satisfying \rust{r"A+"}, and (2) the longest match of \rust{r"B"} in the remaining string. 
        If no such match exists, then the possessive quantifier reports that no match is found.
        In contrast, if no such match exists, the greedy quantifier given by \rust{r"A+B"} backtracks to the next longest prefix matching \rust{r"A+"}, and attempts to match \rust{r"B"} from there.
        As an example, \rust{r".++x"} would fail to match \rust{"abcx"}, since (1) matches the whole string leaving no characters for (2), whereas \rust{r".+x"} would find a match by backtracking one character.
        
        Suppose for contradiction that $F_{\refT}(x; R_1) \neq F_{\refT}(x; R_2)$ for some $x$.
        Let $a_1 = F_{\refT}(x; \rust{r"A+"})$ be the longest prefix match of \rust{r"A+"} in $x$.
        Then, there is no match for \rust{r"B"} in $x \setminus a_1$.
        On the other hand, there is a prefix $a_2 \prec a_1$ such that $x \setminus a_2$ contains a match for \rust{r"B"}.
        In particular, a prefix of $x \setminus a_2$ is a suffix of $a_1$, a contradiction since \rust{r"A"} and \rust{r"B"} are over disjoint sets of literals.
    \end{proof}
    
    \Cref{clm:possessive-equiv} applies to the possessive symbol run followed by the disjoint newline class in the fourth rule.
    In the second rule, the optional possessive prefix excludes letters, so giving back its character cannot enable the following letter match.
    The bounded digit pattern in the third rule and the remaining possessive quantifiers in these three rules occur at the ends of their rules, where no subsequent sub-pattern can require backtracking.
    Thus removing the possessive modifiers preserves the matches of the second, third, and fourth rules.

    Consider the fifth rule \eqref{pat:space-eof}.
    In this case, the simplified pattern matches \eqref{pat:newline-space} or the final case of all horizontal whitespace.
    In the former case, we trigger {\bf extend} since the match ends with a newline, is all whitespace, and the remaining string is all whitespace, so we extend the pretoken to the whole string.
    In the latter case, the entire input string is all horizontal whitespace followed by EOF, so that neither {\bf extend} or {\bf trim} trigger, and we return the full matched string, as desired.
    
    For the sixth rule, observe that the sixth rule occurs in the simplified string, and since we did not match the fifth rule, {\bf extend} will not trigger since the remaining string is not all whitespace.
    Since $x_{j}$ is not horizontal whitespace, {\bf trim} does not trigger either.

    Consider the seventh rule, i.e. a run of whitespace with no newlines (since the sixth rule did not match) such that the following character is a horizontal whitespace (but there is some non-whitespace character afterwards since the fifth rule did not match).
    In this case, we match the last rule of the simplified pattern (which matches all runs of horizontal whitespace).
    Since we matched also the trailing horizontal whitespace {\bf trim} will remove this to produce an identical match to the original pattern.

    In the final rule, we have a single whitespace character; a similar argument follows as in the last rule of $\rfk$.

    \paragraph{O200K.}
    As with $\rfk$, the simplified pattern and the original pattern are identical for the first five rules, so we may assume that the first match does not involve the first five rules.
    The fifth rule did not match, hence the run contains no newline.
    In particular, we may assume that the first match is all horizontal whitespace.
    
    Suppose the original pattern matches the sixth rule \rust{r"\s+(?!\S)"}. 
    Let $w$ be the maximal initial run of horizontal whitespace.
    If $w$ reaches EOF, both patterns match all of $w$, {\bf trim} does not trigger.
    Otherwise, $w$ is followed by a non-whitespace character.
    The sixth rule can match only if $|w|\geq 2$, in which case it matches $w$ without its final character. 
    The simplified pattern matches all of $w$, and {\bf trim} removes its final character, yielding the same match.

    Finally, suppose the original pattern matches the seventh rule.
    Since the sixth rule did not match, the initial run consists of a single horizontal whitespace character followed by a non-whitespace character.
    The simplified pattern matches this same character, and {\bf trim} does not trigger because the candidate has length one. This concludes the proof.
\end{proof}

Next, we verify the correctness of our algorithm empirically by enforcing equal output with reference tokenizers on test datasets.

\paragraph{Empirical Validation: Equivalence on Datasets}

Finally, we ensure that our algorithm's output is byte-identical to \python{tiktoken} and \python{tokenizers} on the three given datasets (English, Chinese, and Github).
We also empirically verify that our implementation is byte-identical to \python{tiktoken} and \python{tokenizers} on a large number of short strings.

\section{Experimental Results}
\label{sec:experiment-results}

We now present the results of our empirical evaluations.

\subsection{Time to First Token}
Next, we measure the time to first token (TTFT), a key metric in streaming and online applications.
For each \python{tiktoken} encoding and dataset, we ignore encoder construction time, and measure the time from when the input is provided to the encoder to when the first token is formed.
Note that we measure the time to when the first token is formed, not when it is output, since both \python{tiktoken} and \python{gigatoken} do not output any tokens until all tokens are formed.
For our tokenizer, we provide the input as a file, so that input-reading time is included.
For \python{tiktoken} and \python{gigatoken}, we provide the string as input in memory, thus excluding input-reading time.
We note that \python{tiktoken} and \python{gigatoken} do not natively provide support for TTFT profiling: we make minimal changes to their code to allow for this measurement, as the alternative simpler method of measuring TTFT as the total time of encoding excessively advantages our algorithm (see code for details).
TTFT results are provided in \Cref{tbl:ttft-results}.

\begin{table}[ht]
\caption{Time to first token (TTFT) of all comparison of all tokenizers on \python{tiktoken} encodings measured in milliseconds (ms). We report the median TTFT over 30 runs for each row with upper and lower quartiles in the parentheses.}
\label{tbl:ttft-results}
\begin{center}
\begin{tabular}{|l|l|c|c|c|c|}
\hline
{\bf Encoding} & {\bf Dataset} & \python{tiktoken} & \python{gigatoken} & {\bf Ours} & Speedup \\
\hline
\python{r50k} & Github & 10.91 (10.89-16.77) & 10.70 (10.67-10.74) & {\bf 0.49} (0.45-0.56) & 21.7$\times$ \\
\python{r50k} & English & 10.95 (10.94-10.99) & 10.86 (10.85-10.89) & {\bf 0.61} (0.51-0.84) & 17.8$\times$ \\
\python{r50k} & Chinese & 11.04 (11.00-11.11) & 50.24 (46.48-172.62) & {\bf 0.77} (0.72-0.85) & 14.3$\times$ \\
\hline
\python{p50k} & Github & 10.88 (10.86-10.91) & 10.70 (10.68-10.72) & {\bf 0.56} (0.49-0.61) & 19.2$\times$ \\
\python{p50k} & English & 10.95 (10.93-10.99) & 10.83 (10.81-10.86) & {\bf 0.53} (0.47-0.68) & 20.3$\times$ \\
\python{p50k} & Chinese & 11.01 (10.99-11.04) & 46.63 (31.03-64.26) & {\bf 0.80} (0.72-0.87) & 13.8$\times$ \\
\hline
\python{cl100k} & Github & 19.93 (19.91-19.96) & 18.21 (18.19-18.23) & {\bf 0.52} (0.47-0.56) & 34.8$\times$ \\
\python{cl100k} & English & 6.44 (6.41-6.46) & 6.96 (6.92-7.00) & {\bf 0.58} (0.50-0.70) & 11.1$\times$ \\
\python{cl100k} & Chinese & 6.51 (6.50-6.54) & 236.34 (233.10-238.22) & {\bf 0.82} (0.77-0.86) & 7.9$\times$ \\
\hline
\python{o200k} & Github & 10.62 (10.61-10.67) & 18.78 (18.73-18.82) & {\bf 0.60} (0.56-0.66) & 17.8$\times$ \\
\python{o200k} & English & 10.95 (10.92-10.97) & 7.03 (6.96-7.05) & {\bf 0.60} (0.51-0.72) & 11.8$\times$ \\
\python{o200k} & Chinese & 10.99 (10.98-11.02) & 235.69 (94.11-239.96) & {\bf 0.88} (0.78-0.90) & 12.5$\times$ \\
\hline
\end{tabular}
\end{center}
\end{table}

\subsection{Memory}
Next, we examine the memory required by each tokenizer.
For each encoding and dataset, we measure the peak RSS (Resident Set Size) required by tokenization for input sizes of 1, 2, 4, 8, 16, 32, 64, 128 MiB.
We invoke \python{tiktoken}, \python{tokenizers} and \python{gigatoken} through their public APIs, which produce the output as an array in memory.
Recall that our algorithm streams an input file of characters to an output file of tokens.
The peak RSS is plotted against input size on a $\log$-$\log$ scale.
Across the evaluated datasets of size 1--128\,MiB, our implementation exhibits approximately constant peak RSS, while existing tokenizers show memory growth with input size.
Note that on small inputs, our algorithm will incur large RSS relative to \python{tiktoken} but still smaller than \python{gigatoken} due to initialization costs.

\begin{figure}[ht]
    \centering
    \begin{subfigure}[b]{0.48\textwidth}
        \centering
        \includegraphics[width=\linewidth]{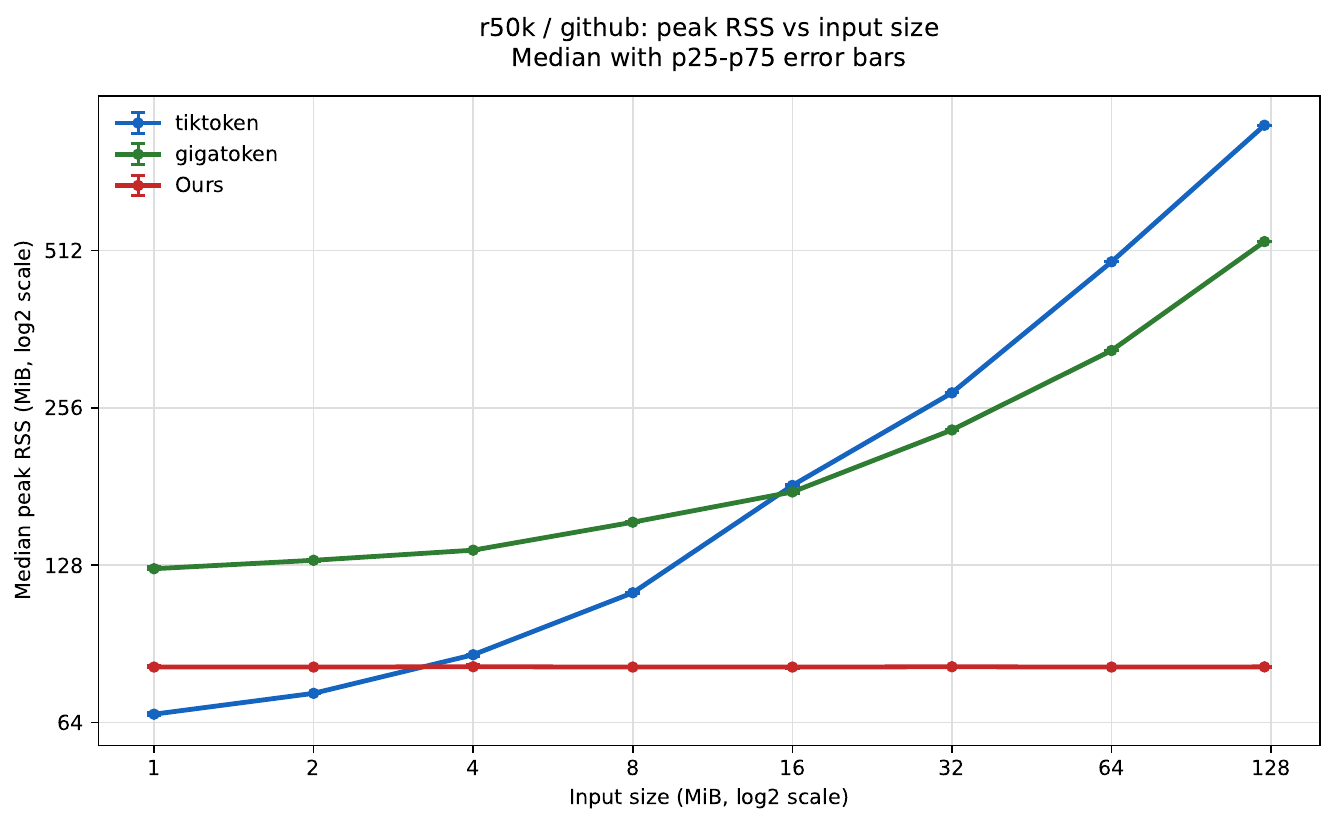}
        \caption{\python{r50k} peak RSS vs. input size}
    \end{subfigure}
    \hfill
    \begin{subfigure}[b]{0.48\textwidth}
        \centering
        \includegraphics[width=\linewidth]{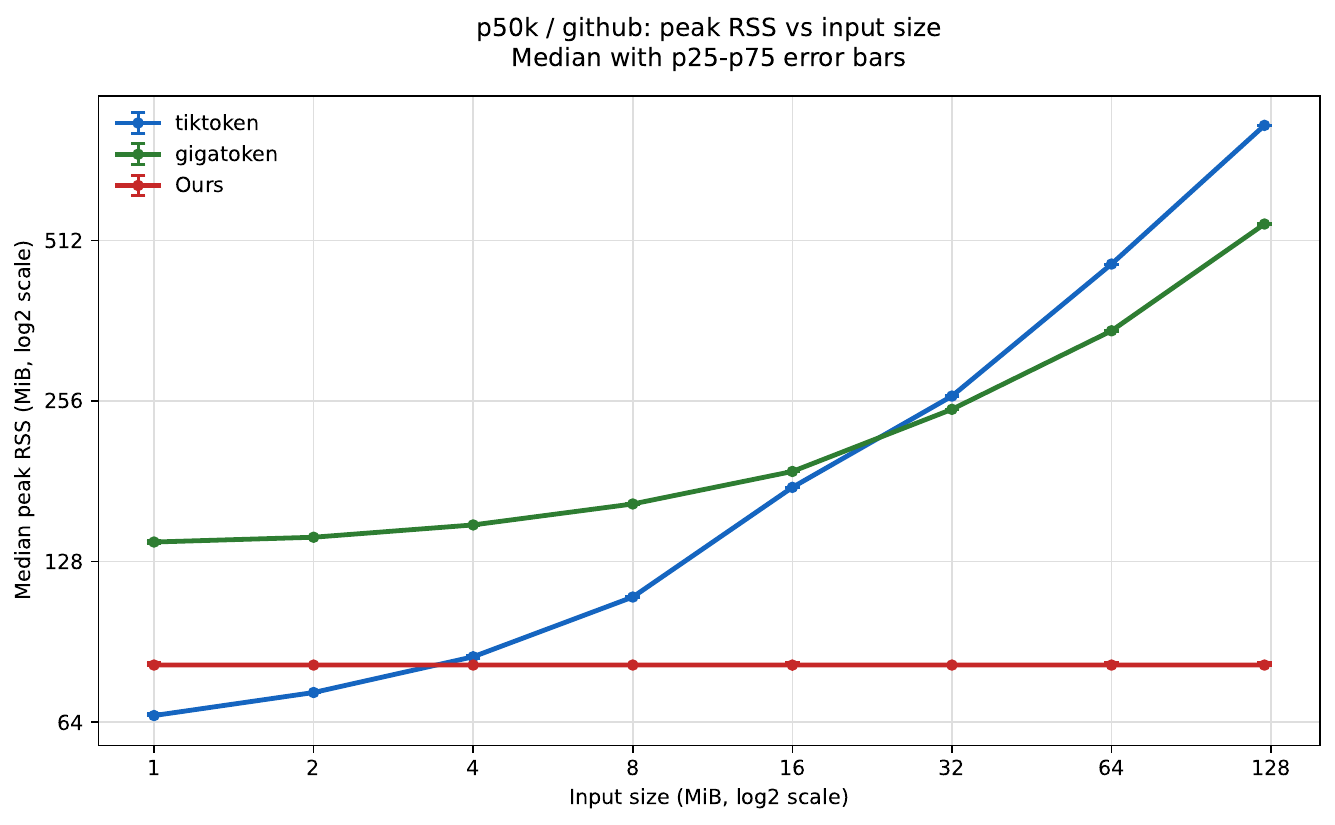}
        \caption{\python{p50k} peak RSS vs. input size}
    \end{subfigure}
    
    \vspace{0.5em}
    
    \begin{subfigure}[b]{0.48\textwidth}
        \centering
        \includegraphics[width=\linewidth]{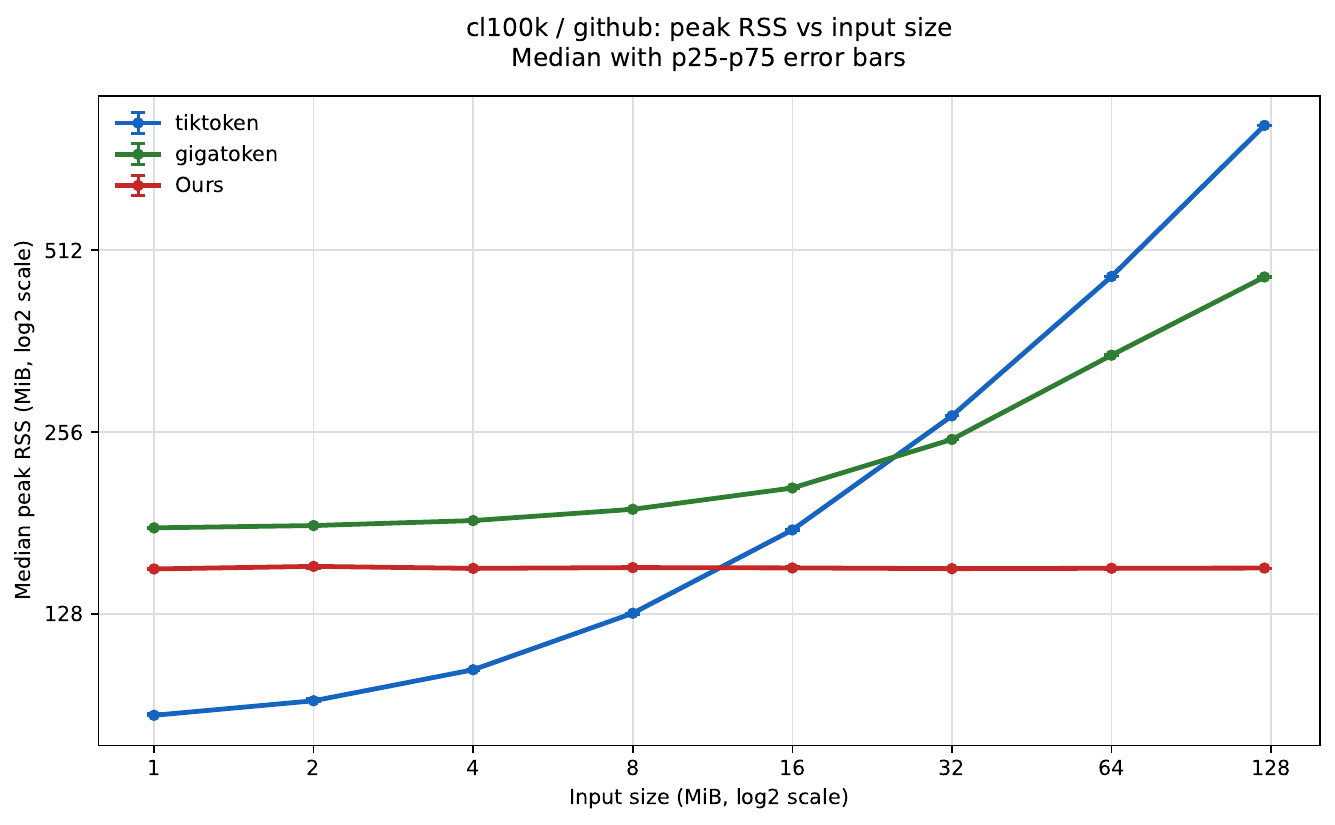}
        \caption{\python{cl100k} peak RSS vs. input size}
    \end{subfigure}
    \hfill
    \begin{subfigure}[b]{0.48\textwidth}
        \centering
        \includegraphics[width=\linewidth]{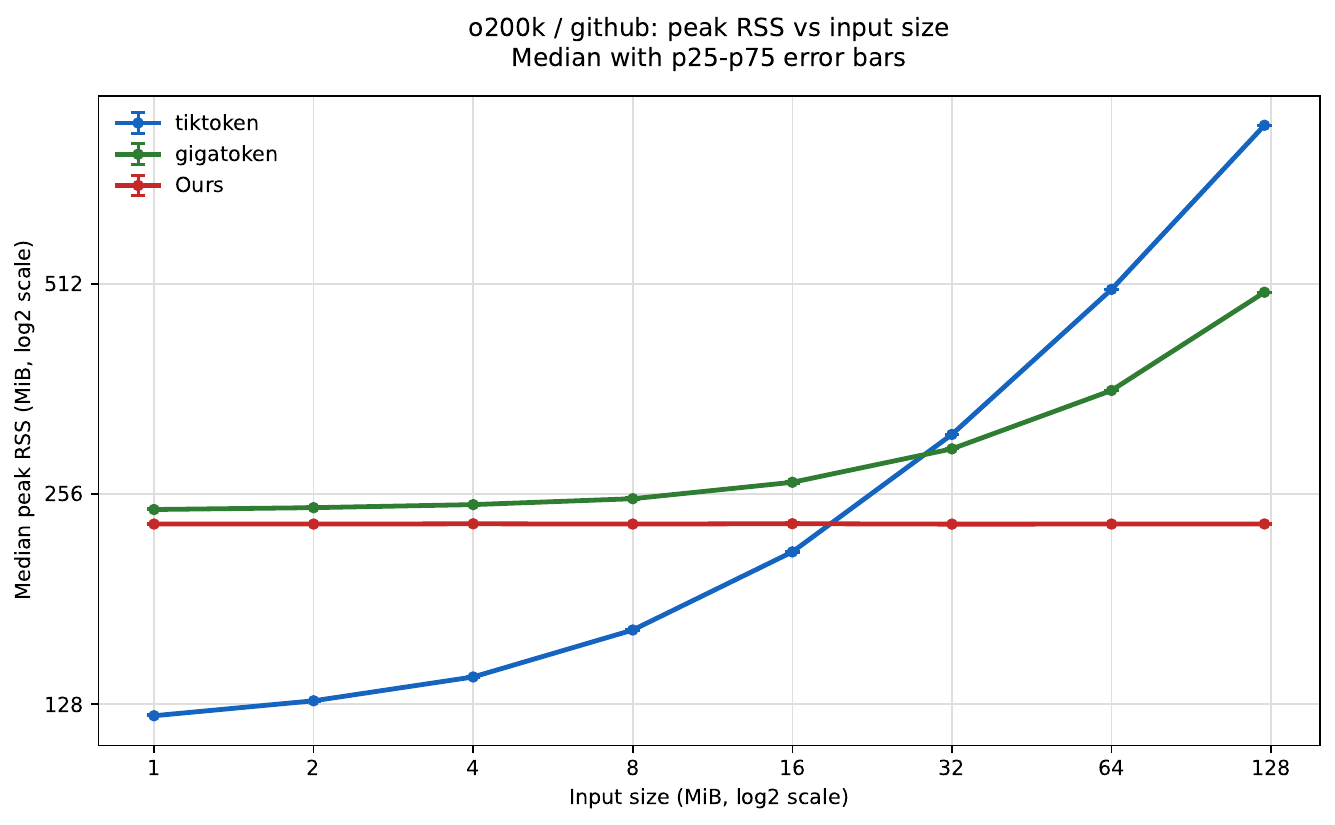}
        \caption{\python{o200k} peak RSS vs. input size}
    \end{subfigure}
    
    \caption{Peak RSS vs. input size on Github datasets. The blue line is \python{tiktoken}, the green line is \python{gigatoken}, while the red line is our algorithm.}
    \label{fig:github-peak-rss}
\end{figure}

\begin{figure}[ht]
    \centering
    \begin{subfigure}[b]{0.48\textwidth}
        \centering
        \includegraphics[width=\linewidth]{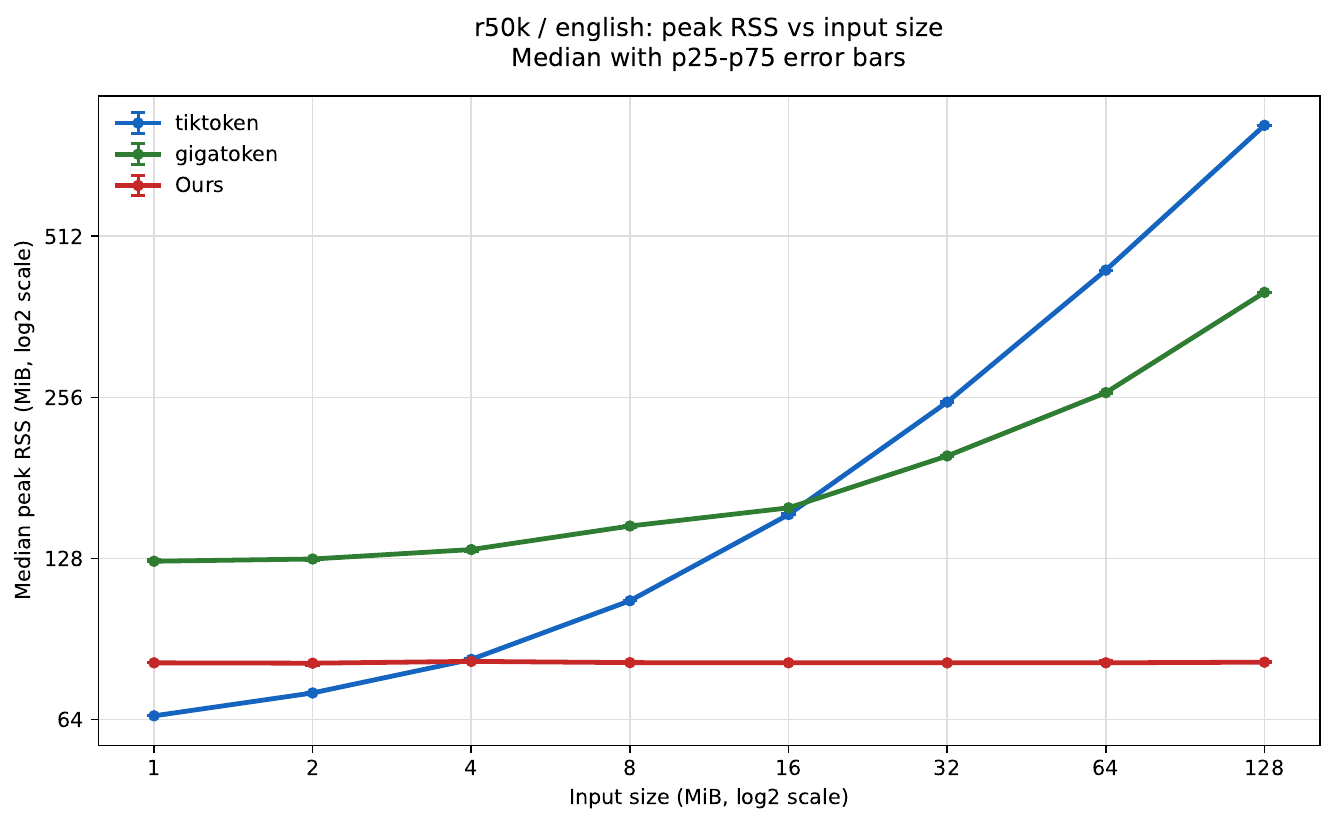}
        \caption{\python{r50k} peak RSS vs. input size}
    \end{subfigure}
    \hfill
    \begin{subfigure}[b]{0.48\textwidth}
        \centering
        \includegraphics[width=\linewidth]{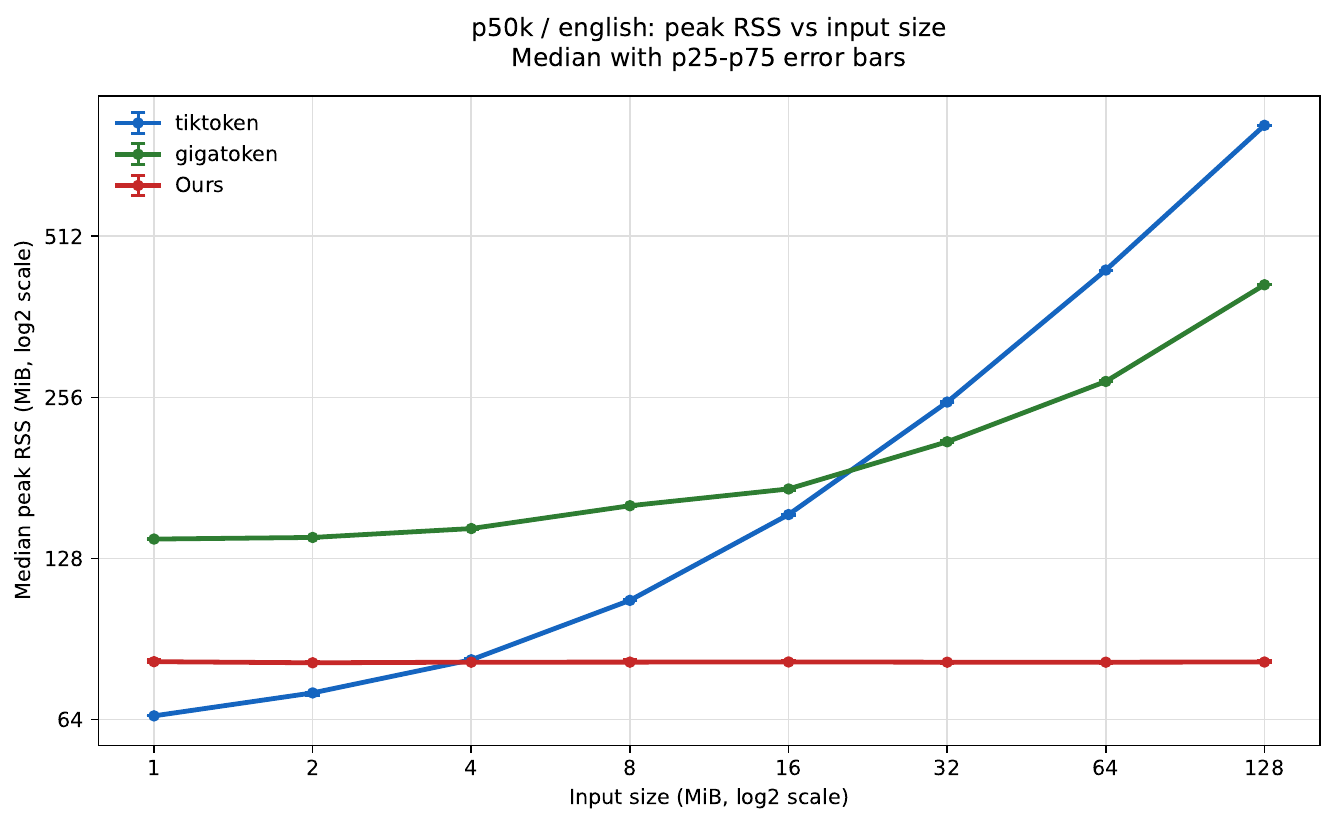}
        \caption{\python{p50k} peak RSS vs. input size}
    \end{subfigure}
    
    \vspace{0.5em}
    
    \begin{subfigure}[b]{0.48\textwidth}
        \centering
        \includegraphics[width=\linewidth]{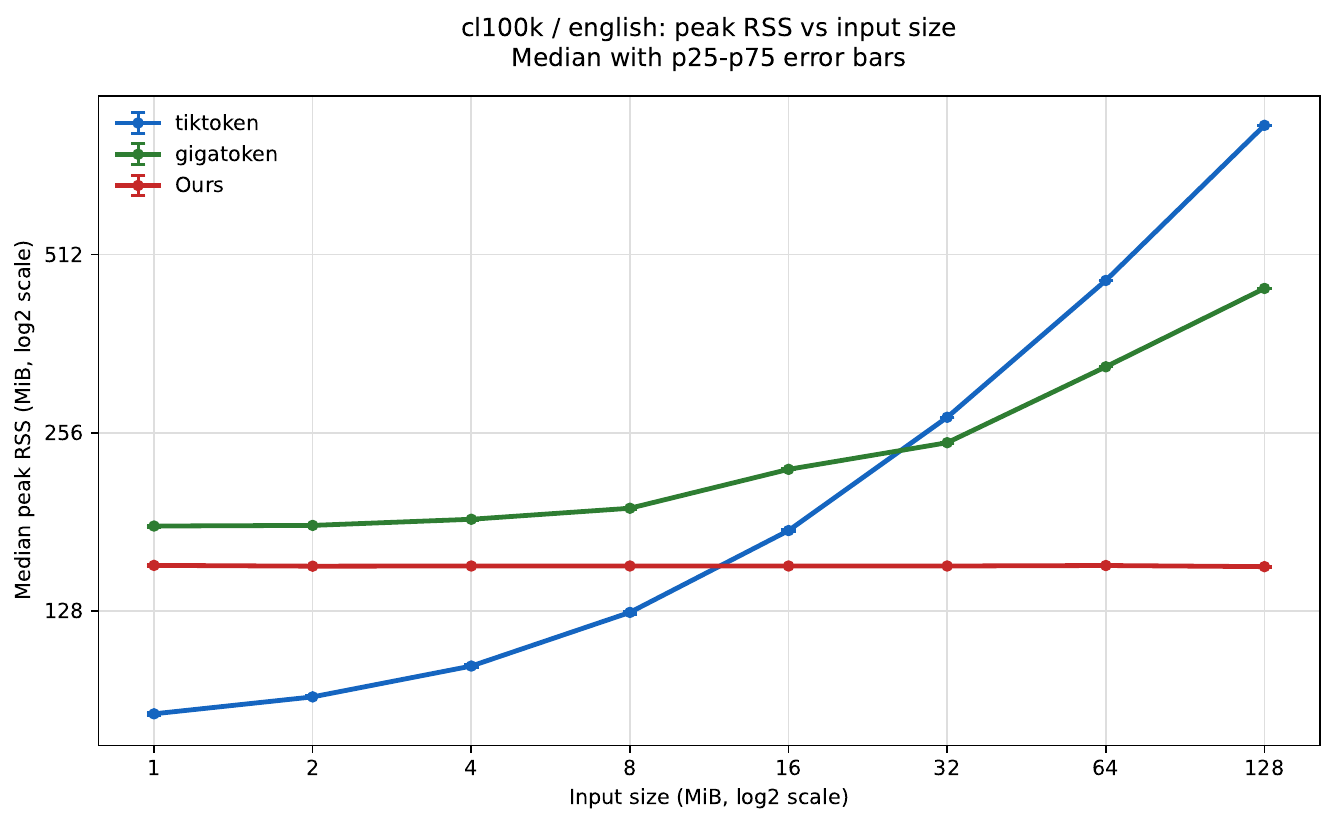}
        \caption{\python{cl100k} peak RSS vs. input size}
    \end{subfigure}
    \hfill
    \begin{subfigure}[b]{0.48\textwidth}
        \centering
        \includegraphics[width=\linewidth]{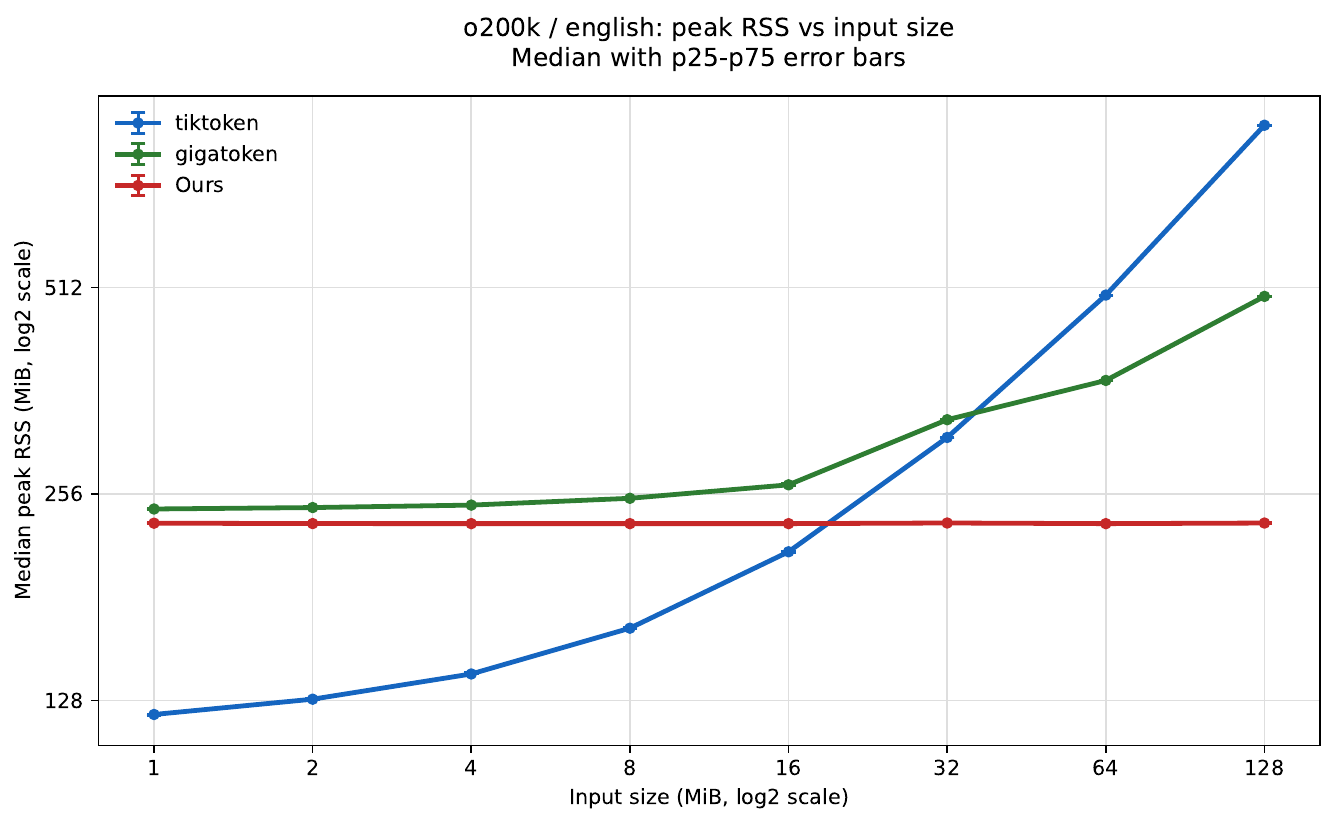}
        \caption{\python{o200k} peak RSS vs. input size}
    \end{subfigure}
    
    \caption{Peak RSS vs. input size on English datasets. The blue line is \python{tiktoken}, the green line is \python{gigatoken}, while the red line is our algorithm.}
    \label{fig:english-peak-rss}
\end{figure}

\begin{figure}[ht]
    \centering
    \begin{subfigure}[b]{0.48\textwidth}
        \centering
        \includegraphics[width=\linewidth]{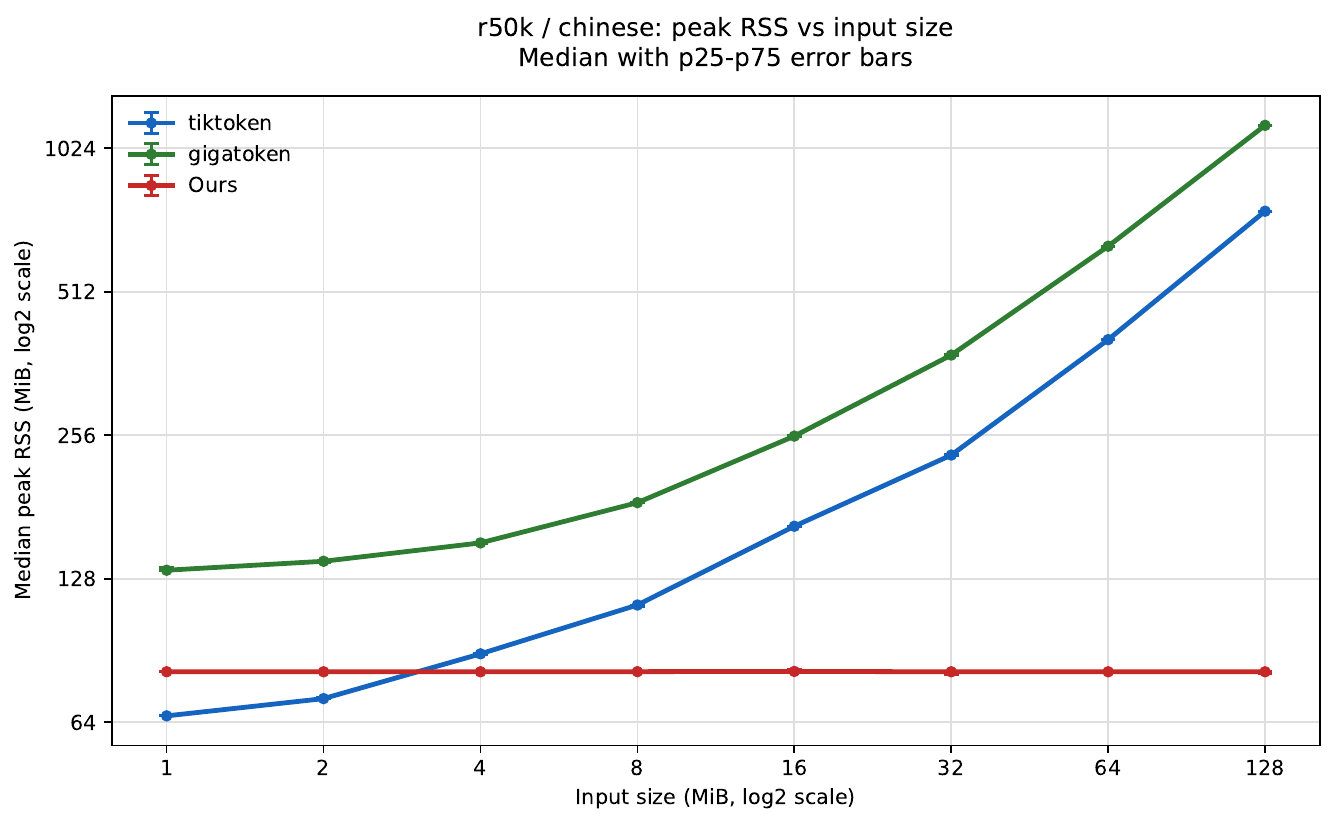}
        \caption{\python{r50k} peak RSS vs. input size}
    \end{subfigure}
    \hfill
    \begin{subfigure}[b]{0.48\textwidth}
        \centering
        \includegraphics[width=\linewidth]{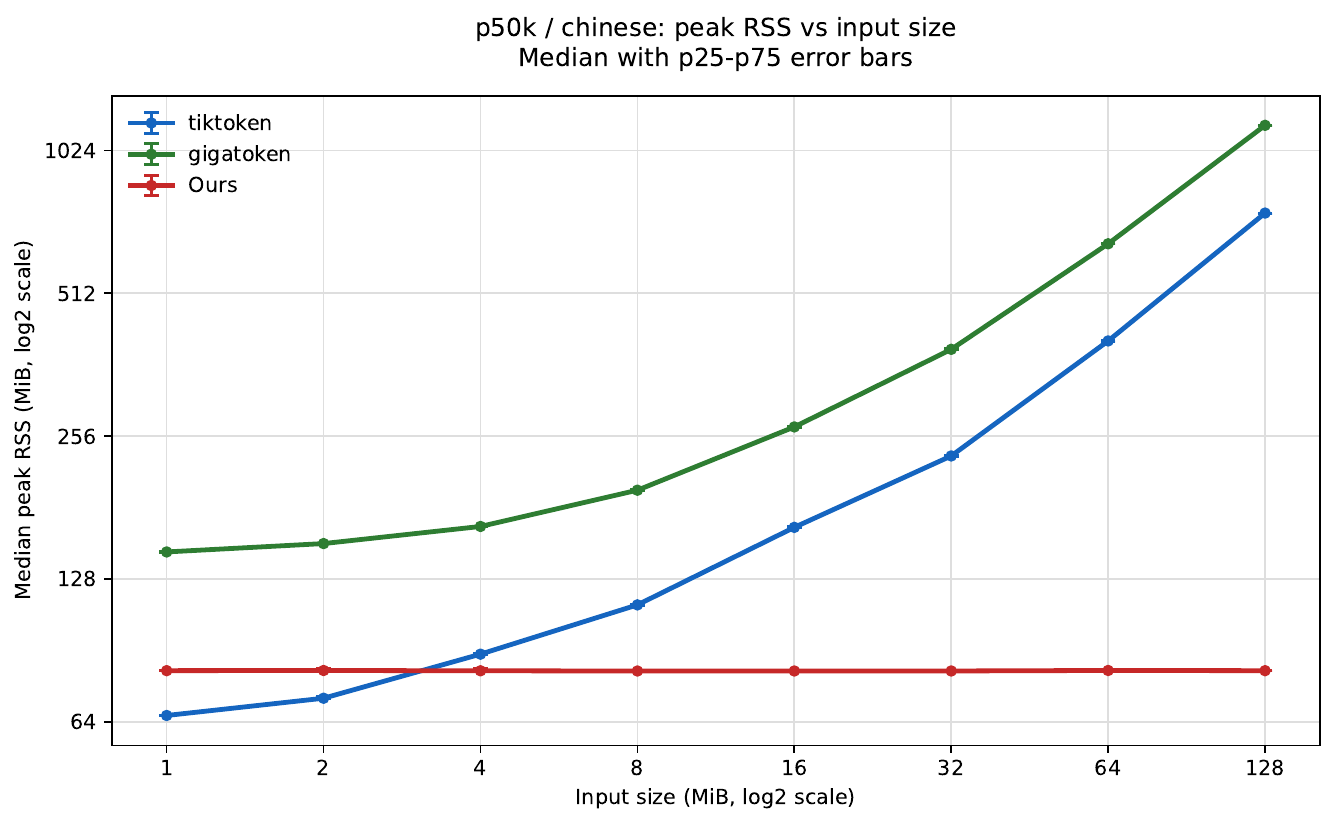}
        \caption{\python{p50k} peak RSS vs. input size}
    \end{subfigure}
    
    \vspace{0.5em}
    
    \begin{subfigure}[b]{0.48\textwidth}
        \centering
        \includegraphics[width=\linewidth]{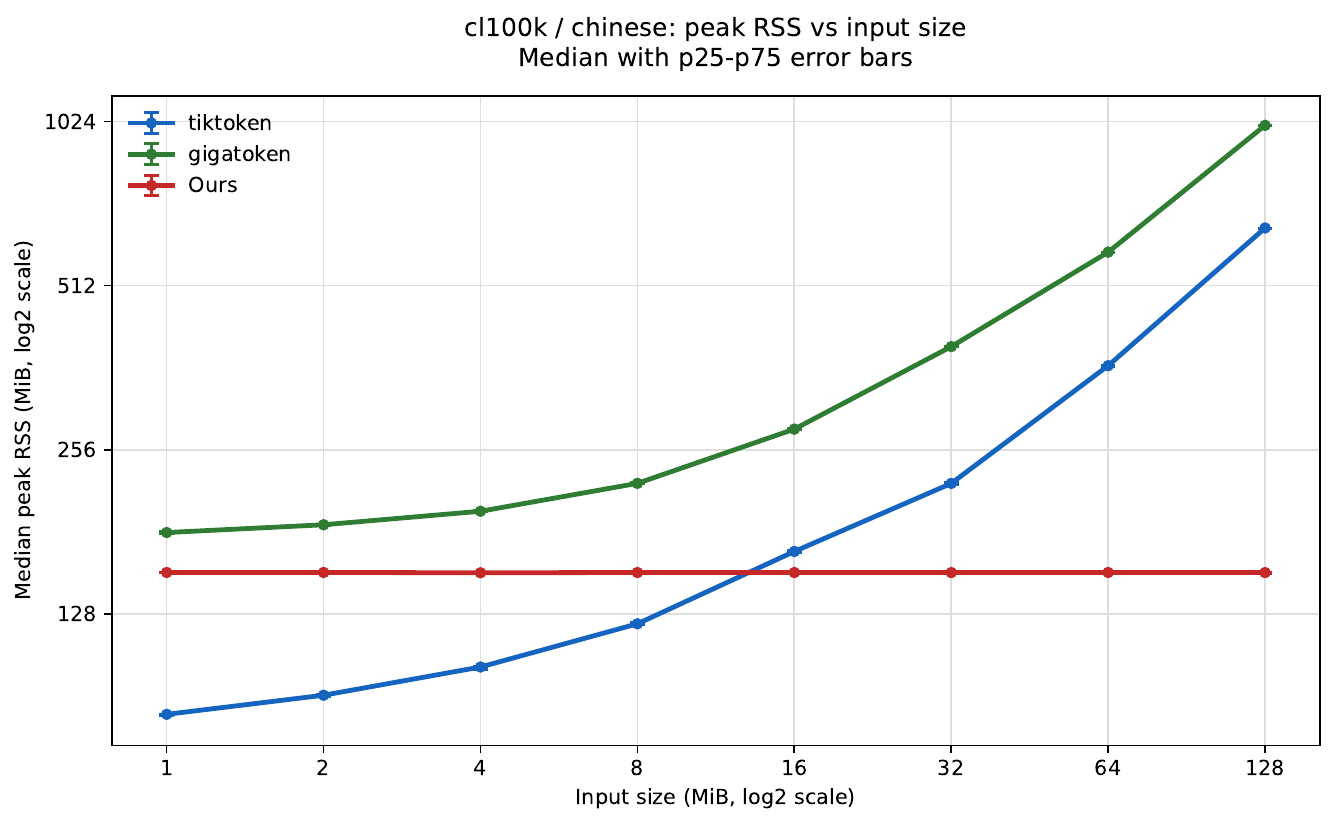}
        \caption{\python{cl100k} peak RSS vs. input size}
    \end{subfigure}
    \hfill
    \begin{subfigure}[b]{0.48\textwidth}
        \centering
        \includegraphics[width=\linewidth]{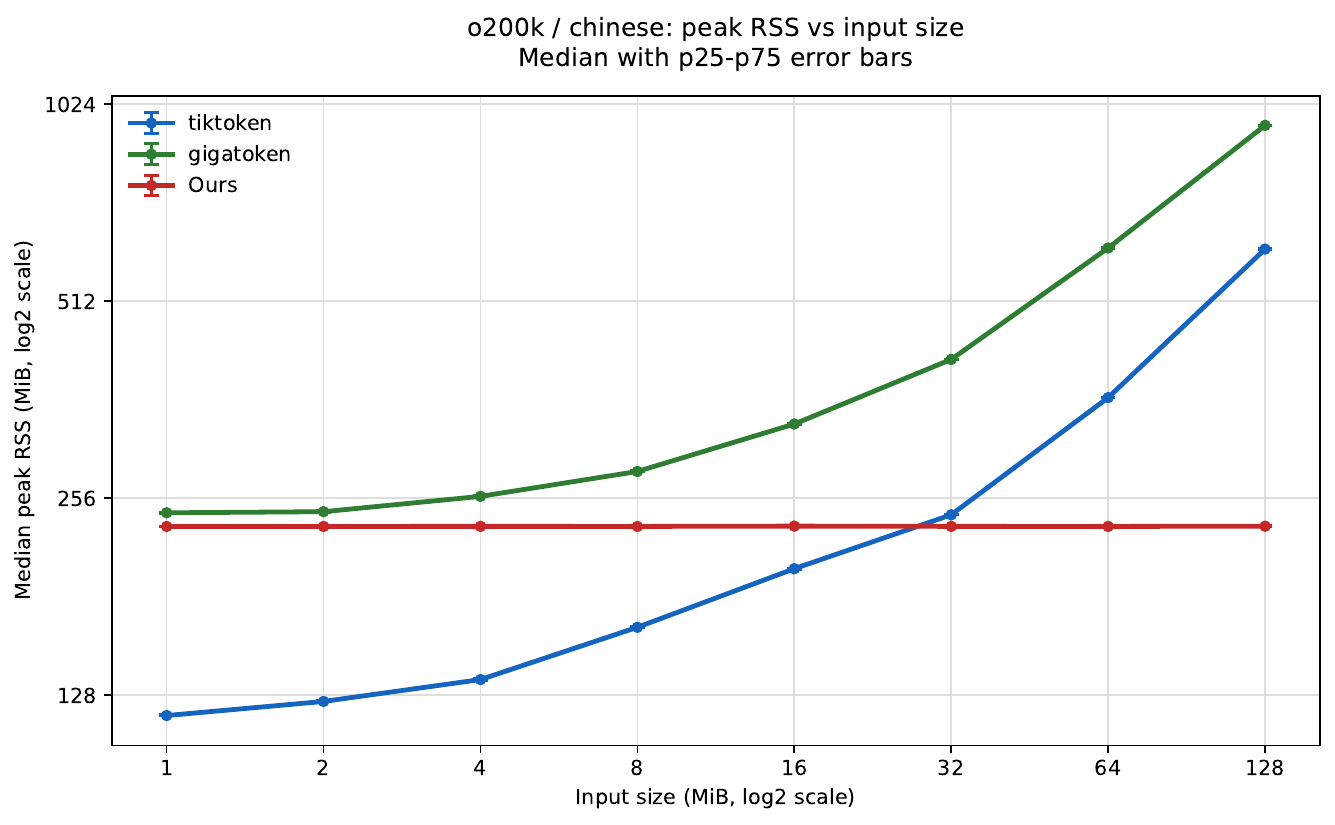}
        \caption{\python{o200k} peak RSS vs. input size}
    \end{subfigure}
    
    \caption{Peak RSS vs. input size on Chinese datasets. The blue line is \python{tiktoken}, the green line is \python{gigatoken}, while the red line is our algorithm.}
    \label{fig:chinese-peak-rss}
\end{figure}

\subsection{Throughput Comparison}

Finally, we measure throughput. 
The size of each dataset is approximately 128 MiB.
For every dataset and every \python{tiktoken} encoding, we compare the throughput between \python{tiktoken} and our tokenizer.
For every dataset and every \python{tokenizers} model, we compare the throughput between \python{tokenizers} and our tokenizer.
For both \python{tiktoken} and \python{tokenizers}, input is provided as a string in memory, while output is given as an array in memory.
For our algorithm, input is provided as a file path and output is streamed to a temporary file.
For the throughput benchmarking, output is streamed to a temporary file (\rust{tmpfs}) backed by memory instead of an actual file on disk.
In particular, the cost of reading and writing the input and output file is factored into the throughput of our algorithm.
Note that encoder construction time is not timed.
We measure performance of the encoding function itself.
\Cref{tbl:dfa-tiktoken-throughput} shows the throughput comparison on \python{tiktoken} encodings and \Cref{tbl:hf-tokenizers-throughput} shows the throughput comparison on \python{tokenizers} encodings.
Our algorithm significantly outperforms \python{tiktoken}, achieving throughput speedups of 2.4$\times$-7.0$\times$ depending on the encoding and dataset.
We achieve a more dramatic speedup over \python{tokenizers}, with speedup factors between 15.8$\times$-38.1$\times$.

\begin{table}[ht]
\caption{Throughput comparison of our tokenizer on \python{tiktoken} encodings measured in MiB/s. \python{tiktoken} is provided with a string as input and gives an array as output, both in memory. Our algorithm is provided with a file path as input and writes output to file. We report the median throughput over 10 runs for each row with upper and lower quartiles in the parentheses.}
\label{tbl:dfa-tiktoken-throughput}
\begin{center}
\begin{tabular}{|l|l|c|c|c|}
\hline
{\bf Encoding} & {\bf Dataset} & \python{tiktoken} & {\bf Ours} & Speedup \\
\hline
\python{r50k} & Github & 11.3 (11.3-11.3) & {\bf 43.6} (43.5-43.7) & 3.9$\times$ \\
\python{r50k} & English & 17.8 (17.7-17.8) & {\bf 47.7} (47.4-48.3) & 2.7$\times$ \\
\python{r50k} & Chinese & 15.0 (14.9-15.0) & {\bf 38.0} (38.0-38.1) & 2.5$\times$ \\
\hline
\python{p50k} & Github & 10.8 (10.8-10.8) & {\bf 45.0} (44.7-45.2) & 4.2$\times$ \\
\python{p50k} & English & 18.0 (17.9-18.0) & {\bf 49.2} (48.8-49.4) & 2.7$\times$ \\
\python{p50k} & Chinese & 15.0 (14.9-15.0) & {\bf 36.6} (36.5-36.7) & 2.4$\times$ \\
\hline
\python{cl100k} & Github & 11.1 (11.1-11.1) & {\bf 48.4} (48.1-48.6) & 4.4$\times$ \\
\python{cl100k} & English & 15.5 (15.4-15.6) & {\bf 45.3} (44.3-46.4) & 2.9$\times$ \\
\python{cl100k} & Chinese & 13.6 (13.6-13.6) & {\bf 37.6} (37.5-37.6) & 2.8$\times$ \\
\hline
\python{o200k} & Github & 7.8 (7.7-7.8) & {\bf 54.8} (54.7-55.0) & 7.0$\times$ \\
\python{o200k} & English & 8.7 (8.7-8.7) & {\bf 41.5} (41.3-41.7) & 4.8$\times$ \\
\python{o200k} & Chinese & 9.9 (9.9-9.9) & {\bf 36.5} (36.2-36.9) & 3.7$\times$ \\
\hline
\end{tabular}
\end{center}
\end{table}

\begin{table}[ht]
\caption{Throughput comparison of our tokenizer on \python{tokenizers} models measured in MiB/s. \python{tokenizers} is provided with a string as input and gives an array as output, both in memory. Our algorithm is provided with a file path as input and writes output to file. We report the median throughput over 10 runs for each row with upper and lower quartiles in the parentheses.}
\label{tbl:hf-tokenizers-throughput}
\begin{center}
\begin{tabular}{|l|l|c|c|c|}
\hline
{\bf Model} & {\bf Dataset} & \python{tokenizers} & {\bf Ours} & Speedup \\
\hline
\python{GPT2} & Github & 1.2 (1.2-1.2) & {\bf 45.7} (45.5-45.8) & 38.1$\times$ \\
\python{GPT2} & English & 1.4 (1.4-1.4) & {\bf 50.4} (50.0-51.3) & 36.0$\times$ \\
\python{GPT2} & Chinese & 1.2 (1.2-1.2) & {\bf 38.6} (38.4-38.6) & 32.2$\times$ \\
\hline
\python{RoBERTa} & Github & 1.1 (1.1-1.2) & {\bf 40.7} (37.5-41.0) & 37.0$\times$ \\
\python{RoBERTa} & English & 1.3 (1.3-1.3) & {\bf 46.2} (45.7-46.4) & 35.5$\times$ \\
\python{RoBERTa} & Chinese & 1.2 (1.2-1.2) & {\bf 35.0} (34.7-35.2) & 29.2$\times$ \\
\hline
\python{StarCoder} & Github & 1.4 (1.4-1.4) & {\bf 51.3} (50.8-51.5) & 36.6$\times$ \\
\python{StarCoder} & English & 1.3 (1.3-1.4) & {\bf 35.8} (35.7-35.8) & 27.5$\times$ \\
\python{StarCoder} & Chinese & 1.5 (1.4-1.5) & {\bf 32.1} (32.0-32.2) & 21.4$\times$ \\
\hline
\python{Mistral} & Github & 1.7 (1.6-1.7) & {\bf 55.8} (55.4-56.3) & 32.8$\times$ \\
\python{Mistral} & English & 1.7 (1.7-1.7) & {\bf 47.4} (47.1-47.6) & 27.9$\times$ \\
\python{Mistral} & Chinese & 1.7 (1.7-1.7) & {\bf 39.3} (39.3-39.4) & 23.1$\times$ \\
\hline
\python{Llama4} & Github & 1.8 (1.8-1.8) & {\bf 56.3} (56.2-56.7) & 31.3$\times$ \\
\python{Llama4} & English & 1.7 (1.7-1.7) & {\bf 41.1} (40.9-41.4) & 24.2$\times$ \\
\python{Llama4} & Chinese & 1.7 (1.7-1.7) & {\bf 33.7} (33.5-33.9) & 19.8$\times$ \\
\hline
\python{Qwen-3} & Github & 1.6 (1.6-1.6) & {\bf 30.8} (30.6-31.0) & 19.2$\times$ \\
\python{Qwen-3} & English & 1.5 (1.5-1.5) & {\bf 29.6} (29.5-30.1) & 19.7$\times$ \\
\python{Qwen-3} & Chinese & 1.7 (1.7-1.7) & {\bf 26.8} (26.6-26.9) & 15.8$\times$ \\
\hline
\python{gpt-oss} & Github & 1.9 (1.9-1.9) & {\bf 56.6} (56.4-57.1) & 29.8$\times$ \\
\python{gpt-oss} & English & 1.8 (1.8-1.8) & {\bf 44.4} (44.2-44.7) & 24.7$\times$ \\
\python{gpt-oss} & Chinese & 1.9 (1.9-1.9) & {\bf 36.8} (36.6-37.0) & 19.4$\times$ \\
\hline
\end{tabular}
\end{center}
\end{table}

\subsubsection{Comparison with \python{gigatoken}}
\label{sec:gigatoken}
Recently, the release of the \python{gigatoken} tokenizer presents a parallelized tokenizer that achieves significant throughput speedups over previous tokenizers \citep{roed2026gigatoken}.
At a high level, \python{gigatoken} uses the following ideas to dramatically speed up tokenization:
\begin{enumerate}
    \item Parallelism - \citet{roed2026gigatoken} attempts to process a batch of files or a single file in parallel. 
    In order to break up a single file, one can attempt certain heuristics to split up a large file based on known priors for pretokenization. 
    For example, examining the \python{tiktoken} patterns, if we find a space between two words (a fairly common occurence), we can safely split a file at this space.
    \item Optimized Pretokenization - Using SIMD and many other tricks, \citet{roed2026gigatoken} drastically reduces the workload of pretokenization, which is typically deferred to a regex engine. 
    \item Caching - In order to avoid repeatedly merging the same pretoken, \citet{roed2026gigatoken} maintains a cached mapping from pretokens to tokens. 
    While some form of caching is already implemented in \python{tiktoken}, \citet{roed2026gigatoken} uses a heavily optimized cache to improve performance.
\end{enumerate}

To make a fair comparison with \python{gigatoken} in the streaming and online settings which are inherently sequential, we compare with the serial implementation of \python{gigatoken} (the package already provides such an option by setting \python{parallel=False}). 
Furthermore, we note that given the above tricks to split a single file or a batch of files, one can easily obtain a parallel algorithm by applying an optimized serial tokenizer to each part of the resulting split.
Regarding the optimized pretokenization and caching, we observe that many optimizations remain compatible with an online framework, particularly composition: one can replace our DFA-based pretokenization algorithm with the SIMD-optimized pretokenization, and one can incorporate an optimized pretoken caching policy into any online BPE algorithm.
We note that our original algorithm, identically to \python{tiktoken}, maintains a cache of all pretokens that map directly to tokens, and thus maintains a cache of approximately size $O(\merges \longest)$.
The cache of \python{gigatoken} instead caches all observed pretokens, thus resulting in a cache that grows with input size in the worst case.
In particular, after incorporating the more aggressive caching policy which caches all pretokens, online tokenization will no longer maintain a constant memory footprint independent of input size.

In \Cref{tbl:gigatoken-throughput-array}, we compare the throughput of \python{tiktoken}, \python{gigatoken}, and two implementations of our algorithm with optimized pretokenization and caching from \citet{roed2026gigatoken}.
Both our algorithm and \python{gigatoken} are provided with a file path as input and supports two output types: an array in memory, and a written output file (for the benchmark, we use a temporary file backed by RAM to avoid blocking on disk).
Since \python{gigatoken} does not support streaming, we write the array to file after output.
On array output, we attain similar performance to \python{gigatoken}, producing lower throughput on English and code datasets, while producing higher throughput on Chinese.
The more aggressive caching policy can use increasing memory as additional distinct pretokens are encountered. 
Our DFA-based implementation voids this growing pretoken cache, but its input buffer can still grow with an unfinished pretoken (as discussed earlier).

\begin{table}[ht]
\caption{Throughput comparison of optimized tokenizer on \python{tiktoken} encodings (MiB/s). \python{gigatoken} and our algorithm are provided with a file path as input. The array (A) column reports throughput when output is collected as array in memory, the stream (S) column reports when output is written to a temporary file. We report the median of 10 runs (quartiles omitted for space).}
\label{tbl:gigatoken-throughput-array}
\begin{center}
\begin{tabular}{|l|l|c|c|c|c|}
\hline
{\bf Encoding} & {\bf Dataset} & \python{gigatoken} (A) & \python{gigatoken} (S) & {\bf Ours} (A) & {\bf Ours} (S) \\
\hline
\python{r50k} & Github & 178.3 & 26.9 & {\bf 185.5} & 117.5 \\
\python{r50k} & English & {\bf 331.1} & 44.6 & 298.7 & 181.9 \\
\python{r50k} & Chinese & 39.9 & 13.1 & {\bf 45.3} & 37.8 \\
\hline
\python{p50k} & Github & 188.7 & 32.2 & {\bf 191.5} & 126.8 \\
\python{p50k} & English & {\bf 329.2} & 44.6 & 297.6 & 180.0 \\
\python{p50k} & Chinese & 39.8 & 13.0 & {\bf 45.9} & 37.8 \\
\hline
\python{cl100k} & Github & {\bf 172.2} & 40.4 & 159.7 & 119.7 \\
\python{cl100k} & English & {\bf 278.3} & 44.2 & 248.3 & 158.4 \\
\python{cl100k} & Chinese & 27.8 & 13.5 & {\bf 43.7} & 36.5 \\
\hline
\python{o200k} & Github & {\bf 253.4} & 43.3 & 236.6 & 155.9 \\
\python{o200k} & English & {\bf 240.9} & 43.6 & 216.0 & 140.8 \\
\python{o200k} & Chinese & 24.0 & 15.0 & {\bf 41.2} & 36.4 \\
\hline
\end{tabular}
\end{center}
\end{table}

%% file: appendix_omitted_proofs.tex
\section{Omitted Proofs}
\label{app:omitted}

We prove \Cref{lem:batch-bpe}.

\begin{proof}[Proof of \Cref{lem:batch-bpe}]
    We use the incremental BPE algorithm of \citet{jiang2025incrementalbpe}.
    Their algorithm computes the last token of every prefix, whereas we require the first token of every suffix.  
    The algorithm and proof are almost identical, with some minor modifications to handle left-to-right tie-breaking in the computation of first token $\phi$ rather than last token $\lambda$.

    We recap the algorithm of \citet{jiang2025incrementalbpe}.
    In preprocessing, they compute the following data structures:
    \begin{enumerate}
        \item A successor forest $\calF$ where nodes are tokens $\tau$ and edges $(\tau, \suc(\tau))$ where $\tau$ is produced by the merge rule $(\pre(\tau), \suc(\tau)) \mapsto \tau$.
        We call $\pre(\tau), \suc(\tau)$ the predecessor and successor of $\tau$.
        We replace this with a predecessor forest $\calF$ where the nodes are tokens and edges $(\tau, \pre(\tau))$.
        This requires $O(\merges)$ time.

        \item For each token $\tau$, compute the suffix-successor tree of $\tau$, which is the induced subgraph of $\calF$ containing tokens $\mu$ where $\raw(\mu)$ is a suffix of $\raw(\tau)$.
        On each of these trees, compute a centroid decomposition (using e.g. \citet{della2019new}).
        Instead, we replace this with a prefix-predecessor tree, i.e. the induced subgraph of the predecessor forest $\calF$ consisting of tokens $\mu$ where $\raw(\mu) \preceq \raw(\tau)$.
        There are $\merges$ tokens, each with at most $\longest$ prefixes, so this requires $O(\merges \longest)$ time.

        \item \citet{jiang2025incrementalbpe} use an Aho--Corasick automaton to report the longest suffix of $x_{\leq i}$ that is a token.
        Instead, we scan right-to-left using an automaton on reversed token strings to report the longest prefix of $x_{\geq i}$ that is a token.
        This requires $O(\merges\longest)$ preprocessing.
    \end{enumerate}

    We then process the input string $x$ right-to-left.
    Assume that $\phi(x_{\geq j})$ is computed for $j > i$.
    We have an analogue of \citet{jiang2025incrementalbpe} Prefix Last Token condition (see Definition 4.1), which we call the Suffix First Token condition.
    That is, $\tau$ satisfies the Suffix First Token condition if it is an atomic token (i.e. $|\raw(\tau)| = 1$) or satisfies the following:
    \begin{enumerate}
        \item (Reachability) $\phi(x_{\geq i + |\pre(\tau)|})$ is in the subtree of the predecessor forest rooted at $\suc(\tau)$.
        This ensures that $\suc(\tau)$ is formed in the suffix $x_{\geq i + |\pre(\tau)|}$.
        
        \item (Priority Dominance) If $\phi(x_{\geq i + |\pre(\tau)|}) \neq \suc(\tau)$, let $u$ be the unique child of $\suc(\tau)$ that lies on the path to $\phi(x_{\geq i + |\pre(\tau)|})$.
        Then, the rule producing $\tau$ has priority at least as high as the rule producing $u$.
        This ensures that $\tau$ is produced before $\suc(\tau)$ is consumed.
        Note that here we require priority to be at least as high, instead of strictly higher as in \citet{jiang2025incrementalbpe}, due to left-to-right tie-breaking.
    \end{enumerate}

    Let $\rho$ be the longest token such that $\raw(\rho) \preceq x_{\geq i}$.
    Then $\raw(\phi(x_{\geq i})) \preceq \raw(\rho)$.
    Analogously to \citet{jiang2025incrementalbpe}, one can show that on the prefix-predecessor tree of $\rho$, a vertex is on the path from $x_{i}$ to $\phi(x_{\geq i})$ iff it satisfies the Suffix First Token condition, which can be checked in $O(1)$ time. 
    The idea then is to start from the root $x_{i}$ and follow the path of children satisfying the Suffix First Token (SFT) condition.
    Since we have constructed a centroid decomposition, we begin from the centroid.
    We first test the centroid's SFT condition.
    If it is valid, we search for a child satisfying the condition.
    We identify the valid child by binary search, with children ordered by the left endpoints of their disjoint valid intervals, as in \citet{jiang2025incrementalbpe}, with predecessors and successors interchanged.
    Here children refer to the original prefix-predecessor tree.
    If the centroid is valid but has no valid child, return it.
    Otherwise, either the centroid has a valid child, in which case we remove the centroid and continue at the centroid of the remaining component containing the valid child, or the centroid is invalid, in which case we remove it and continue at the centroid of the remaining component containing its original-tree parent, as in \citet{jiang2025incrementalbpe}.
    This proceeds for at most $O(\log \longest)$ steps, and searching for the right child requires $O(\log \longest)$ time, leading to $O(\log^2 \longest)$ time.
\end{proof}

We provide the reference and simplified DFA illustrations for $\rfk$, $\clok$, and $\otk$.

\input{tikz/r50K_ref}
   
\input{tikz/r50k_simp}

\input{tikz/cl100k}
    
\input{tikz/o200k}

\input{tikz/o200k_simp}

%% file: tikz/r50K_ref.tex
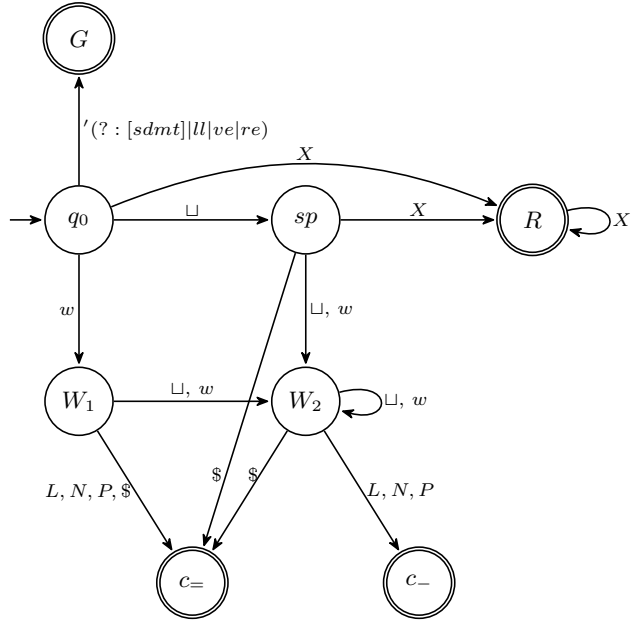
\begin{figure}[t]\centering
\begin{tikzpicture}[
  >={Stealth[round]}, shorten >=1pt, semithick,
  every state/.style={minimum size=9mm, inner sep=1pt, font=\small},
  el/.style={font=\scriptsize, inner sep=1.2pt}]

  \node[state, initial, initial text={}] (q0) at (0,0)     {$q_0$};
  \node[state, accepting]                (G)  at (0,2.4)   {$G$};
  \node[state]                           (sp) at (3,0)     {$sp$};
  \node[state, accepting]                (R)  at (6,0)     {$R$};
  \node[state]                           (W1) at (0,-2.4)  {$W_1$};
  \node[state]                           (W2) at (3,-2.4)  {$W_2$};
  \node[state, accepting]                (ceq) at (1.5,-4.8){$c_{=}$};
  \node[state, accepting]                (cmin) at(4.5,-4.8){$c_{-}$};

  \path[->]
    (q0) edge[el] node[above] {$\sqcup$} (sp)
    (q0) edge[el] node[right] {${}'(?:[sdmt]|ll|ve|re)$} (G)
    (q0) edge[el, bend left=22] node[above] {$X$} (R)
    (sp) edge[el] node[above] {$X$} (R)
    (R)  edge[el, loop right] node {$X$} ()
    (q0) edge[el] node[left] {$w$} (W1)
    (sp) edge[el] node[right] {$\sqcup,\,w$} (W2)
    (W1) edge[el] node[above] {$\sqcup,\,w$} (W2)
    (W2) edge[el, loop right] node {$\sqcup,\,w$} ()
    (W1) edge[el] node[left, align=center] {$L,N,P,\$$} (ceq)
    (W2) edge[el] node[right] {$L,N,P$} (cmin)
    (W2) edge[el] node[pos=.45, above] {$\$$} (ceq)
    (sp) edge[el] node[pos=.75, left] {$\$$} (ceq);
\end{tikzpicture}
\caption{R50K reference pattern as a delayed-accept DFA. $X$ denotes non-whitespace characters as    in \Cref{fig:r50k-simp};
$\sqcup$ = space, $w$ = other whitespace, $\$$ = end of input. 
The non-whitespace part ($q_0,sp,R,G$) is an ordinary DFA (boundary unconditional at run/contraction
end); the whitespace part ($W_1,W_2,c_{=},c_{-}$) is a delayed DFA.
$c_{=}$ includes the last ws char; $c_{-}$ trims it (imposed by the rule $\backslash\mathrm s^{+}(?!\backslash\mathrm S)$).
Peeling fires only from $W_2$, i.e. when the run of whitespace is at least $2$ characters long.}
\label{fig:r50k-ref}
\end{figure}

%% file: tikz/r50k_simp.tex
\begin{figure}[t]\centering
\begin{tikzpicture}[
  >={Stealth[round]}, shorten >=1pt, semithick, on grid,
  every state/.style={minimum size=9mm, inner sep=1pt, font=\small},
  el/.style={font=\scriptsize, inner sep=1.5pt}]
  \node[state, initial, initial text={}] (q0) {$q_0$};
  \node[state, accepting] (G)  [above=20mm of q0] {$G$};
  \node[state]            (sp) [right=26mm of q0] {$sp$};
  \node[state, accepting] (R)  [right=52mm of q0] {$R$};
  \node[state, accepting] (W)  [below=20mm of q0] {$W$};
  \path[->]
    (q0) edge [el] node[above] {$\sqcup$} (sp)
    (q0) edge [el] node[left,pos=.55] {${}'(?:[sdmt]|ll|ve|re)$} (G)
    (q0) edge [el] node[left] {$w$} (W)
    (q0) edge [el, bend left=24] node[above] {$X$} (R)
    (sp) edge [el] node[above] {$X$} (R)
    (sp) edge [el, bend left=14] node[pos=.35,left] {$\sqcup,\,w$} (W)
    (R)  edge [el, loop right]  node {$X$} ()
    (W)  edge [el, loop below]  node {$\sqcup,\,w$} ();
\end{tikzpicture}
\caption{R50K simplified pattern.
$X$ denotes non-whitespace characters such as letters, digits, and symbols.
In each case, the run continues on the \emph{same} class $X$;
$\sqcup$ = literal space, $w$ = other whitespace $(\backslash\mathrm s\setminus\sqcup)$.
$W$ accepts the maximal $\backslash\mathrm s^{+}$ run \emph{unconditionally} (no lookahead).}
\label{fig:r50k-simp}
\end{figure}
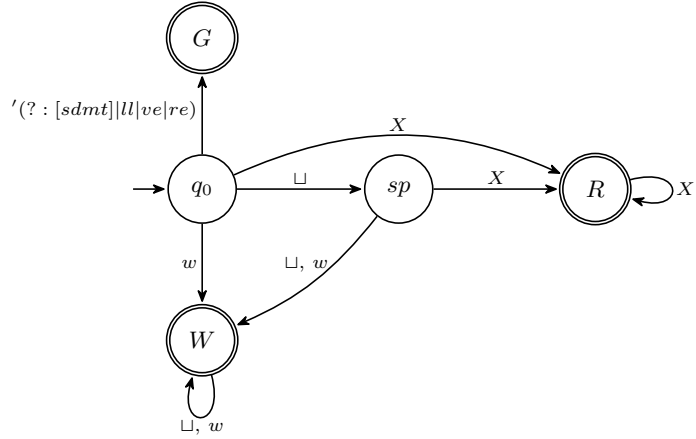

%% file: tikz/cl100k.tex
\begin{figure}[t]\centering
\begin{tikzpicture}[
  >={Stealth[round]}, shorten >=1pt, semithick,
  every state/.style={minimum size=8mm, inner sep=1pt, font=\small},
  el/.style={font=\scriptsize, inner sep=1.2pt}, add/.style={dashed}]
  \node[state, initial, initial text={}] (q0) at (0,0)      {$q_0$};
  \node[state]                           (sp) at (3, 0)     {$\mathrm{sp}$};
  \node[state, accepting]                (G)  at (0,2.6)    {$G$};
  \node[state, accepting]                (Wd) at (6,-0.9)  {$\mathrm{L}$};
  \node[state, accepting]                (D1) at (3,3)  {$D_1$};
  \node[state, accepting]                (D2) at (4.8,3)  {$D_2$};
  \node[state, accepting]                (D3) at (6.6,3)  {$D_3$};
  \node[state, accepting]                (Pu) at (6,0.9) {$\mathrm{Pu}$};
  \node[state, accepting]                (PN) at (8,0.9) {$\mathrm{PN}$};
  \node[state]                           (Wh) at (3,-2.7)   {$\mathrm{Wh}$};
  \node[state]                           (Wh2) at (1.5,-5)   {$\mathrm{Wh}_2$};
  \node[state]                           (Wnl) at (0,-2.7) {$\mathrm{Wnl}$};
  \node[state, accepting]                (ceq) at (-2,-4) {$c_{=}$};
  \node[state, accepting]                (cmin) at (5,-4) {$c_{-}$};
  
  \path[->]
    (q0) edge[el] node[left] {${}'(?:[sdmt]|ll|ve|re)$} (G)
    (q0) edge[el] node[above] {$\sqcup$} (sp)
    (sp) edge[el] node[above=.1,pos=.4] {$\pL$} (Wd)
    (q0) edge[el, bend right] node[above=.1,pos=.6] {$\pL$} (Wd)
    (Wd) edge[el, loop right] node {$\pL$} ()
    (q0) edge[el] node[above=.4,pos=.4] {$\pN$} (D1)
    (D1) edge[el] node[above] {$\pN$} (D2)
    (D2) edge[el] node[above] {$\pN$} (D3)
    (sp) edge[el] node[above=.3,pos=.4] {$\punct$} (Pu)
    (Pu) edge[el, loop above] node {$\punct$} ()
    (Pu) edge[el] node[above] {$\Rclass$} (PN)
    (PN) edge[el, loop right] node {$\Rclass$} ()
    (q0) edge[el] node[below=.5, pos=.3] {$\Hclass$} (Wh)
    (Wh) edge[el] node[below] {$\Rclass$} (Wnl)
    (Wh) edge[el] node[below=.5] {$\pL$} (Wd)
    (Wh) edge[el] node[right=.1] {$\Hclass$} (Wh2)
    (q0) edge[el] node[left,pos=.25] {$\Rclass$} (Wnl)
    (Wnl) edge[loop left] node {$\Rclass$} (Wnl)
    (Wnl) edge[el] node[right=.1] {$\Hclass$} (Wh2)
    (Wnl) edge[el] node[below=.3, pos=.2] {$\backslash \mathrm{S}, \$$} (ceq)
    (Wh2) edge[el] node[below=.1] {$\backslash \mathrm{S}$} (cmin)
    (Wh2) edge[el, bend left] node[left=.1] {$\Rclass$} (Wnl)
    (Wh2) edge[loop below] node[below] {$\Hclass$} ()
    (Wh2) edge[el] node[below] {$\$$} (ceq);

\end{tikzpicture}
\caption{CL100K reference pattern. 
$G$ denotes the DFA to compute the first rule. 
Next, $L$ denotes the DFA to compute runs of letters with leading horizontal whitespace or punctuation (e.g. $\leadc$).
$D_i$ denotes the DFA to compute runs of $\leq 3$ digits. 
$Pu$ and $PN$ denote the DFA to compute runs of symbols followed by newlines.
Priority: contraction, word, digit, punct, ws.
$\mathrm{Wh}$ and $\mathrm{Wh}_{2}$ denote the first and subsequent horizontal whitespace respectively.
$\mathrm{Wnl}$ denotes vertical whitespace/newline. 
From $\mathrm{Wnl}$ we always keep the last whitespace.
From $\mathrm{Wh}_{2}$, we trim the last character when a non-whitespace is encountered.}
\label{fig:cl100k-ref}
\end{figure}

%% file: tikz/o200k.tex
\begin{figure}[t]\centering
\begin{tikzpicture}[
  >={Stealth[round]}, shorten >=1pt, semithick,
  every state/.style={minimum size=8mm, inner sep=1pt, font=\small},
  el/.style={font=\scriptsize, inner sep=1.2pt}, add/.style={dashed}]
  \node[state, initial, initial text={}] (q0) at (0,0)      {$q_0$};
  \node[state, accepting]                (Wu) at (0,2.6)  {$\mathrm{Wu}$};
  \node[state, accepting]                (Wl) at (5.4,2.6)  {$\mathrm{Wl}$};
  \node[state, accepting]                (D1) at (4.0,0.9)  {$D_1$};
  \node[state, accepting]                (D2) at (5.7,0.9)  {$D_2$};
  \node[state, accepting]                (D3) at (7.4,0.9)  {$D_3$};
  \node[state, accepting]                (Pu) at (4.0,-0.9) {$\mathrm{Pu}$};
  \node[state, accepting]                (PN) at (6,-0.9) {$\mathrm{PN}$};
  \node[state]                           (Wh) at (3.2,-2.7)   {$\mathrm{Wh}$};
  \node[state, accepting]                (Wnl)at (0,-2.7) {$\mathrm{Wnl}$};
  \node[state, accepting]                (con)at (3, 5)     {$G$};
  \node[state]                           (Who) at (0,-5)   {$\mathrm{Wh}_o$};
  \node[state]                           (Wh2) at (2,-5)   {$\mathrm{Wh}_2$};
  \node[state, accepting]                (ceq) at (6,-3) {$c_{=}$};
  \node[state, accepting]                (cmin) at (4,-5) {$c_{-}$};
  \path[->]
    (q0) edge[el] node[left,pos=.55] {$\leadc?\,U$} (Wu)
    (q0) edge[el, bend left=5] node[left=.4,pos=.7] {$\leadc?\,l$} (Wl)
    (Wu) edge[el, loop above] node {$U$} ()
    (Wu) edge[el] node[above] {$l$} (Wl)
    (Wu) edge[el] node[left=.1, pos=.7] {$(?i{:}'s|\dots)?$} (con)
    (Wl) edge[el, loop above] node {$l$} ()
    (Wl) edge[el] node[right=.1, pos=.7] {$(?i{:}'s|\dots)?$} (con)
    (q0) edge[el] node[above=.1,pos=.55] {$\pN$} (D1)
    (D1) edge[el] node[above] {$\pN$} (D2)
    (D2) edge[el] node[above] {$\pN$} (D3)
    (q0) edge[el] node[above=.2,pos=.6] {$\sqcup?\,\punct$} (Pu)
    (Pu) edge[el, loop below] node {$\punct$} ()
    (Pu) edge[el] node[above] {$[\backslash\mathrm r\backslash\mathrm n/]$} (PN)
    (PN) edge[el, loop right] node {$[\backslash\mathrm r\backslash\mathrm n/]$} ()
    (q0) edge[el] node[left=.2] {$\Hclass$} (Wh)
    (Wh) edge[el] node[right] {$\Hclass$} (Wh2)
    (Wh) edge[el] node[below] {$\Rclass$} (Wnl)
    (Wh) edge[el] node[below=.1, pos=.4] {$[\hat{}\backslash\mathrm{s}Ul]$} (ceq)
    (q0) edge[el] node[left,pos=.4] {$\Rclass$} (Wnl)
    (Wnl) edge[el, loop left] node {$\Rclass$} ()
    (Wnl) edge[el] node[right, pos=.7] {$\Hclass$} (Who)
    (Who) edge[el, bend left] node[left=.1] {$\Rclass$} (Wnl)
    (Who) edge[loop below] node[below] {$\Hclass$} ()
    (Wh2) edge[el] node[above] {$\backslash \mathrm{S}$} (cmin)
    (Wh2) edge[loop below] node[below] {$\Hclass$} ()
    (Wh2) edge[el] node[right=.1] {$\Rclass$} (Wnl);
\end{tikzpicture}
\caption{O200K reference pattern. Words are case split, with upper case characters $U=[\backslash\mathrm p\{Lu,Lt,Lm,Lo,M\}]$ first, followed by lower case characters $l=[\backslash\mathrm p\{Ll,Lm,Lo,M\}]$ with an optional contraction suffix $G$.
Digits have run length $\leq 3$.
Symbols can have a leading space, followed by an arbitrary number of newlines.
Horizontal whitespace followed by letters reaches the first two cases, while followed by any other non-whitespace character immediately outputs the single space.
Remaining runs of whitespace either must end with a newline or leave a trailing space before the next whitespace character.}
\label{fig:o200k-ref}
\end{figure}
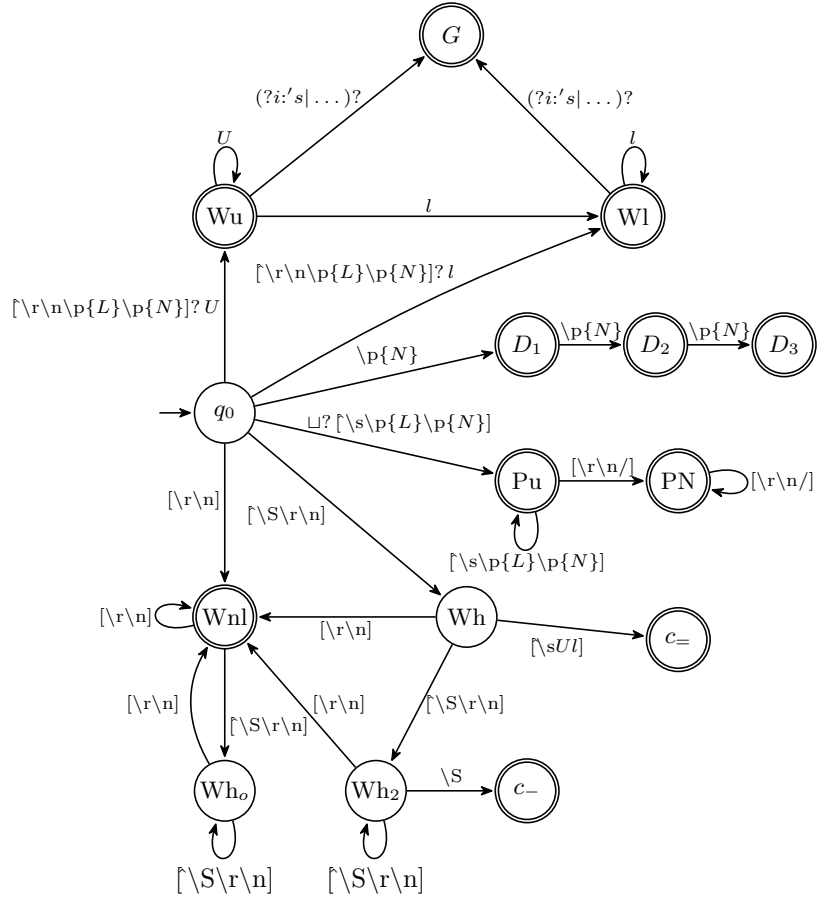

%% file: tikz/o200k_simp.tex
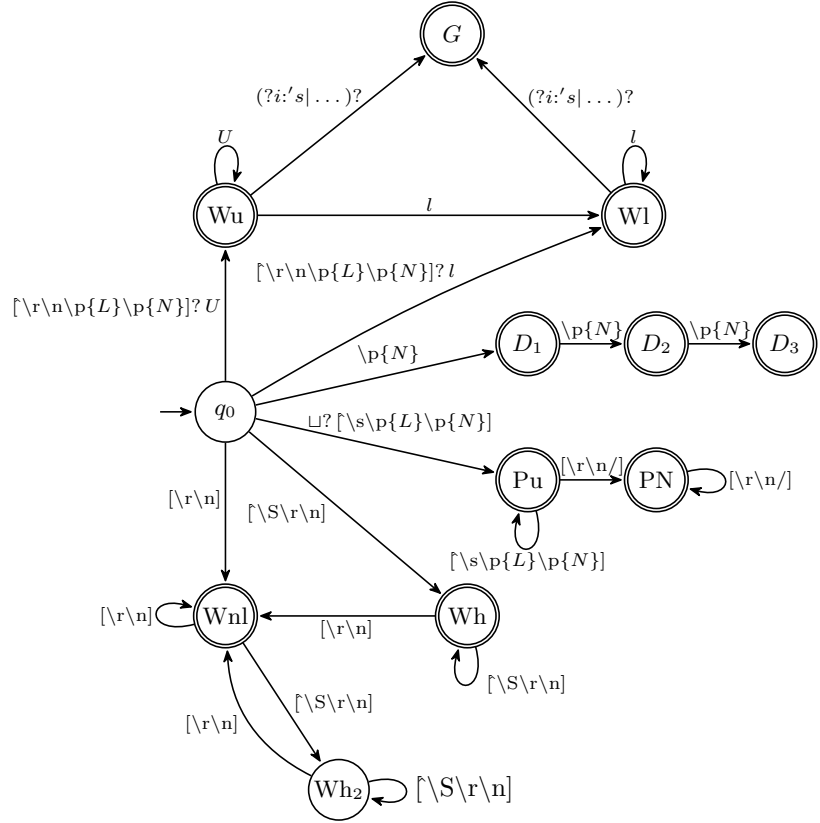
\begin{figure}[t]\centering
\begin{tikzpicture}[
  >={Stealth[round]}, shorten >=1pt, semithick,
  every state/.style={minimum size=8mm, inner sep=1pt, font=\small},
  el/.style={font=\scriptsize, inner sep=1.2pt}, add/.style={dashed}]
  \node[state, initial, initial text={}] (q0) at (0,0)      {$q_0$};
  \node[state, accepting]                (Wu) at (0,2.6)  {$\mathrm{Wu}$};
  \node[state, accepting]                (Wl) at (5.4,2.6)  {$\mathrm{Wl}$};
  \node[state, accepting]                (D1) at (4.0,0.9)  {$D_1$};
  \node[state, accepting]                (D2) at (5.7,0.9)  {$D_2$};
  \node[state, accepting]                (D3) at (7.4,0.9)  {$D_3$};
  \node[state, accepting]                (Pu) at (4.0,-0.9) {$\mathrm{Pu}$};
  \node[state, accepting]                (PN) at (5.7,-0.9) {$\mathrm{PN}$};
  \node[state, accepting]                (Wh) at (3.2,-2.7)   {$\mathrm{Wh}$};
  \node[state, accepting]                (Wnl)at (0,-2.7) {$\mathrm{Wnl}$};
  \node[state, accepting]                (con)at (3, 5)     {$G$};
  \node[state]                           (Wh2) at (1.5,-5)   {$\mathrm{Wh}_2$};
  \path[->]
    (q0) edge[el] node[left,pos=.55] {$\leadc?\,U$} (Wu)
    (q0) edge[el, bend left=5] node[left=.4,pos=.7] {$\leadc?\,l$} (Wl)
    (Wu) edge[el, loop above] node {$U$} ()
    (Wu) edge[el] node[above] {$l$} (Wl)
    (Wu) edge[el] node[left=.1, pos=.7] {$(?i{:}'s|\dots)?$} (con)
    (Wl) edge[el, loop above] node {$l$} ()
    (Wl) edge[el] node[right=.1, pos=.7] {$(?i{:}'s|\dots)?$} (con)
    (q0) edge[el] node[above=.1,pos=.55] {$\pN$} (D1)
    (D1) edge[el] node[above] {$\pN$} (D2)
    (D2) edge[el] node[above] {$\pN$} (D3)
    (q0) edge[el] node[above=.2,pos=.6] {$\sqcup?\,\punct$} (Pu)
    (Pu) edge[el, loop below] node {$\punct$} ()
    (Pu) edge[el] node[above] {$[\backslash\mathrm r\backslash\mathrm n/]$} (PN)
    (PN) edge[el, loop right] node {$[\backslash\mathrm r\backslash\mathrm n/]$} ()
    (q0) edge[el] node[left=.2] {$\Hclass$} (Wh)
    (Wh) edge[el, loop below] node[right=.2] {$\Hclass$}
    (Wh) edge[el] node[below] {$\Rclass$} (Wnl)
    (q0) edge[el] node[left,pos=.4] {$\Rclass$} (Wnl)
    (Wnl) edge[el, loop left] node {$\Rclass$} ()
    (Wnl) edge[el] node[right=.1] {$\Hclass$} (Wh2)
    (Wh2) edge[el, bend left] node[left=.1] {$\Rclass$} (Wnl)
    (Wh2) edge[loop right] node[right] {$\Hclass$} ();
\end{tikzpicture}
\caption{O200K simplified pattern. Most of the DFA (except whitespace handling) is identical to \Cref{fig:o200k-ref}. For whitespace, any run of whitespace that contains vertical whitespace/newline must end on vertical whitespace/newline. Otherwise, a run that contains only horizontal whitespace is matched as a single run.}
\label{fig:o200k-simp}
\end{figure}